\documentclass[11pt]{article}
\usepackage{float}
\usepackage{amsmath,url}
\usepackage{amsthm}
\usepackage{amssymb}
\usepackage{hyperref}
\usepackage{bbm}
\usepackage{algorithm}
\usepackage{algorithmic}
\usepackage{epsfig}
\usepackage{multirow}
\usepackage{pdflscape}
\usepackage{verbatim}
\usepackage[english]{babel}
\usepackage{natbib}
\usepackage{graphicx}
\usepackage{caption}
\usepackage{scalerel}
\usepackage{diagbox}
\usepackage[usenames, dvipsnames]{color}
\definecolor{mypink1}{rgb}{0.858, 0.188, 0.478}
\definecolor{mypink2}{RGB}{219, 48, 122}
\definecolor{mypink3}{cmyk}{0, 0.7808, 0.4429, 0.1412}
\definecolor{mygray}{gray}{0.6}

\theoremstyle{plain}
\newtheorem{theorem}{Theorem}[section]
\newtheorem{lemma}[theorem]{Lemma}
\newtheorem{corollary}[theorem]{Corollary}
\newtheorem{proposition}[theorem]{Proposition}

\theoremstyle{definition}

\theoremstyle{remark}

\newcommand{\modt}{|{\rm T}|}
\newcommand{\modn}{|{\rm N}|}

\numberwithin{equation}{section}

\title{Outlier detection and a tail-adjusted boxplot based on \\ extreme value theory}
\author{S. Bhattacharya and J. Beirlant}

\author{Shrijita Bhattacharya \footnote{Department of Statistics and Probability, Michigan State University,  C419 Wells Hall,  619 Red Cedar Rd East Lansing, MI 48824, {\tt bhatta61@msu.edu}} and \quad Jan Beirlant\footnote{Department of Mathematics, KULeuven, Department of Mathematical Statistics and Actuarial Science, University of the Free State}}

\begin{document}

	\maketitle
	\begin{abstract}
		{\noindent Whether an extreme observation is an outlier or not, depends strongly on the corresponding tail behaviour of the underlying distribution. We develop an automatic, data-driven method to identify extreme tail behaviour that  deviates from the intermediate and central characteristics. This allows for detecting extreme outliers or  sets of extreme data that show less spread than the bulk of the data.  To this end we extend a testing method proposed in \cite{bhatt2019} for the specific case of heavy tailed models, to all max-domains of attraction. Consequently we propose  a tail-adjusted boxplot which yields a more accurate representation of possible outliers. Several examples and simulation results illustrate the finite sample behaviour of this approach.}
	\end{abstract}
	
	

	\section{Introduction}
\label{sec:intro}
The identification of outliers has become an important topic in several statistical methods and fields of application, such as  the general field of novelty detection, next to for instance climatology and cyber crime. It then is important to judge if an observation or a set of observations can be explained from  the model fitted to the bulk of the data or from a deviating model.   

The classical boxplot introduced by \cite
{tukey1977} is the prime classical tool indicating {\it potential} outliers, 	marking individual data with an asterisk or dot when it has a distance  larger than 1.5 interquartile range ($IQR$) from the corresponding quartile $Q_1$ or $Q_3$. In fact, when  the underlying distribution of the data is normal, $X \sim \mathcal{N}(\mu,\sigma^2)$, the probability for an extreme normal observation to be wrongly indicated as an outlier in a classical boxplot is very small: at the right hand side tail of the distribution for large enough sample sizes it is approximately given by $$\mathbb{P}(X > Q_3+ 1.5 IQR)=\mathbb{P}(X > [\mu+ 0.67\, \sigma ] + 1.5 [1.34 \, \sigma])=
\mathbb{P}(Z > 2.68) = 0.0037 .$$ 
However, for distributions with a tail heavier than the normal model  the probability for a data point to be indicated by a dot is growing. The probability $\mathbb{P}(X > Q_R+ 1.5 IQR)$ equals 0.05 for an exponential distribution, while for a lognormal distribution, i.e. $\log X \sim \mathcal{N}(\mu,\sigma^2)$, this probability raises up to 0.08 when $\sigma=1$ and 0.16 when $\sigma=2$. 
For the Pareto distribution  with distribution function $F(x)=1-x^{-\alpha}$ ($x>1$, $\alpha >0$), this probability mounts to 0.073 with $\alpha =4$ and 0.16  when $\alpha =1$. In such cases the boxplot grossly overestimates the number of {\it real} outliers. 
\\
At the other side of the spectrum, considering tails with a finite right endpoint, the probability of indicating an observation as a potential outlier is typically lower, for instance  0 for the uniform distribution on (0,1), since then real outliers can  well be situated below $Q_3+1.5 IQR$.
\\	From this we can conclude that the classical boxplot can be a misleading tool for outlier detection, while it is often erroneously  used for that purpose.	In \cite{hubert2008} the classical boxplot was adjusted for skewness using some robust skewness estimators. However, in general,  tail heaviness  cannot be determined on the basis of  a skewness measure. Here extreme value \\ methodology is used  to detect outliers making direct use of tail heaviness indicators. \\


\noindent
Extreme value theory (EVT) is a  natural methodology  to describe and estimate tail heaviness. EVT starts from considering the limit distribution   of the maximum  of sequences of independent and identically   observations $X_1, X_2, \cdots, X_n$:  
\begin{equation}
\label{e:gpd}
\mathbb{P}(a_n^{-1}(\max_{i=1, \cdots,n}X_i-b_n)\leq y) \rightarrow G_{\xi}(y)=\exp(-(1+\xi y)^{-1/\xi}),\hspace{3mm} 1+\xi y>0, \hspace{3mm} \mbox{as }
n \to \infty,
\end{equation}
for some sequences $(a_n>0)$ and $(b_n)$, where the generalized extreme value distributions (GEV) with distribution function (df) $G_{\xi}$ are the only possible limit distributions. The GEV family is parametrized by the extreme value index (EVI) $\xi \in \mathbb{R}$. In case $\xi=0$ the limit GEV is given by $G_0 (y)= \exp (-\exp (-y)),\; y \in \mathbb{R}$.
The EVI is a measure of the tail-heaviness of the distribution of $X$ with a larger value of $\xi$ implying a heavier tail described by the right tail function (RTF) $\bar{F}(x)=1-F (x) = \mathbb{P}(X > x)$. 
\\

\noindent In the specific case $\xi >0$ the distributions $F$ satisfying the limit result \eqref{e:gpd}, composing the Fr\'echet domain of attraction, are given by the set of Pareto-type distributions with RTF given by
\begin{equation}
\bar{F}(x)  = \mathbb{P}(X>x) = x^{-\frac{1}{\xi}} \, \ell (x),
\label{Patype}
\end{equation}
where $\ell$ is a slowly varying function at infinity,  i.e.
\begin{equation}
{\ell (ty) \over \ell (t)}  \to  1  \text{ as } t\to\infty, \text{ for every  } y>1.
\label{ell}
\end{equation}
The estimation of $\xi$ under \eqref{Patype} has received most attention starting with the Hill (1975) estimator
\begin{equation}
H_{k} = {1 \over k}\sum_{j=1}^k \log X_{n-j+1,n} - \log X_{n-k,n} = {1 \over k}\sum_{j=1}^k V_j,
\label{e:Hill}
\end{equation}
where $X_{1,n} \leq X_{2,n} \leq \ldots \leq X_{n,n}$ denote the ordered observations and   $V_j$  denote the weighted log-spacings
\begin{equation}
V_j = j \, \log \frac{X_{n-j+1,n}}{X_{n-j,n}},\; j=1,\ldots, n-1. 
\label{def:V}
\end{equation}

\vspace{0.3cm}	\noindent
The set of distributions $F$ for which \eqref{e:gpd} holds with $\xi =0$, the Gumbel domain, mainly contains exponentially decreasing tails such as normal, gamma, Weibull and lognormal models. The Weibull domain corresponding to  $\xi<0$ consists of  distributions with a finite endpoint. Estimation of the EVI with $\xi \in \mathbb{R}$ has also been studied in detail, and here we can refer to  \cite{beir2004} and  \cite{de2006} for general reviews. In this paper we make use of the  generalized Hill estimator of $\xi \in \mathbb{R}$ given by
\begin{equation}
GH_{k} = {1 \over k}\sum_{j=1}^k 
\log (X_{n-j+1,n}H_{j}) - \log (X_{n-k,n}H_{k+1}).
\label{e:GHill}
\end{equation}	
This estimator can be visualized through linear regression  on the  extreme right $k$  points of the {\bf generalized QQ-plot} defined as   $$\Big(\log \Big(\frac{j+1}{n+1}\Big),\log (X_{n-j+1,n}H_j)\Big), \:\: j=1, \cdots,n-1.$$  Indeed this plot can be shown to be  ultimately linear and the slope can be measured with $GH_k$ (see section 5.2.3 in \cite{beir2004}).

While the EVT methods as described above concentrate on the right hand side tail, the method can be extended to the left hand side tail by applying it for instance to the transformed variable $1/X$   in case of positive data, or to $-X$ in case of negative values to the left of the median.\\

In this paper we generalize the approach from  \cite{bhatt2019}  and show that  the trimmed Hill statistic 
\begin{eqnarray}
H_{k_0,k}& =&\frac{k_0}{k-k_0}\log \Big(\frac{X_{n-k_0,n}}{X_{n-k,n}}\Big)+\frac{1}{k-k_0}\sum_{i=k_0+1}^k\log \Big(\frac{X_{n-i+1,n}}{X_{n-k,n}}\Big) \nonumber \\
&=& \frac{1}{k-k_0}\sum_{j=k_0+1}^{k} V_j, \; 0\leq k_0 < k \leq n,
\label{e:trim}
\end{eqnarray}
and  the corresponding test for outliers based on the ratio
\begin{equation}
T_{k_0,k} = \frac{(k-k_0-1)H_{k_0+1,k}}{(k-k_0)H_{k_0,k}}, \; 0\leq k_0 < k,
\label{e:stat-trim}
\end{equation}
can  still be used to develop an outlier detection method in the general case  $\xi \in \mathbb{R}$, while \cite{bhatt2019} considered $\xi >0$. 	 An estimate of the number of outliers is then obtained by a sequential method identifying the largest value of $k_0$ for which $T_{k_0,k}$ yields a significant value. Using this outlier detection mechanism, we can construct a {\em tail-adjusted boxplot} (see section \ref{sec:aut-trim} below)  where the whiskers extend from the upper quartile $Q_3$ up to the largest value among all non-outlying observations, and similarly for the left whisker below $Q_1$.
\\

\noindent Next to obtaining a value of $k_0$, an estimate of $\xi \in \mathbb{R}$ is of course very informative as it will offer a first guess concerning the type of the underlying tail. While estimation of the EVI in the presence of outliers is not the main goal of this paper, we will need a trimmed version of the generalized Hill estimator $GH_k$ as a starting value in the outlier detection procedure. Given a  value of $k_0$,  let
\begin{eqnarray}
GH_{k_0,k}& =& \frac{1}{k-k_0}\sum_{j=k_0+1}^{k}	\log (X_{n-j+1,n}H_{k_0,j}) - \log (X_{n-k,n}H_{k_0,k+1})  , \; 0\leq k_0 < k \leq n,
\label{e:gen-hill-trim}
\end{eqnarray}
where, for $k_0 <j$,
\begin{equation}
H_{k_0,j} = {1 \over j-k_0}\sum_{i=k_0+1}^j \log X_{n-i+1,n} - \log X_{n-j,n} \, , 
\label{e:hill-trim}
\end{equation}
which is a Hill-type estimator based on the observations $X_{n-j,n}\leq X_{n-j+1,n} \leq \cdots \leq X_{n-k_0,n}$.
\\
Plotting the  estimates $GH_{k_0,k}$ as a function of $k_0$ will be used as a visual device in identifying change points (such as outliers or shifts of regime) in the data characteristics.  The plot of $GH_{k_0,k}$ as a function of $k_0$ for some particular  values of $k$ will be referred to as the {\bf diagnostic $k_0$ plot}.\\


In order to illustrate the importance of tail behaviour in the determination of outliers, we consider precipitation data from France as considered in \cite{fre-clim}. Weekly maxima of hourly precipitation at 92 French stations during the fall season from 1993 to 2011 are considered, as provided by the French meteorological service M\'et\'eo-France.
The stations were  chosen in function of their quality and to have a fairly homogeneous coverage of France. Here we consider the data from Chamonix and Uzein, respectively situated in that South-East and South-West of France. \cite{fre-clim} situated the Chamonix case in the Gumbel domain while the Uzein data appeared to be Pareto distributed.  In order to avoid issues with ties a small uniform noise $U(-0.01, 0.01)$ was added to the data. The diagnostic $k_0$ plots for both stations will be given in section \ref{sec:case}. 

Concerning the  Chamonix precipitation data (N $45.93^\circ$, E $6.88^\circ$) \cite{fre-clim} proposed the estimate 0.01 for $\xi$, so that this right hand tail is situated near the Gumbel domain.   The linear fit on the exponential QQ-plot of the Chamonix data given in Figure \ref{fig:precip_1_intro} indeed shows that the tail behavior is close to exponential. While the classical boxplot indicates 4 potential outliers, the method presented here  indicates  two deviating points which arises from the fact that the second and third largest points are almost equal. When enlarging the random noise the top two values are no longer flagged as outliers. 

In contrast to the Chamonix data, the precipitation at  Uzein (N $43.38^\circ$ E $-0.42^\circ$) appears to have a Pareto-type behaviour as concluded already in \cite{fre-clim} with an EVI estimate of 0.35. This is also confirmed by a linear regression fit on the top 64 points  of the Pareto QQ-plot in Figure \ref{fig:precip_2_intro}.  Here the classical boxplot indicates 8 possible outliers. The present tail-adjusted approach indicates 13 possible outliers when using the full Pareto-type character in the top 64 data points. 
While the top 13 observations can be stated to be in line with the Pareto fit as indicated in the Pareto QQ-plot in Figure \ref{fig:precip_2_intro}, these data are indeed separated from the bulk of the data as can be observed from the time plot in Figure \ref{fig:precip_2_intro} .

\begin{figure}[H]
	\centering
	\hspace{-3mm}\includegraphics[width=0.33\textwidth]{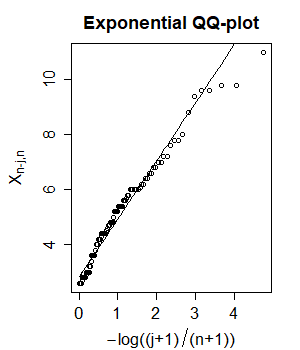}
	\hspace{-5mm}\includegraphics[width=0.33\textwidth]{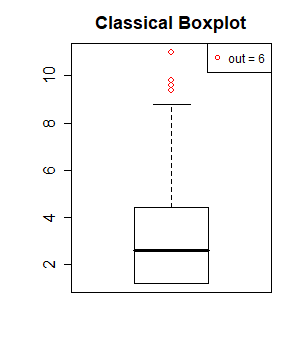}
	\hspace{-5mm}\includegraphics[width=0.33\textwidth]{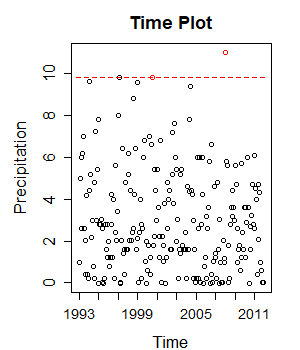}
	\caption{ {\em Precipitation} data from Chamonix station. {\em Left:} Exponential QQ-plot. {\em Middle:} Classical boxplot. {\em Right:}  Time plot with indication of the outlier threshold for the tail-adjusted method when using the top 114 data.}
	\label{fig:precip_1_intro}
\end{figure}

\begin{figure}[H]
	\centering
	\hspace{-3mm}\includegraphics[width=0.33\textwidth]{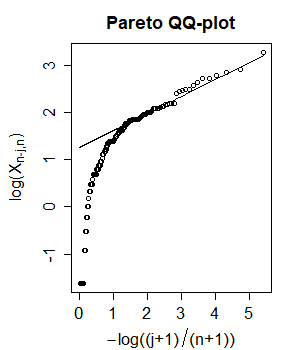}
	\hspace{-5mm}\includegraphics[width=0.33\textwidth]{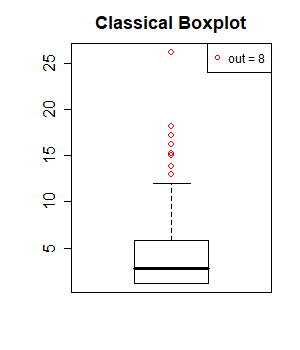}
	\hspace{-5mm}\includegraphics[width=0.33\textwidth]{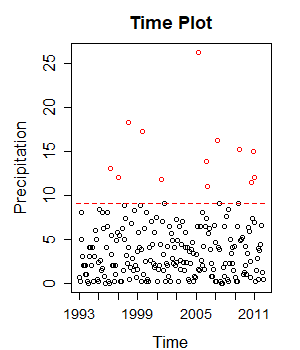}
	\caption{{\em Precipitation} data from Uzein station. {\em Left:} Pareto QQ-plot. {\em Middle:} Classical boxplot. {\em Right:}  Time plot with indication of the outlier threshold for the tail-adjusted method when using the top 64 data. }
	\label{fig:precip_2_intro}
\end{figure}

\noindent
The rest of the paper is organized as follows. In section \ref{sec:test} we develop tests for outliers based on extreme value methodology  whatever the tail behaviour of the underlying distribution is. We provide mathematical justification for the proposed procedure, the proof of which are deferred to the Appendix. In section \ref{sec:sim}  we report on a simulation study, while in section \ref{sec:case} we  discuss several case studies in order to illustrate the practical use of the proposed methods.
Proofs are deferred to the Appendix, as well as some parts of the simulation study.


	\section{Testing for outliers under extreme value conditions}
\label{sec:test}
\noindent
The main goal of this paper is to generalize the approach from \cite{bhatt2019} based on $T_{k_0,k}$ from $\xi >0$ to $\xi \in \mathbb{R}$. To this end we derive distributional properties of the trimmed Hill statistic $H_{k_0,k}$ and the test statistic $T_{k_0,k}$ for all three domains of attraction. This then leads to an algorithm for automated selection of the number of outliers $k_0$.	

\subsection{On the distribution of $H_{k_0,k}$ and $T_{k_0,k}$}

We here consider asymptotic properties as $k,n \to \infty$ with $k/n \to 0$ and $k_0/k \to 0$. 
We use conditions under which \eqref{e:gpd} holds in terms of the tail quantile function $U(x)=F^{\leftarrow}(1-x^{-1})$ where $F^{\leftarrow}$ denotes the quantile function which is defined as the left continuous inverse of the df $F$. Given that the Hill statistic is defined in terms of the logarithm of the empirical quantiles, we use the equivalent condition of  \eqref{e:gpd} in terms of $\log U$, as can be found for instance in section 3.5 in \cite{de2006}:
for some positive function $q$ 
$$
{\log U(tx) -\log U(t) \over q(t)}\to \Psi_{\gamma_{-}}(x) :={x^{\xi_-}-1 \over \xi_-} \text{ as } t\to \infty, \text{ for every } x>1,
$$
where $\xi_- = \min (\xi,0)$ and where $\Psi_{\gamma_{-}}(x)$ reads as $\log x$ in case $\xi \geq 0$ or $\xi_- =0$. \\

\noindent
In order to study the bias of the trimmed Hill $H_{k_0,k}$ and $T_{k_0,k}$ statistic we will in fact need a second order condition as can be found in Chapter 3 of  \cite{de2006}: as $t \to \infty$
\begin{equation}
\frac{{\log U(tx) -\log U(t) \over q(t)}- {x^{\xi_-}-1 \over \xi_-}}{Q(t)} \to \int_1^x s^{\xi_- -1}\int_1^s
u^{\rho -1}du \, ds,\; \text{ for all } x>0,
\label{2ndorder}
\end{equation}
with $Q$ not changing sign eventually and $Q(t) \to 0$ as $t \to \infty$, and $\rho \leq 0$. \\

\noindent
By Theorem 2.3.6 in \cite{de2006},  functions $q_0(t)$ and $Q_0(t)$ can be constructed following (3.5.12) and (3.5.13) in that reference, so that for every $\epsilon, \delta>0$, there exists $t_0=t_0(\epsilon,\delta)$ such that for all $t$, $t\geq t_0$ and $x>1$:
\begin{equation}
\Bigg|\frac{\frac{\log U(tx)-\log U(t)}{q_0(t)}-\Psi_{\xi_{-}}(x)}{Q_0(t)}-\Phi_{\xi_-,\rho}(x) \Bigg|\leq \epsilon x^{\xi_{-}+\rho +\delta},
\label{e:bounds}
\end{equation}
with $\Phi_{\xi_{-}, \rho}(x)$ defined as
$$ \Phi_{\xi_{-}, \rho}(x)=\begin{cases}
\frac{x^{\xi_{-}+\rho}-1}{\xi_{-}+\rho},&\xi_{-}+\rho < 0, \rho<0,\\
\frac{x^{\xi_{-}}}{\xi_{-}}\log x,& \xi_{-} < 0, \rho=0,\\
\frac{1}{2}(\log x)^2,& \gamma_{-}=0, \rho=0.\\
\end{cases}$$

To obtain an approximation of the distribution of the test statistic $T_{k_0,k}$, we first state a general representation theorem for the trimmed Hill statistics $H_{k_0,k}$ after which 	we can provide an asymptotic representation of $T_{k_0,k}$.	
\begin{theorem}
	\label{thm:xi-dist}
	Suppose \eqref{2ndorder} holds for some $\xi \in \mathbb{R}$, $\rho \leq 0$. Then for $k/n \rightarrow 0$, $k_0=o(k)$, and \begin{equation}
	\label{e:q-asy}
	\sqrt{k} Q(n/k) \to \lambda
	\end{equation}
	with $\lambda$ finite, we have
	\begin{itemize}
		\item when $\xi \geq 0$,
		$$\frac{k-k_0}{k}\frac{H_{k_0,k}}{q_0(Y_{n-k,n})}\stackrel{d}{=}
		\frac{1}{k}\sum_{i=k_0+1}^{k}Z_i+\frac{\lambda c_{0,\rho}}{\sqrt{k}}+o_{\mathbb{P}}(k^{-1/2})
		$$
		\item when $\xi<0$,
		$$\frac{k-k_0}{k}\frac{H_{k_0,k}}{q_0(Y_{n-k,n})}\stackrel{d}{=}
		\frac{k_0}{k\xi}(\exp(\xi \sum_{j=k_0+1}^{k}Z_j/j)-1)+ \frac{1}{k\xi}\sum_{i=k_0}^{k-1}(\exp(\xi \sum_{j=i+1}^{k}Z_j/j)-1)+\frac{\lambda c_{\xi,\rho}}{\sqrt{k}}+o_{\mathbb{P}}(k^{-1/2})$$
	\end{itemize}
	where the constants $c_{0,\rho}$ and $c_{\xi,\rho}$ satisfy
	\begin{align}
	\label{e:bias-def}
	c_{0,\rho}=\begin{cases}
	\frac{1}{\rho(1-\rho)}\Big(1-(\frac{k_0}{k})^{1-\rho}\Big), & \rho<0\\
	\Big(1-\frac{k_0}{k}\Big(1+\log \frac{k_0}{k}\Big)\Big), & \rho=0
	\end{cases}; \;\;
	c_{\xi,\rho}=\begin{cases}
	\frac{1}{\rho(1-\xi-\rho)}\Big(1-{(\frac{k_0}{k})}^{1-\xi-\rho}\Big), &\rho<0,\\
	\frac{1}{\xi(1-\xi)^2}\Big(1-{(\frac{k_0}{k})}^{1-\xi}\Big(1+\xi(1-\xi)\log\frac{k_0}{k}\Big)\Big), &\rho=0.
	\end{cases}
	\end{align}
\end{theorem}

\noindent	Next, we consider the asymptotic distribution of $T_{k_0,k}$.
\begin{theorem}
	\label{thm:t-dist}
	Suppose \eqref{2ndorder} holds for some $\xi \in \mathbb{R}$, $\rho \leq 0$. Then for $k/n \rightarrow 0$, $k_0=o(k)$, and $$
	\sqrt{k} Q(n/k) \to \lambda
	$$
	with $\lambda$ finite,   we have for any $\delta >0$ that
	\begin{itemize}
		\item
		when $\xi \geq 0$		
		$$	1-T_{k_0,k} \stackrel{d}{=} \frac{Z_{k_0+1}	\left(1+O_\mathbb{P}(k^{-1/2})+O_{\mathbb{P}}(k^{\rho+\delta-1/2}k_0^{1-\rho-\delta}) \right)}{\sum_{j=k_0+1}^k Z_j}$$
		\item when $\xi <0$
		\begin{align*} 
		1-T_{k_0,k}\stackrel{d}{=}\Big(\frac{k_0+1}{k}\Big)^{\xi}\:\:\frac{(k_0+1)(\exp(\frac{\xi Z_{k_0+1}}{k_0+1})-1)\left(1+O_\mathbb{P}(k^{-1/2})+O_{\mathbb{P}}(k^{\rho+\delta-1/2}k_0^{1-\rho-\delta})\right)}{k_0(\exp{(\sum_{j=k_0+1}^k \frac{\xi Z_j}{j})}-1)+\sum_{i=k_0}^{k-1}(\exp({\sum_{j=i+1}^k}\frac{\xi Z_j}{j})-1)}
		\end{align*}
	\end{itemize}
	where $Z_{1}, Z_2, \cdots$ are independent standard exponentially distributed.
\end{theorem}

\noindent
{\bf Remark 2.1.} Note that in case $\xi  \geq 0$ the statistic $1-T_{k_0,k}$ is approximated by a Beta$(k-k_0-1,1)$ distribution. This corresponds with Proposition 4.1 in \cite{bhatt2019} which was obtained for $\xi >0$ with $\ell$ in \eqref{Patype} constant.\\

\noindent
The next result follows  from Theorem 2.2 using the central limit theorem on $k^{-1}\sum_{j=k_0+1}^k Z_j$ and $k^{-1}(k_0\exp{(\sum_{j=k_0+1}^k \xi Z_j/j)}+\sum_{i=k_0}^{k-1}\exp({\sum_{j=i+1}^k}\xi Z_j/j)-k)$.

\begin{corollary}
	\label{cor:k0-k-est}
	Under the conditions of Theorem 2.3, for any $\delta >0$ we have
	$$
	\left({k \over {k_0+1}} \right)^{1-\xi_-}(1-T_{k_0,k})
	\stackrel{d}{=}(1-\xi_-) h_{\xi_-}\left( {Z \over k_0+1}\right)
	\left(1+O_{\mathbb{P}}(k^{-1/2})+ O_{\mathbb{P}}(k^{\rho+\delta-1/2}k_0^{1-\rho-\delta}) \right),
	$$
	with $Z$ standard exponentially distributed and $h_{\xi_-}(u)= (e^{u\xi_-}-1)/\xi_-$ if $\xi_-<0$, and $h_{\xi_-}(u)=u$ otherwise.
\end{corollary}

\begin{corollary}
	\label{cor:k0-k-est-1} Under the conditions of Theorem 2.3, for any $\delta >0$ we have that
	\begin{enumerate}
		\item for $\xi\geq 0$ and $k_0=o(k^{(1/2-\rho-\delta)/(1-\rho-\delta)})$, 
		$$k(1-T_{k_0,k})\stackrel{d}{\to}Z,$$
		\item    for $\xi<0$, $k_0=o(k^{(1/2-\rho-\delta)/(1-\rho-\delta)})$ and $k_0 \to \infty$,
		$$k\Big(\frac{k_0+1}{k}\Big)^{\xi}(1-T_{k_0,k})\stackrel{d}{\to}(1-\xi)Z,$$
		\item for $\xi<0$ and $k_0\to M <\infty$,
		$$\Big(\frac{k}{k_0+1}\Big)^{1-\xi}(1-T_{k_0,k})\stackrel{d}{\to}\frac{1-\xi}{\xi}\Big(e^{\frac{\xi Z}{M}}-1\Big),$$
	\end{enumerate}
	with $Z$ standard exponentially distributed.
\end{corollary}

\subsection{Algorithm for automated selection}
\label{sec:aut-trim}

\hspace{5mm} Here we describe the formal methodology for  estimating  the number of outliers. \\
In view of Corollary \ref{cor:k0-k-est}, note that the distribution of $T_{k_0,k}$  depends on the true value of $\xi$.
Setting
\begin{equation*}
\label{e:Ek-def-g}
E^{\xi}_{k_0,k}:=\begin{cases}
k(1-T_{k_0,k}), & \xi>0\\
\frac{k_0+1}{\xi}\log\Big(1+ (\frac{k}{k_0+1})^{1-\xi}\frac{\xi}{1-\xi}(1-T_{k_0,k})\Big), &\xi\leq 0,
\end{cases}
\end{equation*}
with  $k_0=0,1, \cdots, k-2$,   we obtain from Corollary \ref{cor:k0-k-est} that $E^{\xi}_{k_0,k}$ is approximately standard exponentially distributed for large enough values of $k$ and $n$ with $k/n$ and $k_0/k$ sufficiently small, which implies that 
\begin{equation}
\label{e:U-i-k}
U^{\xi}_{k_0,k}=2|0.5-\exp(-E^{\xi}_{k_0,k})|,\:\: k_0=0,1, \cdots, k-2
\end{equation}
approximately follows the uniform (0,1) distribution.
The main items in the selection algorithm are given next.
\\ 

\noindent {\bf An initial estimator of $\xi$.} Assuming an upper bound $k_0^*$ for the number of outliers $k_0$,  one can construct a crude initial estimate of $\xi$ based on $GH_{k^*_0,k^*}$ as defined in \eqref{e:gen-hill-trim} using some appropriate values $k^*_0$ and $k^*$ where $k^*$ can be different from the value $k$ used in \eqref{e:U-i-k}. \\

\noindent Imputing $GH_{k_0^*,k^*}$ for $\xi$ in \eqref{e:U-i-k}, we can identify $k_0$ as the largest value $j$ for which $U_{j,k}^{\xi}$ fails a test for uniformity on significance level $\alpha_j$ for some sequence $(\alpha_j; \; j=0,\cdots, k-2$) discussed below.  Numerically, one can thus write 	 
\begin{equation}
\label{e:k0-1}
\widehat{k}_0^{\{0\}}=\underset{j =1, \cdots, k_0^*}{{\rm argmax}} \:\{j:U_{j,k}^{GH_{k_0^*,k^*}}>1-\alpha_j \}
\end{equation}
where the ${\rm argmax}$ of an empty set is set to 0. With this value  $\widehat{k}_0^{\{0\}}$ we define a revised estimate of $\xi$ as
\begin{equation}
\label{e:xi-1-def}
\widehat{\xi}=GH_{\widehat{k}^{\{0\}}_0,k^*}.
\end{equation}
Note that the estimator in \eqref{e:xi-1-def} is not an optimal estimator of $\xi$. This is because $GH_{i,k^*}$ is asymptotically suboptimal especially at large values of $i$. Thus, one must guard against choosing large values of $k_0^*$ in defining the initial estimate of $\xi$. On the other hand if $k_0^*$ is less than the true number of outliers, $GH_{k_0^*,k^*}$ will involve outlier observations, thereby resulting in a severely biased estimate of $\xi$.  Extensive analysis showed that $k_0^*=o({k^*}^{1/2})$ works well in practice. A more specific proposal is discussed in the simulation section. \\	

\noindent {\bf Estimation of $k_0$.} Based on the initial estimate $\widehat{\xi}$ from \eqref{e:xi-1-def}, we obtain an improved estimate of $k_0$:
\begin{equation}
\label{e:k0-def}
\widehat{k}_0^{\{1\}}=\underset{j =1, \cdots, k_0^*}{{\rm argmax}} \:\{j: U^{\widehat{\xi}}_{j,k}>1-\alpha_j \}.
\end{equation} 

\noindent
{\bf Groups of outliers.} One can extend the above approach to allow for different groups of outliers. In this direction, we define
$$\mathcal{V}=\{j \in 1,\cdots,  \widehat{k}_0^{\{1\}}: U^{\widehat{\xi}}_{j,k}>1-\alpha_j  \}$$
which are the possible points at which a significance is detected among the outliers. Let $0=v_0\leq v_1\leq v_2 \leq  \cdots \leq v_{|\mathcal{V}|}= \widehat{k}_0^{\{1\}}$ denote the entries of $|\mathcal{V}|$ in increasing order.

With an a priori assumption  that there are no more than $V\leq |\mathcal{V}|$ regime outliers, we next define $V$ regime of outliers as
\begin{equation}
\label{e:l-def}
\widetilde{\ell}_{\{r\}}=\begin{cases}
\{v_{r-1},\cdots, v_r\}, & r=1, \cdots, V-1\\
\{v_{r-1}, \cdots, v_{|\mathcal{V}|}\}, & r=V
\end{cases}
\end{equation}
with $\widetilde{\ell}_{\{r\}}=\phi$.\\

This then leads to the following detection algorithm and consequently a tail-adjusted boxplot.\\


\begin{algorithm}[H]
	\caption{Domain Adapted Sequential Testing (DAST)}
	\label{algo:dast}
	\small
	\begin{algorithmic}[1]
		\STATE Input $k$, $k^*$, $k_0^*$ and $V$.
		\STATE Set $\widetilde{V}=0$ and $\widehat{\ell}^{\{0\}}=0$.
		\STATE Set $\alpha_j \in (0,1), \:j=0, 1, \cdots, k-2$ satisfying
		\begin{equation}
		\label{e:alpha-j-def}
		\prod_{j=0}^{k-2}(1-\alpha_j)=1-q
		\end{equation}
		\STATE Compute $\widehat{k}^{\{0\}}_0$ and $\widehat{\xi}$  using Relations \eqref{e:k0-1} and \eqref{e:xi-1-def} respectively.
		\STATE Compute $U^{\widehat{\xi}}_{j,k}$ as in Relation \eqref{e:U-i-k}.
		\STATE Compute $\widehat{k}_0^{\{1\}}$ as in Relation \eqref{e:k0-def}.
		\STATE If $\widetilde{k}_0^{\{1\}}=0$ goto step 10, else define
		$$\mathcal{V}=\{j \in 1,\cdots,  \widehat{k}_0^{\{1\}}| U^{\widehat{\xi}}_{j,k}>1-\alpha_j\}.$$
		\STATE Set $\widetilde{V}=\min( {|\mathcal{V}|,V})$
		\STATE If $\widetilde{V}=1$, define $\ell_{\{1\}}=\widehat{k}_0^{\{1\}}$ and goto step 10, else define 
		\begin{equation*}
		\widetilde{\ell}_{\{r\}}=\begin{cases}
		\{v_{r-1},\cdots, v_r\}, & r=1, \cdots, \widetilde{V}-1\\
		\{v_{r-1}, \cdots, v_{|\mathcal{V}|}\}, & r=\widetilde{V}
		\end{cases}
		\end{equation*}
		where $0=v_0\leq v_1\leq v_2 \leq  \cdots \leq v_{|\mathcal{V}|}= \widehat{k}_0^{\{1\}}$ denote the entries of $|\mathcal{V}|$ in increasing order.
		\STATE Return $\{\widehat{\ell}_{\{0\}}, \cdots, \widehat{\ell}_{\{\widetilde{V}\}}\}$ 
	\end{algorithmic}
\end{algorithm}

\noindent
{\bf Choice of $\alpha_j$.} As in \cite{bhatt2019}, the levels $\alpha_j$ in the above algorithm are chosen as 
\begin{equation}
\label{e:alpha-j}
\alpha_j=1-(1-q)^{ca^{k-j-1}}, \hspace{5mm} j=0, \cdots, k-2,
\end{equation}
with $a>1$ and  $c=1/\sum_{j=0}^{k-2}a^{k-j-1}$. This choice of $\alpha_j$ satisfies  \eqref{e:alpha-j-def}, which implies 
$$P_{k_0=0}[\widehat{k}_0>0]=q,$$
so that the algorithm is well calibrated (see Proposition 2.9 in \cite{bhatt2019}). In addition, this choice puts less weight on large values of $j$  which implies that large values of $j$ are less likely to be chosen over smaller ones. This guards against encountering spurious values of $\widehat{k}_0$ close to $k$. Extensive sensitivity analysis with a variety of sequential tests indicate that the choice of levels as in  \eqref{e:alpha-j} with $a=1.2$ works well in practice.\\

\noindent	
{\bf Tail-adjusted boxplot. }
We propose to apply the DAST algorithm to obtain the right tail outliers, while the bottom outliers are determined  by applying the DAST algorithm to the right tail of the transformed random variable  $1/X$ in case of positive data, and $-X$ in case of negative left tail data. One can indicate the different sets of outliers as detected by the above algorithm using different symbols such as +, $\circ$, $\cdots$ . The whiskers of the tail-adjusted boxplot extends from the upper, respectively lower, quartile to the extreme observations just below, respectively above, the set of identified outliers. In the case studies below, we also report the p-values of the different sets of outliers. For each regime $\widetilde{\ell}_r$, these p-values are equal to $U^{\widehat{\xi}}_{v_r,k}$ with $\widehat{\xi}$ and $v_r$ as in \eqref{e:xi-1-def} and \eqref{e:l-def} respectively.

	\section{Simulations.}
\label{sec:sim}

In this section, we study the  accuracy of the DAST algorithm (see Algorithm \ref{algo:dast} of Section \ref{sec:aut-trim}) as an estimator of the true number of outliers $k_0$. The parameters $a$ and $q$ are set at 1.2 and  0.05 respectively. We concentrate only on one regime of outliers, i.e. $V$ is set equal to 1 in Algorithm \ref{algo:dast}. 

{\bf \vspace{2mm}\noindent Measures of Performance:} The performance of $\widehat{k}_0$ as generated from the DAST algorithm is evaluated in terms of its expected value,	${\rm E}(\widehat{k}_0)$ and the standard deviation $\sqrt{{\rm Var}(\widehat{k}_0)}$.
These computations are based on 2500 independent Monte Carlo simulations.

{\bf \vspace{2mm} \noindent Data generating models:} We generate $n$ i.i.d. observations from three domains of attractions, viz $\xi>0$, $\xi=0$ and $\xi>0$ for the following distributions.
\begin{enumerate}
	\item {\bf Case $\xi>0$:}
	\begin{eqnarray}
	\label{e:xi-pos}
	&&\modt(n)\::\:1-F(x)=\int_{x}^{\infty}\frac{2\Gamma(\frac{n+1}{2})}{\sqrt{n\pi}\Gamma(\frac{n}{2})}\left(1+\frac{w^2}{n} \right)^{-\frac{n+1}{2}}dw, \:\: x>0,\:\:\xi=\frac{1}{n} \nonumber\\
	&&	{\rm Burr} (\eta,\lambda,\tau)\::\:1-F(x)=\left(\frac{\eta}{\eta+x^{\tau}}\right)^{\lambda},\:\:  x>0,\:\: \eta, \lambda, \tau>0, \:\: \xi=\frac{1}{\lambda \tau}
	\end{eqnarray}
	
	\item {\bf Case $\xi=0$:}
	\begin{eqnarray}
	\label{e:xi-zero}
	&&{\rm Lognormal}(\mu,\sigma^2)\::\:1-F(x)=\int_{x}^{\infty}\frac{1}{x \sqrt{2\pi \sigma^2}}\exp\left(-\frac{(\log (x)-\mu)^2}{2\sigma^2}\right),\:\: x>0\nonumber\\
	&&\modn (\mu,\sigma^2)\::\:1-F(x)=\int_{x}^{\infty}\frac{1}{\sqrt{2\pi \sigma^2}}\exp\left(-\frac{(x-\mu)^2}{2\sigma^2}\right)\nonumber\\
	&&{\rm Weibull}(\lambda,\tau)\::\:1-F(x)=\exp(-\lambda x^\tau),\:\: x>0,\:\:\lambda>0,\tau\geq 0
	\end{eqnarray}
	
	\item {\bf Case $\xi<0$:}
	\begin{eqnarray}
	\label{e:xi-neg}
	&&{\rm Beta}(p,q)\::\:1-F\left(x_+-\frac{1}{x}\right)=\int_{1-\frac{1}{x}}^{1}\frac{\Gamma(p+q)(1-u)^{q-1}}{\Gamma(p)\Gamma(q)u^{1-p}}du,\:\:x>1,p,q>0,\:\:\xi=-\frac{1}{q}\nonumber\\
	&&\text{Reverse Burr}(\eta,\lambda, \tau)\::\:1-F(x_+-\frac{1}{x})=\left(\frac{\eta}{\eta+x^{\tau}}\right)^{\lambda},\:\:x>0,\eta,\lambda, \tau>0,\:\:\xi=-\frac{1}{\lambda\tau}
	\end{eqnarray}
	
\end{enumerate}
The performance of the DAST algorithm for the Fr\'echet domain of attraction, i.e. $\xi>0$, is discussed in section \ref{sec:xi-pos}.  In sections \ref{sec:xi-zero} and \ref{sec:xi-neg}, we discuss the performance under  Gumbel $(\xi=0)$ and reverse-Weibull $(\xi<0)$ domains of attractions respectively.

\vspace{2mm}\noindent	
{\bf Choice of $k$, $k^*$ and $k_0^*$:} In Appendix A.1, we study the sensitivity of  the DAST algorithm with respect to the choice of $k$, $k^*$ and  $k_0^*$. Both $k$ and $k^*$ in a neighborhood of $k^{\rm opt}$ where $k^{\rm opt}$ is given by
\begin{equation}
\label{e:k-star-def}
k^{\rm opt}= \underset{k}{{\rm argmin}}\:\: {\rm Var}(GH_k)
\end{equation}
Thus, $k^{\rm opt}$ is the optimal $k$ at which the generalized Hill, $GH_{k}$ (see \eqref{e:GHill}) attains the minimal variance. This $k^{\rm opt}$ is empirically estimated using monte carlo simulations.

Sensitivity analysis showed that choosing  $k_0^*$ as:
\begin{equation}
\label{e:k0-star-def}
k_0^*=c{k^*}^{\frac{1}{3}}
\end{equation}
for $c \in (5,10)$ works well in practice.

\vspace{2mm}
\noindent	{\bf  Outlier Scenarios:} In Sections \ref{sec:xi-pos}, \ref{sec:xi-zero} and \ref{sec:xi-neg},outliers are introduced in the extreme observations of the data where the top-$k_0$ order statistics are perturbed as follows:  	
\begin{eqnarray}
\label{e:exp-trans-0}
&&\textit{Exponentiated Outliers}\::\:
X_{(n-i+1,n)}:= X_{(n-k_0,n)}\left(\frac{X_{(n-i+1,n)}}{X_{(n-k_0,n)}}\right)^L, \hspace{3mm} i=1,\cdots,k_0.\\ 
\label{e:scl-trans-0}
&&\textit{Scaled Outliers}\::\:
X_{(n-i+1,n)}:= X_{(n-k_0,n)}+C(X_{(n-i+1,n)}-X_{(n-k_0,n)})), \hspace{3mm} i=1,\cdots,k_0. 
\end{eqnarray}

Note, the number of outliers is fixed at $k_0$ and the constants $L$ and  $C$  control the intensity of the injected outliers. Whereas $L<1$ and $C<1$ shrinks the top order statistics,  scenarios $L>1$ and $C>1$ inflate the top values. The case of $L=1$ and $C=1$ correspond to the regime of no outliers.  For the above two kinds of outliers, the order of the bottom $(n-k_0)$ observations is preserved. 

\subsection{Case $\xi>0$}
\label{sec:xi-pos}

We evaluate the performance of the DAST algorithm for the distribution models in Fr\'echet domain of attraction, i.e. $\xi>0$ (see  \eqref{e:xi-pos}). With $n=1000$, we generate data from $\modt(1/\xi)$ and ${\rm Burr}(1,0.5,2/\xi)$.

\begin{table}[H]
	\centering
	\footnotesize
	\begin{tabular}{|c|cccc|cccc|cccc|}
		\hline
		& \multicolumn{4}{c|}{$k$=200} & \multicolumn{4}{c|}{$k$=400} & \multicolumn{4}{c|}{$k$=600}	\\
		$k^*=$ & 200 & 400 & 600 & $\xi$ known  & 200 & 400 & 600 & $\xi$ known & 200 & 400 & 600 & $\xi$ known  \\\hline
		$\xi$ =0.25 & 0.092& 0.036& 0.034&  0.034& 0.09& 0.048& 0.048 & 0.048 & 0.081& 0.066& 0.066& 0.066\\
		$\xi$ =0.5& 0.05& 0.042& 0.042& 0.042 &0.046& 0.038& 0.038& 0.038 & 0.048& 0.045& 0.045 & 0.045\\
		$\xi$= 1& 0.046& 0.046& 0.046& 0.046& 0.036& 0.036& 0.036& 0.036& 0.032& 0.032& 0.032 & 0.032	\\\hline			
	\end{tabular}	
\end{table}
\vspace{-5mm}
\begin{table}[H]
	\centering
	\footnotesize
	\begin{tabular}{|c|cccc|cccc|cccc|}
		\hline
		& \multicolumn{4}{c|}{$k$=100} & \multicolumn{4}{c|}{$k$=200} & \multicolumn{4}{c|}{$k$=300}	\\
		$k^*=$ & 100 & 200 & 300 & $\xi$ known  & 100 & 200 & 300 & $\xi$ known & 100 & 200 & 300 & $\xi$ known  \\\hline
		$\xi$ =0.25 & 0.585& 0.206& 0.114&0.048& 0.772& 0.265& 0.124& 0.036& 0.829& 0.295& 0.115& 0.034\\
		$\xi$ =0.5 	& 0.231& 0.063& 0.048& 0.048& 0.382& 0.04& 0.03& 0.03& 0.456& 0.039& 0.03 & 0.03\\ 
		$\xi$ =1 & 0.074& 0.058& 0.058& 0.058 &0.074& 0.043& 0.043& 0.043&0.074& 0.036& 0.036&0.036\\\hline 
	\end{tabular}			
	\caption{Type 1 error for the DAST algorithm for $\xi>0$ with $k_0^*=7(k^*)^{1/3}$. {\em Top:} $\modt(1/\xi)$ distribution. {\em Bottom:} Burr(1,0.5,2/$\xi$) distribution.}
	\label{tab:xi-pos-h0}
\end{table}

\begin{figure}[H]
	\hspace{-2mm}\includegraphics[width=0.34\textwidth]{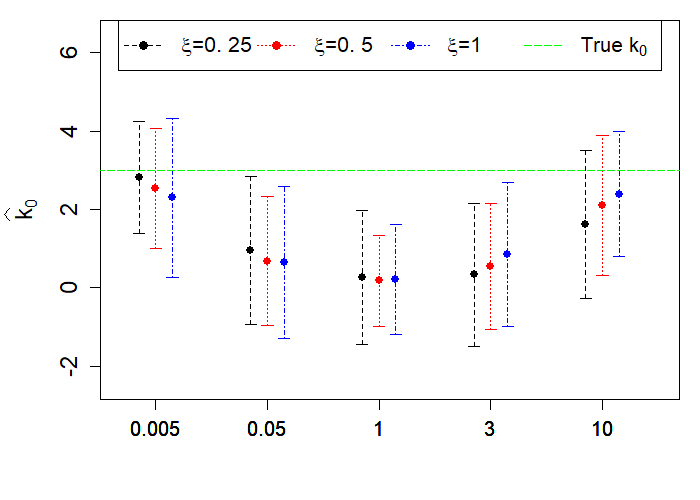}
	\hspace{-3mm}	\includegraphics[width=0.34\textwidth]{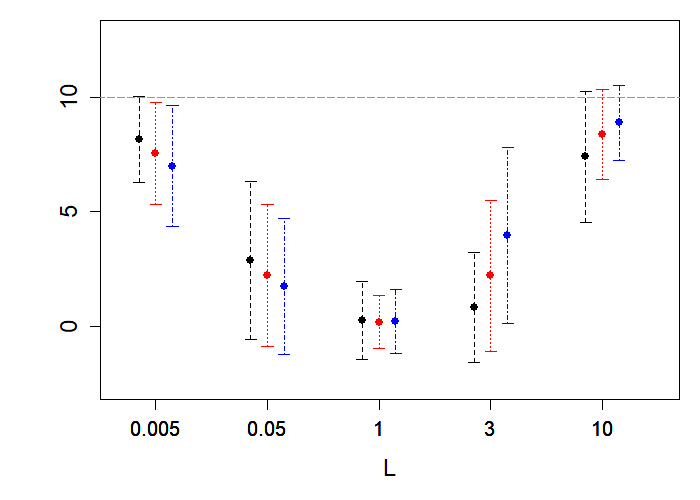}
	\hspace{-3mm}	\includegraphics[width=0.34\textwidth]{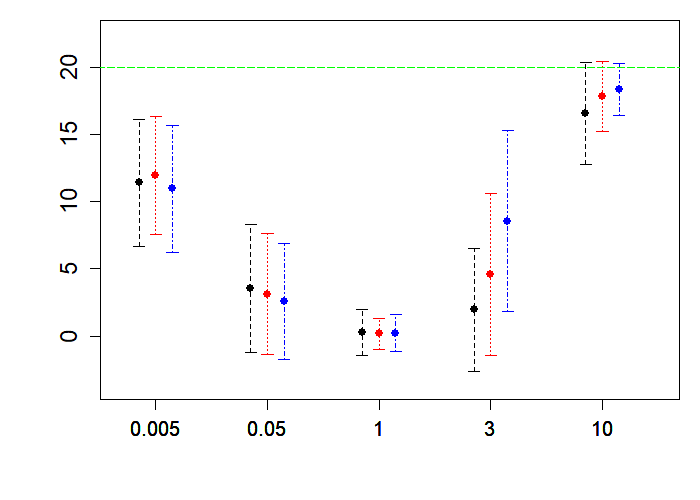}
\end{figure}

\vspace{-7mm}

\begin{figure}[H]
	\hspace{-2mm}\includegraphics[width=0.34\textwidth]{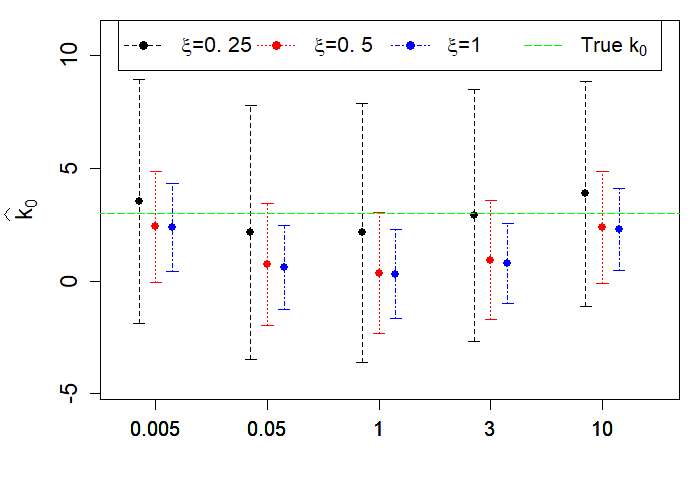}
	\hspace{-3mm}	\includegraphics[width=0.34\textwidth]{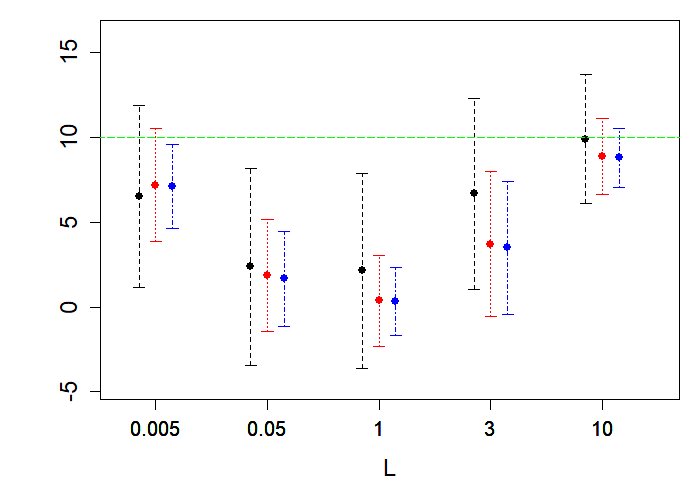}
	\hspace{-3mm}	\includegraphics[width=0.34\textwidth]{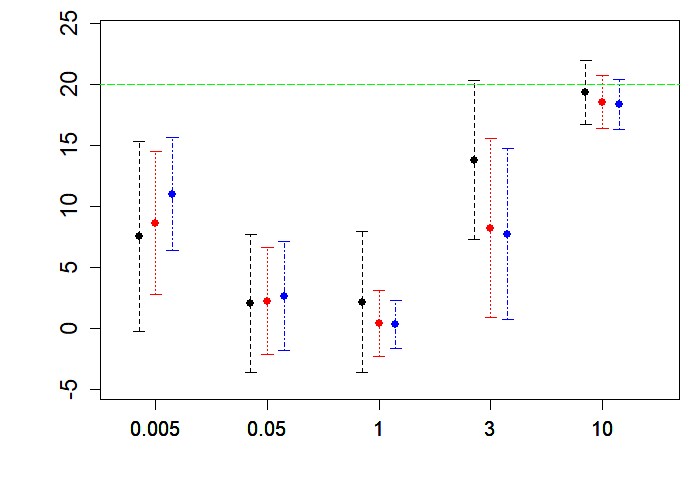}
	\caption{$E(\widehat{k}_0)\pm \sqrt{{\rm Var}(\widehat{k}_0)}$ for exponentiated outliers, $k_0^*=7(k^*)^{1/3}$. {\em Top:} $\modt(1/\xi)$, $k=k^*=400$. {\em Bottom:} ${\rm Burr}(1,0.5,2/\xi)$, $k=k^*=200$. {\em Left:} $k_0=3$. {\em Middle:} $k_0=10$. {\em Right:} $k_0=20$.} 
	\label{fig:xi-pos-ha-exp}
\end{figure}

For the set up where no outliers have been injected, i.e. $k_0=0$, Table \ref{tab:xi-pos-h0} gives the type 1 error $\mathbb{P}(\widehat{k}_0>0)$ for the DAST algorithm.  The algorithm is fairly stable in terms of the choice of $k$ for both Burr and $\modt$ distributions. Different values of $k^*$  which are used in the initial estimate of $\xi$ (see  \eqref{e:xi-1-def}) are indicated in the first 3 columns. On the other hand, the last column uses the true value of $\xi$ as the initial estimate.  For the $\modt$ distribution, the algorithm attains the nominal significance level at all values of $\xi$ irrespective of the initial choice of $\xi$. However, for the Burr distribution, smaller values of $k^*$ produce less reliable estimates of the initial estimate $\xi$ especially when $\xi$ is small. This explains the larger values of type 1 error at $\xi=0.25$ and $k^*=100$.

For $\modt$ and Burr distributions, Figure \ref{fig:xi-pos-ha-exp} shows the performance of the DAST for varying intensity of the exponentiated outliers. We inject $k_0$ outliers according to  \eqref{e:exp-trans-0}. In this scenario, Table \ref{tab:xi-pos-ha-exp} in Appendix A.1 explores in more detail the sensitivity of the algorithm to varying choices of $k$ and $k^*$.

\begin{figure}[H]
	\hspace{-2mm}\includegraphics[width=0.34\textwidth]{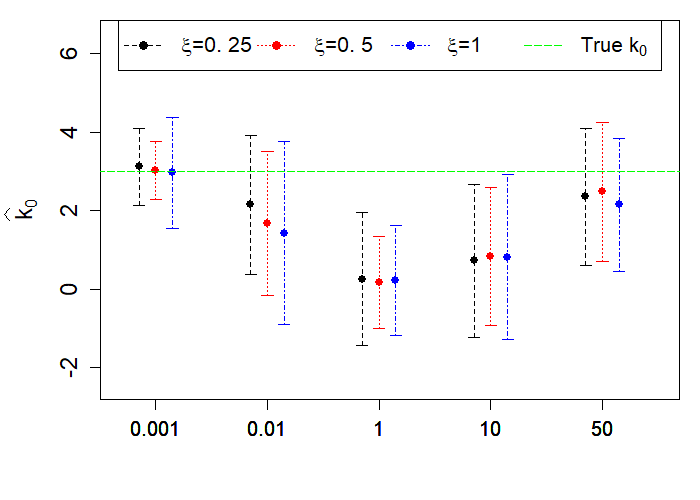}
	\hspace{-3mm}	\includegraphics[width=0.34\textwidth]{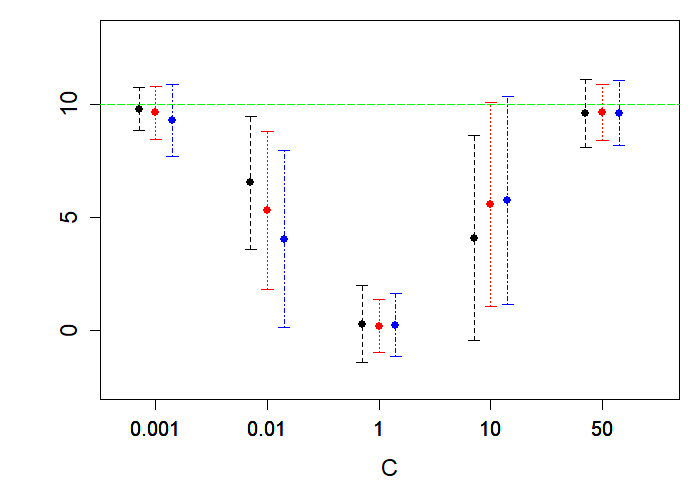}
	\hspace{-3mm}	\includegraphics[width=0.34\textwidth]{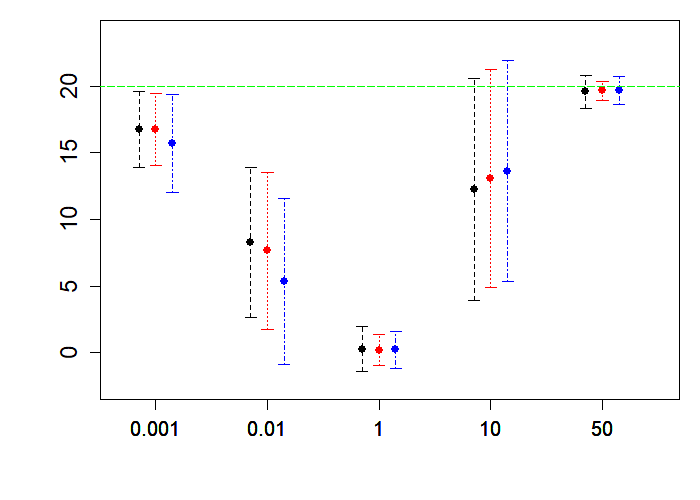}
\end{figure}
\vspace{-7mm}
\begin{figure}[H]
	\hspace{-2mm}	\includegraphics[width=0.34\textwidth]{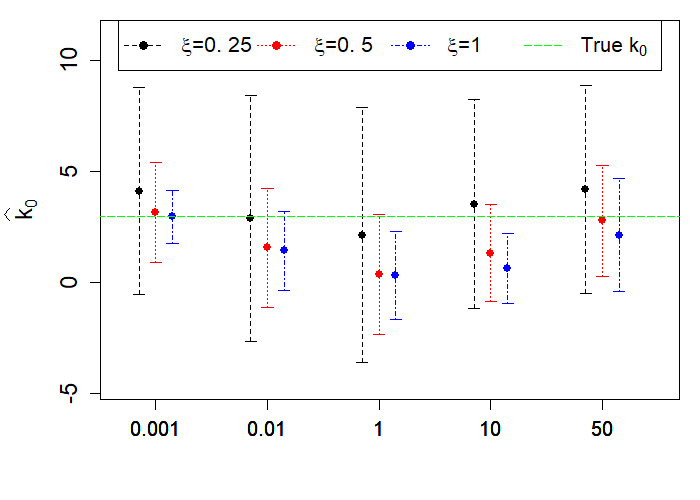}
	\hspace{-3mm}	\includegraphics[width=0.34\textwidth]{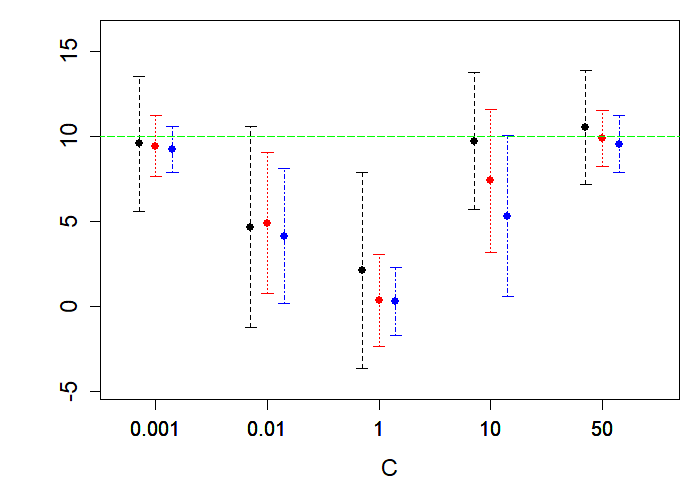}
	\hspace{-3mm}	\includegraphics[width=0.34\textwidth]{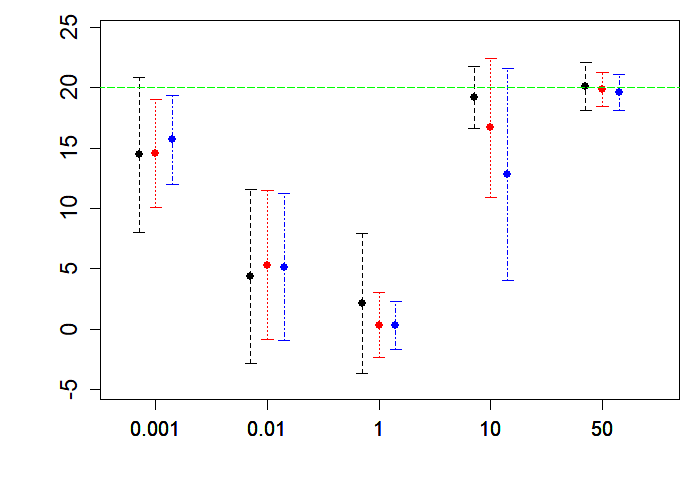}
	\caption{$E(\widehat{k}_0)\pm \sqrt{{\rm Var}(\widehat{k}_0)}$ for scaled outliers, $k_0^*=7(k^*)^{1/3}$. {\em Top:} $\modt(1/\xi)$, $k=k^*=400$. {\em Bottom:} ${\rm Burr}(1,0.5,2/\xi)$, $k=k^*=200$. {\em Left:} $k_0=3$. {\em Middle:} $k_0=10$. {\em Right:} $k_0=20$.} 
	\label{fig:xi-pos-ha-scl}
\end{figure}

Note, that as $L$ deviates from 1, the performance of the algorithm improves and the best results are obtained at $L=0.005$ and $L=10$ for all values of $k_0$. Note the improved performance of the algorithm in the region $L>1$. This is because with $L>1$, the  outliers stand out much further from the true distribution of the data. The performance of the algorithm deteriorates with increase in $k_0$. The true tail index $\xi$ has minimal effect on the algorithm. Finally, DAST works better for the $\modt$ distribution than for the  Burr distribution (observe the different scales of y-axis  in top and bottom panels of Figure \ref{fig:xi-pos-ha-exp}).

For $\modt$ and Burr distributions, Figure \ref{fig:xi-pos-ha-scl} shows the performance of the DAST  algorithm for varying intensity of the scaled outliers. Although the performance is a bit better for the scaled than the exponentiated case, conclusions are fairly similar for both cases. The superior performance under scaled outliers is because the chosen $C$ values produce more severe outliers. In this scenario, Table \ref{tab:xi-pos-ha-scl} in Appendix A.1 explores the sensitivity of the DAST algorithm to the choice of $k$ and $k^*$.

\subsection{Case $\xi=0$}
\label{sec:xi-zero}

Here, with $n=1000$, we generate data from the Gumbel domain , i.e. $\xi=0$ (see \eqref{e:xi-zero}) viz Lognormal(0,1) (denoted by LN(0,1)), $\modn(0,1)$ and Weibull(1,$\tau$) (denoted by W($\tau$)) models for $\tau=0.5,1,2$.

Where no outliers have been injected, i.e. $k_0=0$, Table \ref{tab:xi-zero-h0} gives the type 1 error, $\mathbb{P}(\widehat{k}_0>0)$ for the DAST algorithm.  Note that when the true value of $\xi$ is known, the algorithm is well calibrated at all values of $k$ . However, when $\xi$ is unknown, smaller values of $k^*$ lead to larger values of the type 1 error. This phenomenon is more pronounced at small values of $k$ especially for $\modn$(0,1) and Weibull(2) distributions. This may be explained by the fact that smaller values of $k^*$ lead to poor estimates of the initial value $\widehat{\xi}$. 
It is interesting to note that in a few cases DAST performs better with an estimated $\xi$ rather than with the correct $\xi$ (see the  Weibull(2) distribution). Numerical studies showed that the generalized trimmed Hill estimator, $GH_{k_0,k}$ (see \eqref{e:GHill}) underestimates the tail index $\xi$ for the given $\xi=0$ cases. We suspect that this biased estimator of $\xi$ perhaps facilitates easier detection of outliers.

\begin{table}[H]
	\centering
	\footnotesize
	\begin{tabular}{|c|cccc|cccc|cccc|}
		\hline
		& \multicolumn{4}{c|}{$k$=100} & \multicolumn{4}{c|}{$k$=150} & \multicolumn{4}{c|}{$k$=200}	\\
		$k^*=$ & 100 & 150 & 200 & $\xi$ known  & 100 & 150 & 200 & $\xi$ known & 100 & 150 & 200 & $\xi$ known  \\\hline
		LN(0,1)&  0.204& 0.074& 0.048 &0.036& 0.27& 0.069& 0.044& 0.034& 0.325& 0.072& 0.044 & 0.036\\\hline
		& \multicolumn{4}{c|}{$k$=100} & \multicolumn{4}{c|}{$k$=200} & \multicolumn{4}{c|}{$k$=300}	\\
		$k^*=$ & 100 & 200 & 300 & $\xi$ known  & 100 & 200 & 300 & $\xi$ known & 100 & 200 & 300 & $\xi$ known  \\\hline
		$\modn(0,1)$ & 0.502& 0.19& 0.094& 0.046& 0.574& 0.163& 0.077& 0.058& 0.579& 0.148& 0.06& 0.072\\\hline 
		& \multicolumn{4}{c|}{$k$=100} & \multicolumn{4}{c|}{$k$=150} & \multicolumn{4}{c|}{$k$=200}	\\
		$k^*=$ & 100 & 150 & 200 & $\xi$ known  & 100 & 150 & 200 & $\xi$ known & 100 & 150 & 200 & $\xi$ known  \\\hline
		W(0.5) & 0.078& 0.05& 0.046& 0.043& 0.086& 0.05& 0.049& 0.049&  0.094& 0.06& 0.058 & 0.058\\\hline
		& \multicolumn{4}{c|}{$k$=100} & \multicolumn{4}{c|}{$k$=150} & \multicolumn{4}{c|}{$k$=200}	\\
		$k^*=$ & 100 & 150 & 200 & $\xi$ known  & 100 & 150 & 200 & $\xi$ known & 100 & 150 & 200 & $\xi$ known  \\\hline
		W(1) & 0.336& 0.138& 0.074& 0.045 & 0.386& 0.127& 0.07& 0.046 & 0.425& 0.132& 0.065 & 0.052\\\hline
		& \multicolumn{4}{c|}{$k$=200} & \multicolumn{4}{c|}{$k$=400} & \multicolumn{4}{c|}{$k$=600}	\\
		$k^*=$ & 200 & 400 & 600 & $\xi$ known  & 200 & 400 & 600 & $\xi$ known & 200 & 400 & 600 & $\xi$ known  \\\hline
		W(2)	& 0.341& 0.138& 0.083& 0.054 & 0.299& 0.088& 0.057& 0.081 & 0.206& 0.068& 0.058 & 0.11\\\hline
	\end{tabular}
	\caption{Type 1 error for the DAST algorithm for $\xi=0$  with $k_0^*=7(k^*)^{1/3}$. LN denotes Lognormal(0,1) distribution and W($\tau$) denotes Weibull (1,$\tau$) distribution.}
	\label{tab:xi-zero-h0}
\end{table}


For the distribution models in \eqref{e:xi-zero}, Figures \ref{fig:xi-zero-ha-exp} and \ref{fig:xi-zero-ha-scl}  show the performance of the DAST algorithm for varying intensity of the exponentiated and scaled outliers respectively.  Tables \ref{tab:xi-zero-ha-exp} and \ref{tab:xi-zero-ha-scl} in Appendix A.1 study the sensitivity of DAST to varying choices of $k$ and $k^*$. The performance of the DAST under Gumbel domain is fairly to similar to that under  Fr\'echet domain. This is expected as the limiting distribution of the statistic ${(k /k_0)} ^{1-\xi_-}(1-T_{k_0,k})$ is same for both $\xi>0$ and $\xi=0$ (see Corollary \ref{cor:k0-k-est}). 

\begin{figure}[H]
	\hspace{-2mm}		\includegraphics[width=0.34\textwidth]{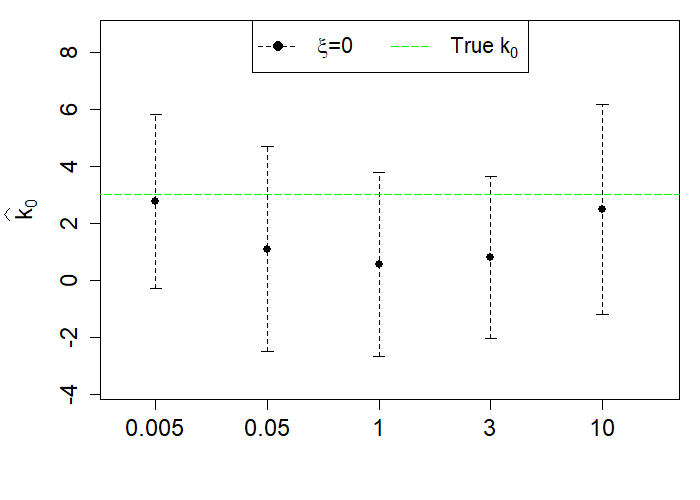}
	\hspace{-3mm}	\includegraphics[width=0.34\textwidth]{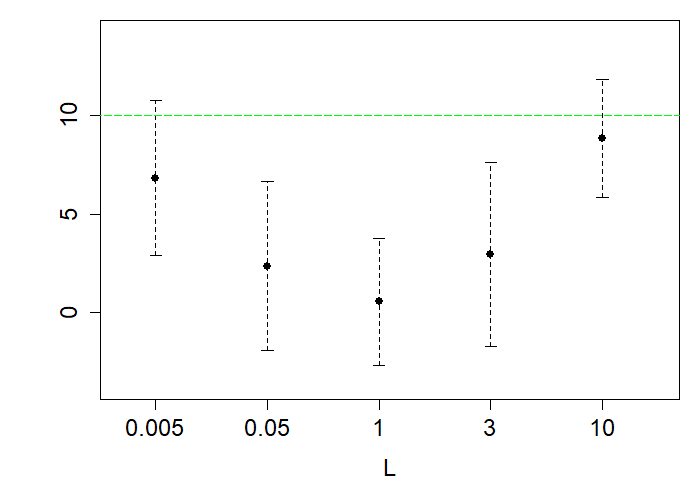}
	\hspace{-3mm}	\includegraphics[width=0.34\textwidth]{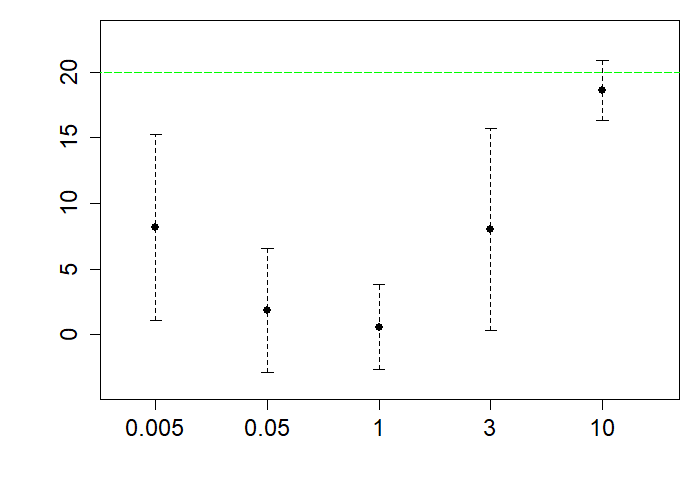}
\end{figure}
\vspace{-7mm}
\begin{figure}[H]
	\hspace{-2mm}		\includegraphics[width=0.34\textwidth]{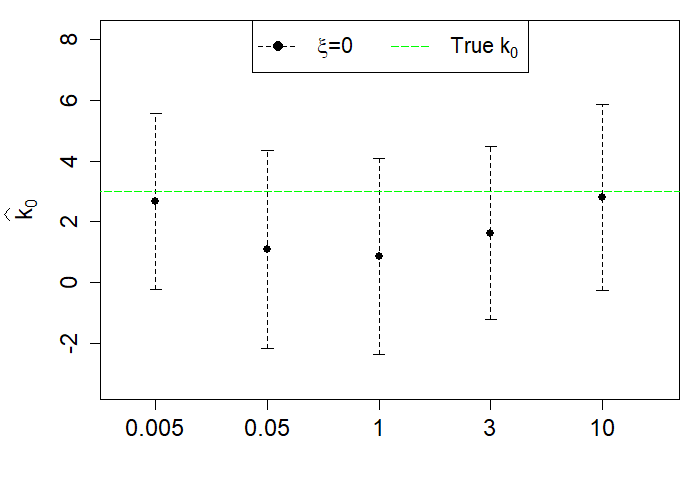}
	\hspace{-3mm}	\includegraphics[width=0.34\textwidth]{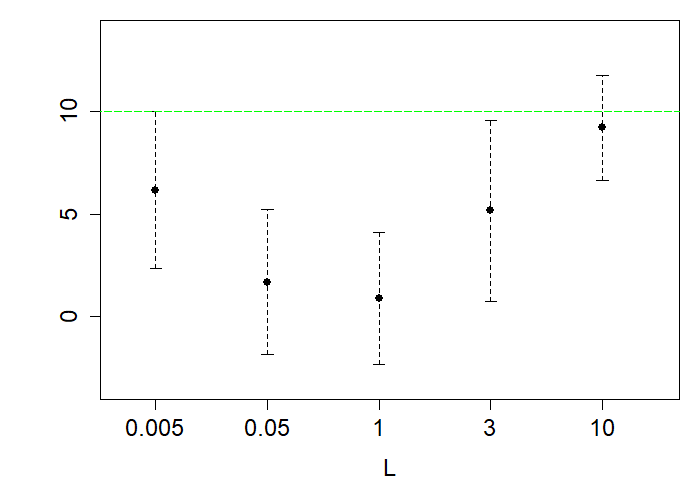}
	\hspace{-3mm}	\includegraphics[width=0.34\textwidth]{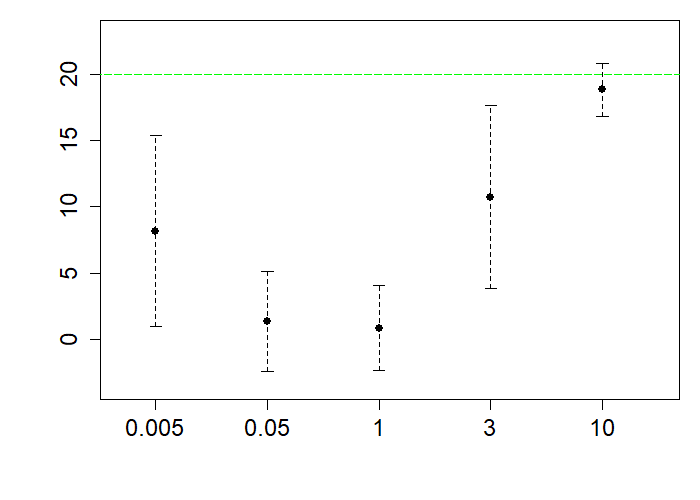}
\end{figure}
\vspace{-7mm}
\begin{figure}[H]
	\hspace{-2mm}		\includegraphics[width=0.34\textwidth]{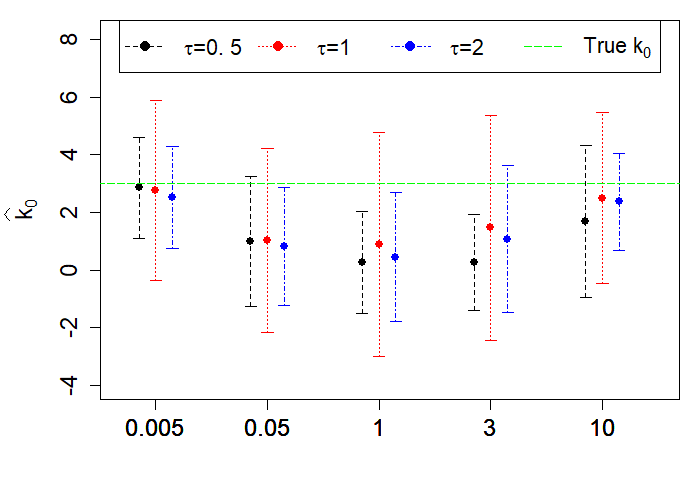}
	\hspace{-3mm}	\includegraphics[width=0.34\textwidth]{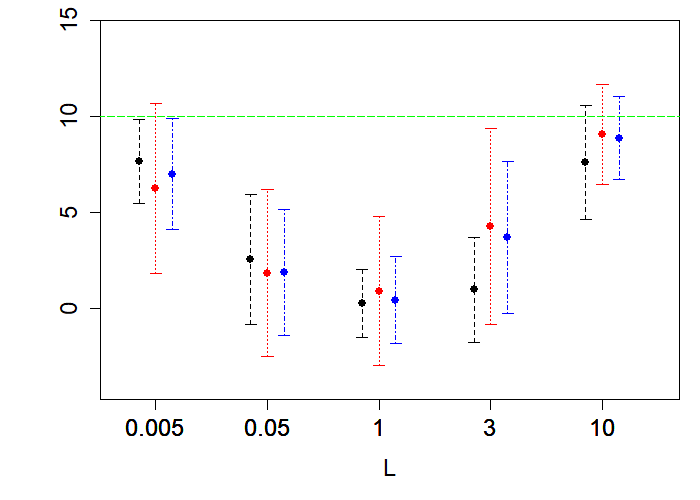}
	\hspace{-3mm}	\includegraphics[width=0.34\textwidth]{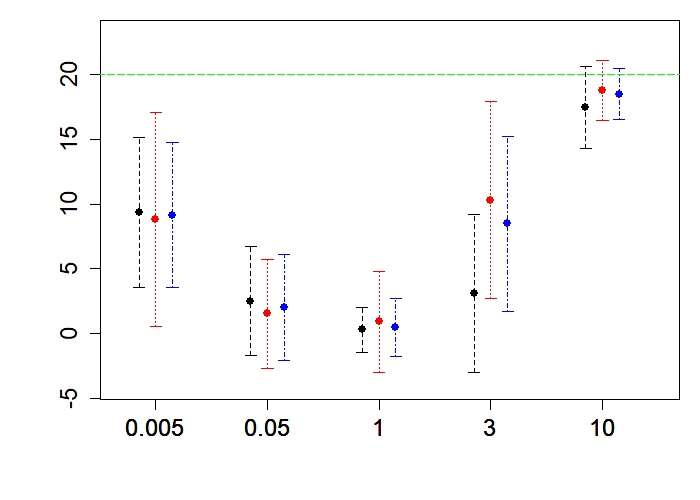}
	\caption{$E(\widehat{k}_0)\pm \sqrt{{\rm Var}(\widehat{k}_0)}$ for exponentiated outliers, $k_0^*=7(k^*)^{1/3}$. {\em Vertical:} {\em Top:} Lognormal(0,1), $k=k^*=150$. {\em Middle:}  $\modn$(0,1), $k=k^*=200$.  {\em Bottom:} Weibull(1,$\tau$) distribution with $k=k^*=150,150,200$ for $\tau=0.5,1,2$ respectively. {\em Horizontal:} {\em Left:} $k_0=3$. {\em Middle:} $k_0=10$. {\em Right:} $k_0=20$.} 
	\label{fig:xi-zero-ha-exp}
\end{figure}


\begin{figure}[H]
	\hspace{-2mm}		\includegraphics[width=0.34\textwidth]{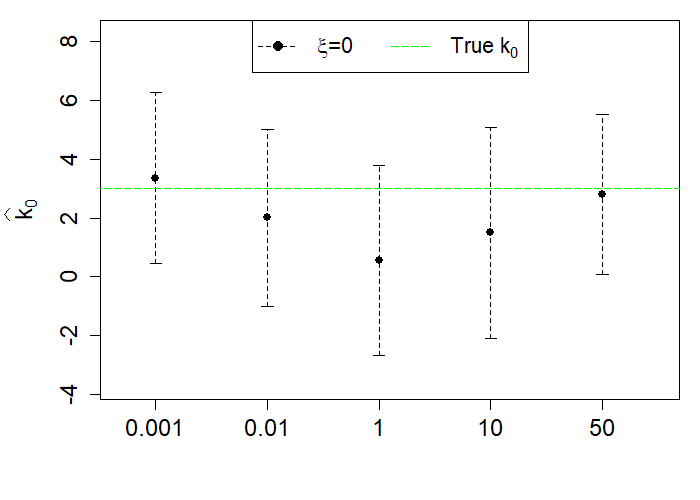}
	\hspace{-3mm}	\includegraphics[width=0.34\textwidth]{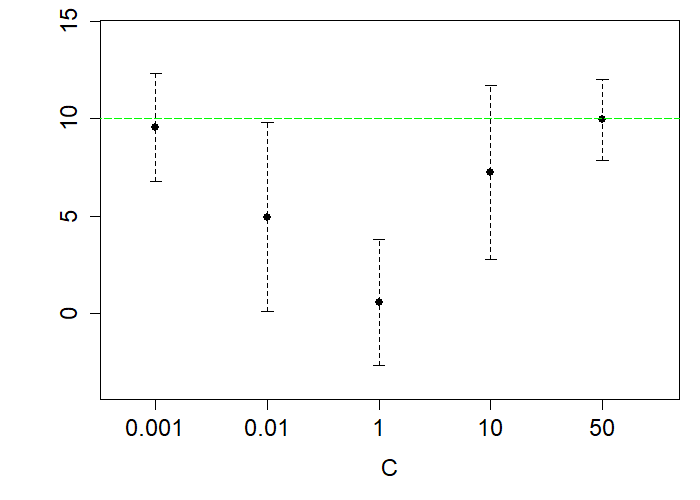}
	\hspace{-3mm}	\includegraphics[width=0.34\textwidth]{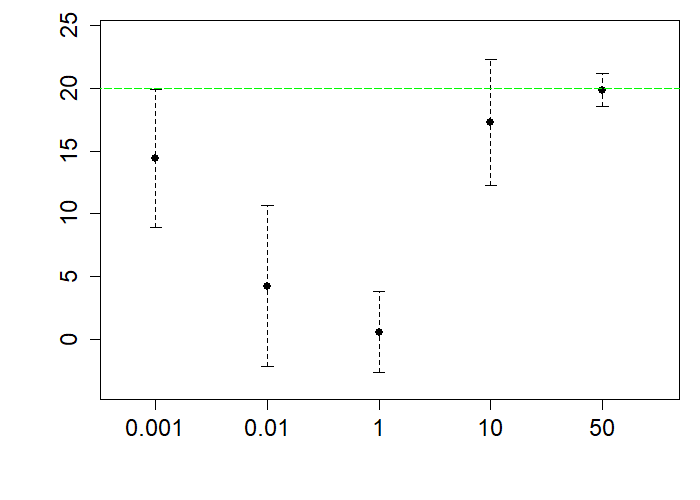}
\end{figure}
\vspace{-5mm}
\begin{figure}[H]
	\hspace{-2mm}\includegraphics[width=0.34\textwidth]{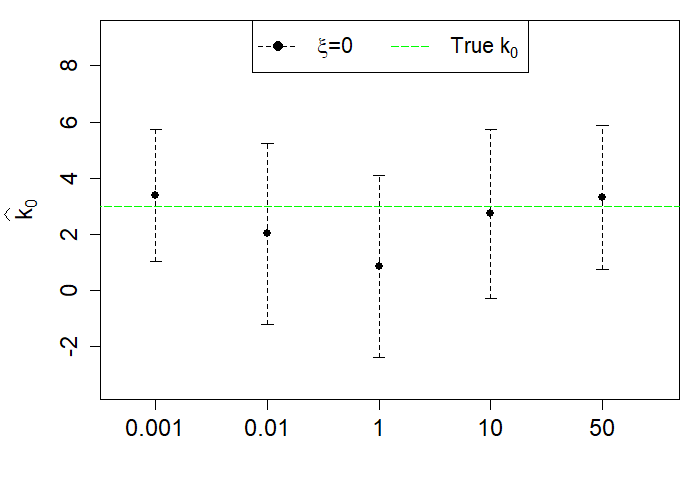}
	\hspace{-3mm}	\includegraphics[width=0.34\textwidth]{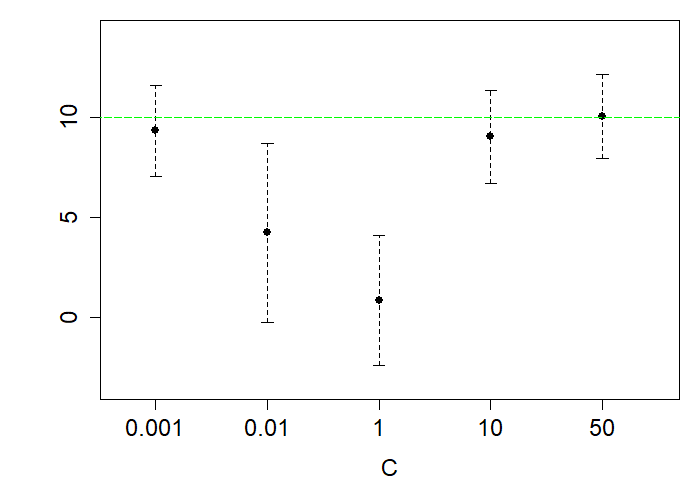}
	\hspace{-3mm}	\includegraphics[width=0.34\textwidth]{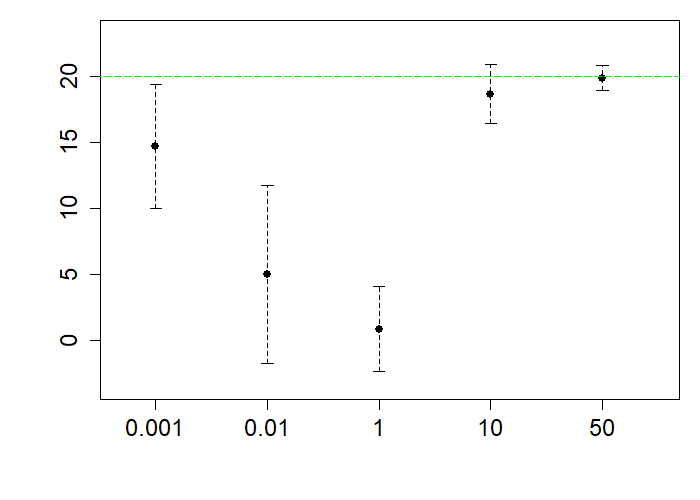}
\end{figure}
\vspace{-5mm}

\begin{figure}[H]
	\hspace{-2mm}
	\includegraphics[width=0.34\textwidth]{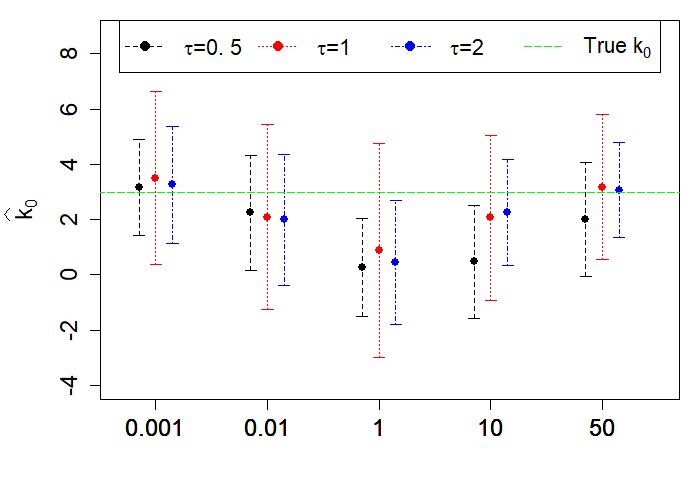}
	\hspace{-3mm}	\includegraphics[width=0.34\textwidth]{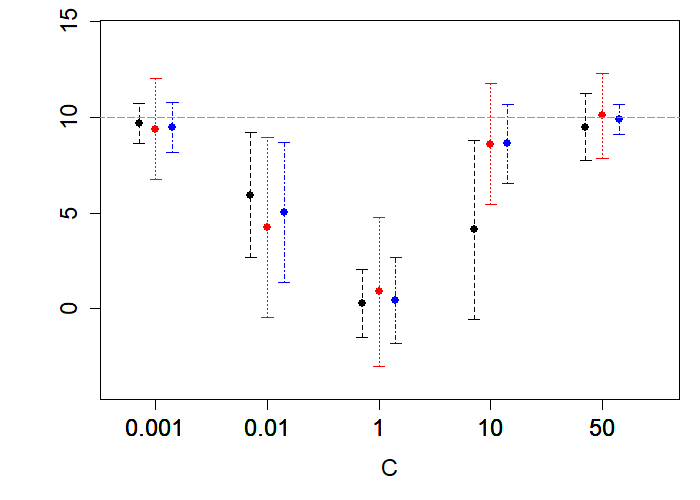}
	\hspace{-3mm}	\includegraphics[width=0.34\textwidth]{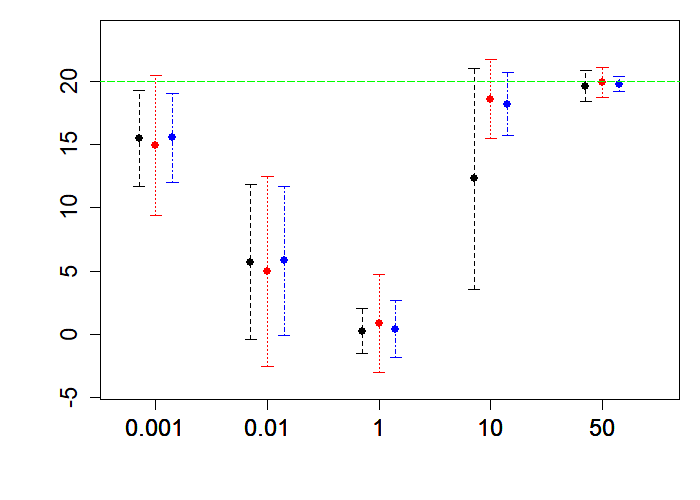}
	\caption{$E(\widehat{k}_0)\pm \sqrt{{\rm Var}(\widehat{k}_0)}$ for scaled outliers, $k_0^*=7(k^*)^{1/3}$.  {\em Vertical:} {\em Top:} Lognormal(0,1), $k=k^*=150$. {\em Middle:}  $\modn$(0,1), $k=k^*=200$.  {\em Bottom:} Weibull(1,$\tau$) distribution with $k=k^*=150,150,200$ for $\tau=0.5,1,2$ respectively. {\em Horizontal:} {\em Left:} $k_0=3$. {\em Middle:} $k_0=10$. {\em Right:} $k_0=20$.} 
	\label{fig:xi-zero-ha-scl}
\end{figure}

\subsection{Case $\xi<0$}

\label{sec:xi-neg}


Here, with $n=1000$, we generate data from Weibull domain, i.e. $\xi<0$ (see \eqref{e:xi-neg}) viz Beta(1,-1/$\xi$) and Reverse Burr(1,0.5,-2/$\xi$). The type 1 error, $\mathbb{P}(\widehat{k}_0>0)$ when no outliers have been injected, i.e. $k_0=0$, is  given in Table \ref{tab:xi-neg-h0}. Note that  for $\xi=-0.25,-0.5$, the algorithm is well calibrated at all values of $k$ when the true value of $\xi$ is used as the initial estimate. When $\xi$ is unknown, smaller values of $k^*$ lead to larger  type 1 errors, especially when $k$ is taken small. For $\xi=-1$, the algorithm is not well calibrated even if the true value of $\xi$ is used as the initial estimate. However, $E(\widehat{k}_0)$ is quite close to zero (see Figure \ref{fig:xi-neg-ha-exp} at $L=1$) which implies that the number of wrongly detected outliers is quite small. at $\xi=-1$ a careful analysis revealed that some of the top order statistics come very close to each other, thereby inflating the statistic $T_{k_0,k}$ (see \eqref{e:stat-trim}) at $k_0\approx 0$, causing DAST to produce false positives.

\begin{table}[H]
	\centering
	\footnotesize
	\begin{tabular}{|c|cccc|cccc|cccc|}
		\hline
		& \multicolumn{4}{c|}{$k$=100} & \multicolumn{4}{c|}{$k$=200} & \multicolumn{4}{c|}{$k$=300}	\\
		$k^*=$ & 100 & 200 & 300 & $\xi$ known  & 100 & 200 & 300 & $\xi$ known & 100 & 200 & 300 & $\xi$ known  \\\hline
		$\xi$=-0.25 & 0.397& 0.122& 0.063& 0.041& 0.453& 0.09& 0.055& 0.039& 0.457& 0.09& 0.063 & 0.045\\
		$\xi$=	-0.5 & 0.44& 0.212& 0.108& 0.117& 0.454& 0.148& 0.078& 0.078& 0.436& 0.133& 0.068 &0.062\\
		$\xi$=	-1 	& 0.512& 0.444& 0.388&  0.544& 0.471& 0.392& 0.31& 0.498&  0.431& 0.343& 0.254 & 0.44\\\hline 
	\end{tabular}	
\end{table}

\vspace{-5mm}
\begin{table}[H]
	\centering
	\footnotesize
	\begin{tabular}{|c|cccc|cccc|cccc|}
		\hline
		& \multicolumn{4}{c|}{$k$=200} & \multicolumn{4}{c|}{$k$=400} & \multicolumn{4}{c|}{$k$=600}	\\
		$k^*=$ & 200 & 400 & 600 & $\xi$ known  & 200 & 400 & 600 & $\xi$ known & 200 & 400 & 600 & $\xi$ known  \\\hline
		$\xi$=	-0.25 & 0.165& 0.057& 0.058& 0.042&0.123& 0.063& 0.086& 0.056&  0.096& 0.091& 0.127 & 0.074\\ 
		$\xi$=	-0.5 & 0.273& 0.138& 0.06&  0.067&  0.187& 0.073& 0.065&  0.047& 0.121& 0.07& 0.089 & 0.063\\
		$\xi$=	-1 &  0.578& 0.64& 0.606&  0.458&0.447& 0.501& 0.471& 0.312& 0.308 & 0.345&  0.182 & 0.182\\\hline 
	\end{tabular}	
	\caption{Type 1 error for the DAST algorithm for $\xi<0$  with $k_0^*=7(k^*)^{1/3}$. {\em Top:} Beta$(1,1/\xi)$ distribution. {\em Bottom:} RBurr(1,0.5,2/$\xi$) distribution.}
	\label{tab:xi-neg-h0}
\end{table}

\begin{figure}[H]
	\hspace{-2mm}\includegraphics[width=0.34\textwidth]{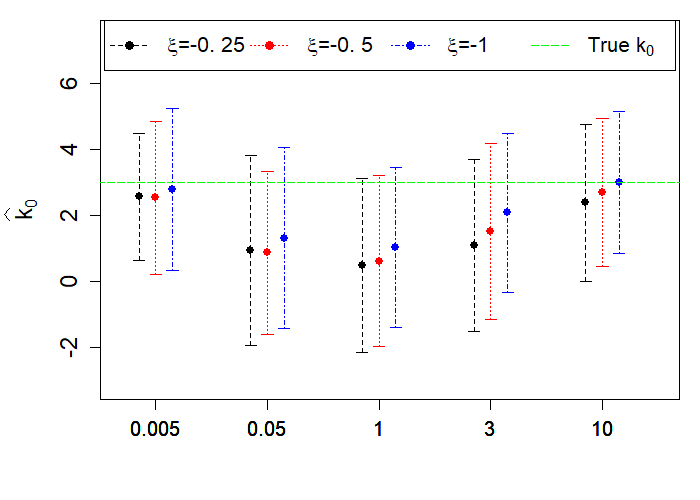}
	\hspace{-3mm}	\includegraphics[width=0.34\textwidth]{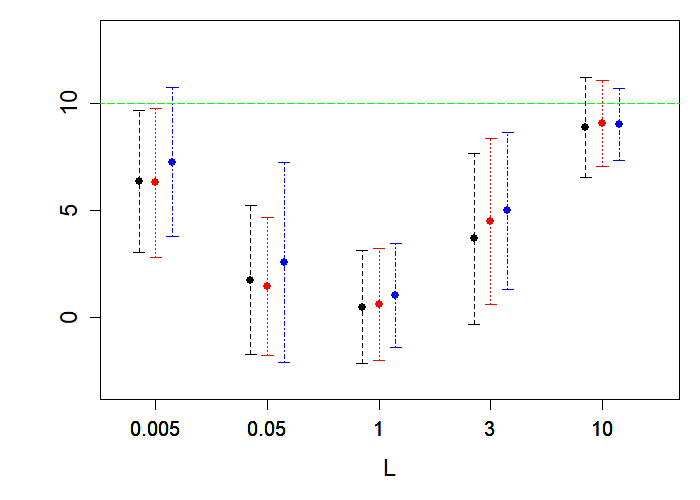}
	\hspace{-3mm}	\includegraphics[width=0.34\textwidth]{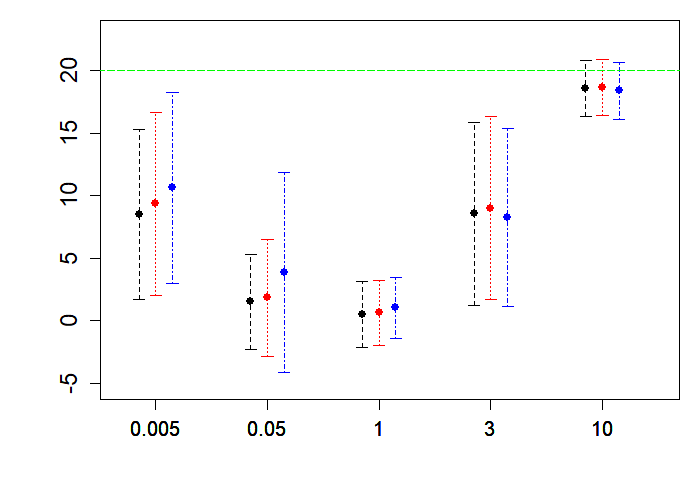}
\end{figure}

\vspace{-5mm}

\begin{figure}[H]
	\hspace{-2mm}\includegraphics[width=0.34\textwidth]{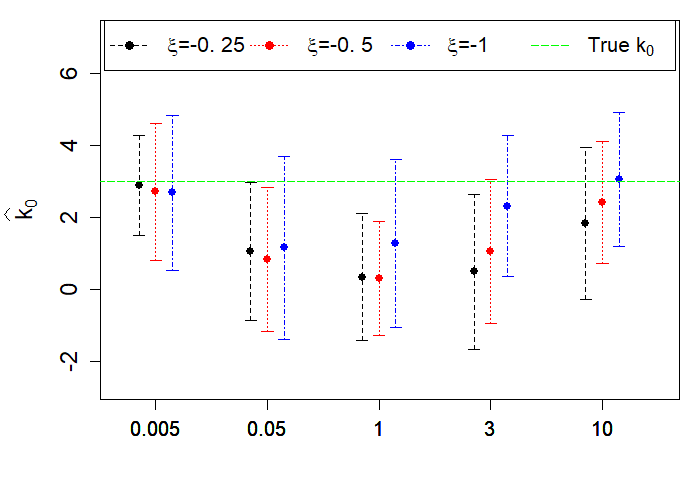}
	\hspace{-3mm}	\includegraphics[width=0.34\textwidth]{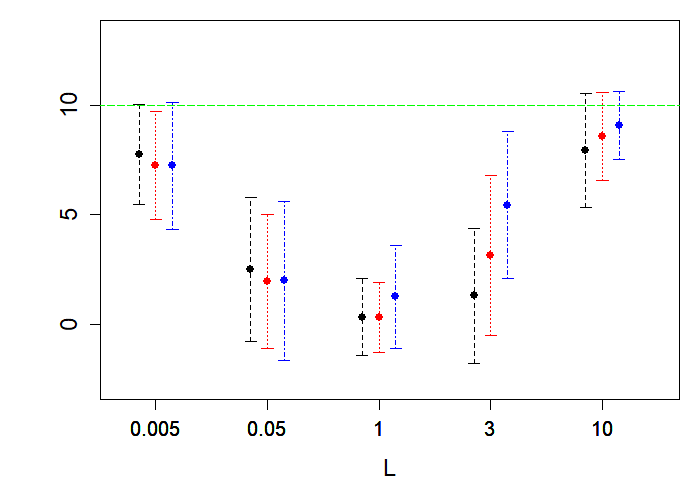}
	\hspace{-3mm}	\includegraphics[width=0.34\textwidth]{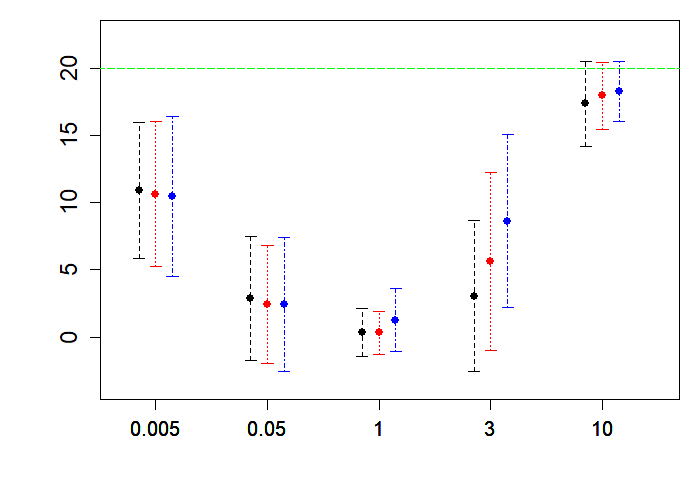}
	\caption{$E(\widehat{k}_0)\pm \sqrt{{\rm Var}(\widehat{k}_0)}$ for exponentiated outliers, $k_0^*=7(k^*)^{1/3}$. {\em Top:} Beta(1,-1/$\xi$) distribution. {\em Bottom:}  Reverse Burr(1,0.5,$-2/\xi$) distribution. {\em Left:} $k_0=3$. {\em Middle:} $k_0=10$. {\em Right:} $k_0=20$.} 
	\label{fig:xi-neg-ha-exp}
\end{figure}

For the distribution models in \eqref{e:xi-neg}, Figures \ref{fig:xi-neg-ha-exp} and \ref{fig:xi-neg-ha-scl} show the performance of the DAST for varying intensity of exponentiated and scaled outliers respectively. Tables \ref{tab:xi-neg-ha-exp} and \ref{tab:xi-neg-ha-scl} in Appendix A.1  exhibit the sensitivity of the algorithm to varying choices of $k$ and $k^*$. The performance of DAST improves on both sides of $L=1$ (case of no outliers) and has greater accuracy at smaller values of $k_0$. Again, the conclusions are fairly similar to those obtained in the $\xi\geq 0$ which suggests that the DAST can adapt itself easily to changing domains of attraction. 
Since the performance of the DAST for $\xi\leq 0$ matches that of $\xi>0$ (see Table \ref{tab:xi-pos-ha-exp}, \ref{tab:xi-zero-ha-exp} and \ref{tab:xi-neg-ha-exp}) the proposed algorithm appears to be quite ubiquitous.

\begin{figure}[H]
	\hspace{-2mm}\includegraphics[width=0.34\textwidth]{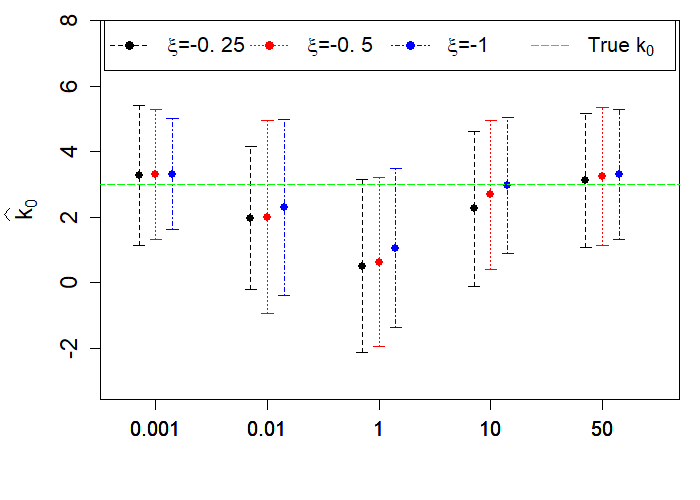}
	\hspace{-3mm}	\includegraphics[width=0.34\textwidth]{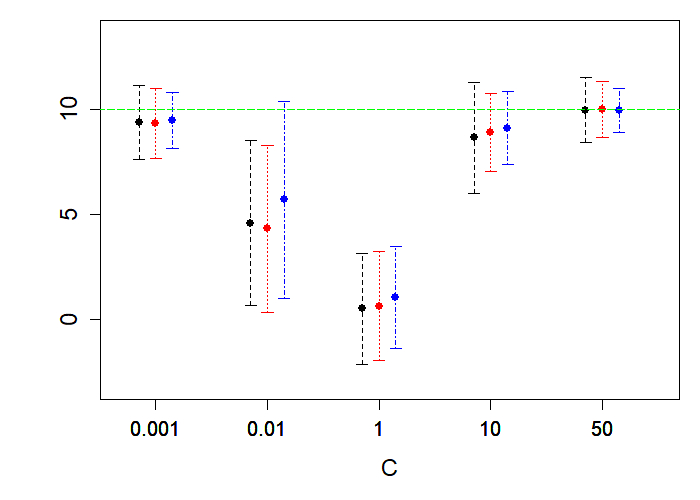}
	\hspace{-3mm}	\includegraphics[width=0.34\textwidth]{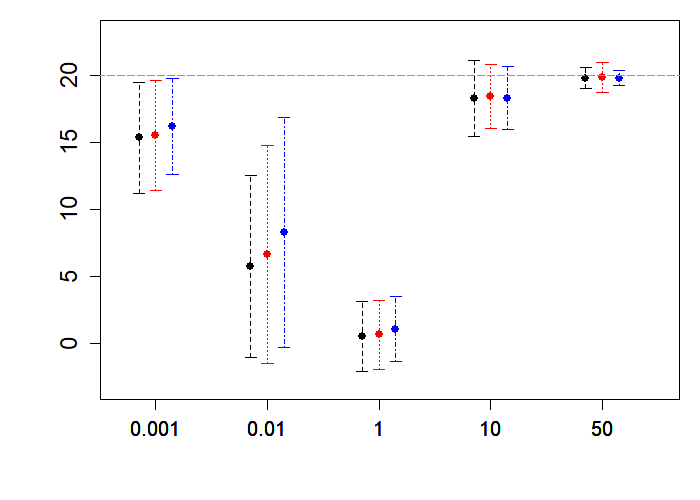}
\end{figure}

\vspace{-5mm}
\begin{figure}[H]
	\hspace{-2mm}\includegraphics[width=0.34\textwidth]{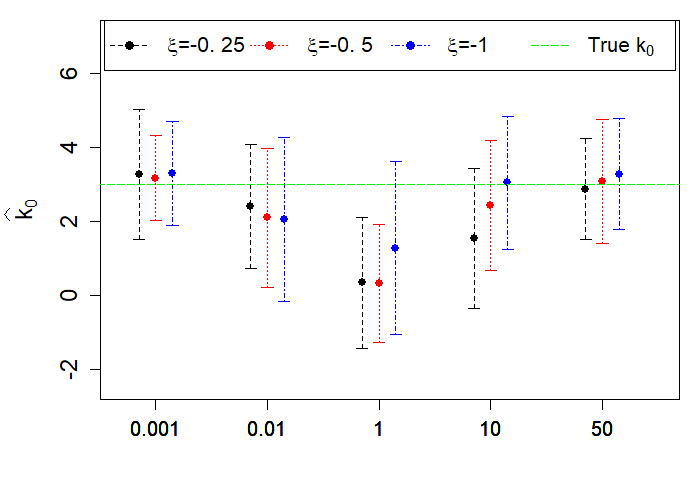}
	\hspace{-3mm}	\includegraphics[width=0.34\textwidth]{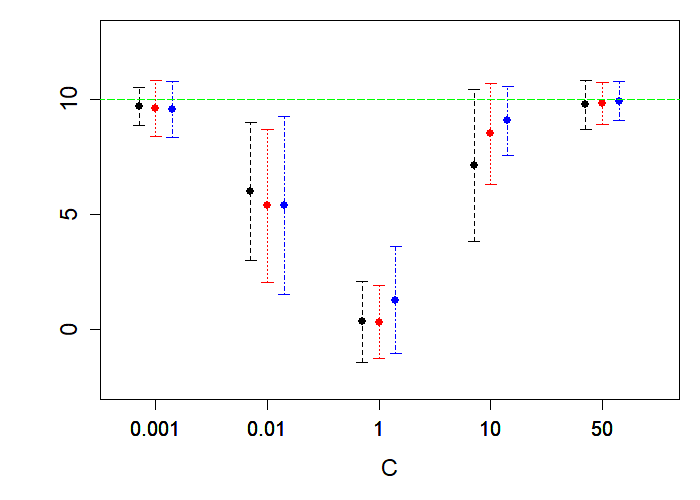}
	\hspace{-3mm}	\includegraphics[width=0.34\textwidth]{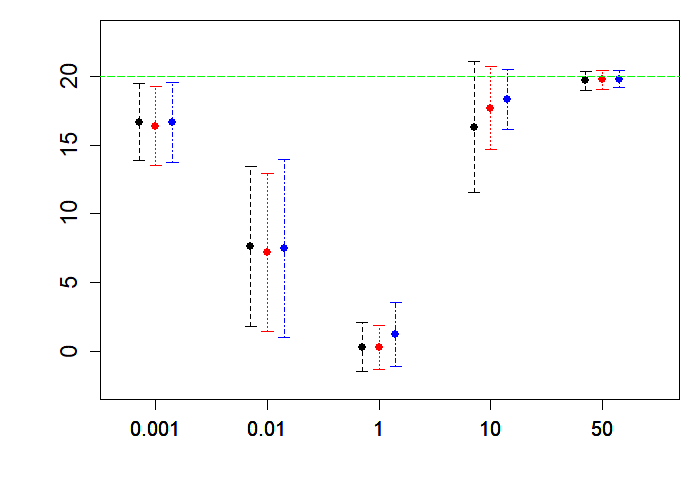}
	\caption{$E(\widehat{k}_0)\pm \sqrt{{\rm Var}(\widehat{k}_0)}$ for scaled outliers, $k_0^*=7(k^*)^{1/3}$. {\em Top:} Beta(1,-1/$\xi$), $k=k^*=200$. {\em Bottom:}  Reverse Burr(1,0.5,$-2/\xi$), $k=k^*=400$. {\em Left:} $k_0=3$. {\em Middle:} $k_0=10$. {\em Right:} $k_0=20$.} 
	\label{fig:xi-neg-ha-scl}
\end{figure}


	\section{Case Studies}
\label{sec:case}
In this section, we apply the DAST algorithm for detection of outliers for some cases that exhibit some deviating data at one or both tails, and that appear to belong to different max-domains. The parameters $a$ and $q$ are set at 1.2 and 0.05. For breaking ties, the data are dithered by adding a small uniform noise from the uniform $U(-0.01,0.01)$ distribution. In each example we take $k=k^*$.  In practice, the choice of $k$ and $k^*$ in the neighborhood of $k^{\rm opt}$ as discussed/used in the section \ref{sec:sim} (see \eqref{e:k-star-def}) is not directly applicable given that $k^{\rm opt}$ is not known. We propose  to consider appropriate  QQ-plots or the generalized QQ-plot in general,  and choose $k=k^*$ around the point where a stable tail behavior starts to kick in. This is illustrated in this section in Figures \ref{fig:freclaim-1} left, \ref{fig:condroz-1} left, \ref{fig:toxic-1} right,  and \ref{fig:air} right with linear fits based solely on the top $k=k^*$ observations, as used in the DAST algorithm.

\subsection{French precipitation data }


\begin{figure}[H]
	\centering
	\includegraphics[width=0.32\textwidth]{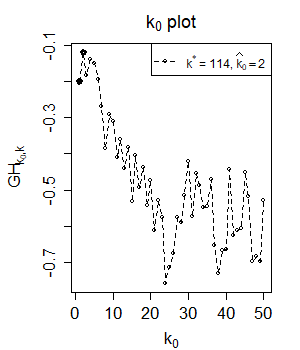}
	\includegraphics[width=0.32\textwidth]{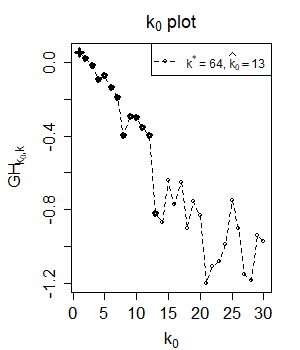}
	\includegraphics[width=0.32\textwidth]{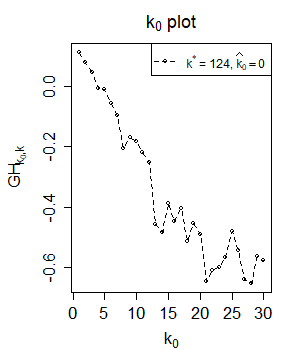}
	\caption{{\em French precipitation} data. {\em Left:} Diagnostic $k_0$ plot for Chamonix with $k=k^*=114$.  {\em Middle:}  Diagnostic $k_0$ plot of Uzein station with  $k=k^*=64$. {\em Right:}  Diagnostic $k_0$ plot of Uzein station with  $k=k^*=124$.  }
	\label{fig:precip-case}
\end{figure}

Concerning  {\em the French precipitation} data discussed  in section \ref{sec:intro}, we  here show the diagnostic $k_0$ plots  for both the Chamonix and Uzein stations.  For the Chamonix station the top two observations exhibit a downward trend in the diagnostic $k_0$ plot in comparison with the remaining points.  Concentrating on the top outliers only, for the Uzein station with $V=2$, $k_0^*=20$ and  $k=k^*=64$ we obtain 13 outliers with the largest data value being indicated separately.  For $V=2$, $k_0^*=20$, if one chooses $k=k^*=124$, no outliers are obtained. The sensitivity to the choice of $k$ can be explained by the Pareto QQ-plot of Figure \ref{fig:precip_2_intro} which shows that Pareto behavior sets in around $k=k^* \leq 64$.


\subsection{Toxicity data set}

The {\em Toxicity} data set  from \url{https://archive.ics.uci.edu/ml/datasets/QSAR+fish+toxicity} was used to develop regression models for the prediction of acute aquatic toxicity towards the Pimephales promelas. We concentrate on the LC50 count, a chemical responsible for 50\% of the deaths in the fish population  (see \cite{qsar}). Note that the classical boxplot identifies top 11 outliers whereas the tail-adjusted boxplot identifies none. This example shows an EVI $\xi \leq 0$, which is confirmed by  the generalized  QQ-plot in the right panel in Figure \ref{fig:toxic-1},  which is non-increasing at the right hand side.

\begin{figure}[H]
	\centering
	\includegraphics[width=0.32\textwidth]{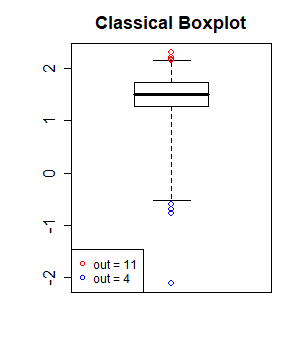}
	\includegraphics[width=0.32\textwidth]{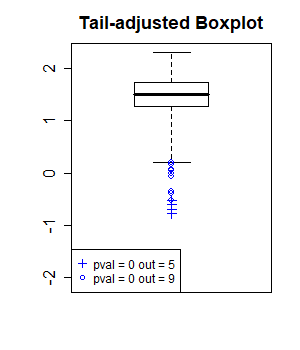}
	\includegraphics[width=0.32\textwidth]{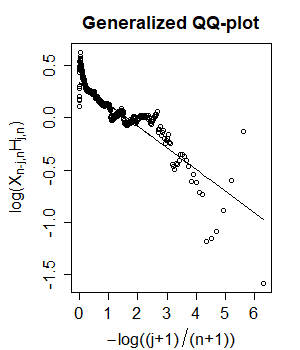}
	\caption{ {\em Toxicity} data. {\em Left:} Classical boxplot.  {\em Right:} Tail-adjusted boxplot.  {\em Right:}  Generalized QQ-Plot.}
	\label{fig:toxic}
\end{figure}

\vspace{-5mm}

\begin{figure}[H]
	\centering
	
	\includegraphics[width=0.32\textwidth]{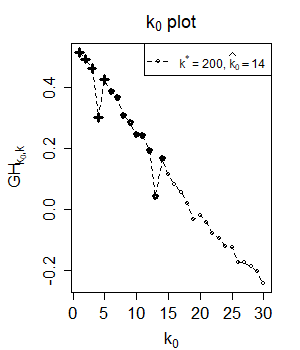}
	\includegraphics[width=0.32\textwidth]{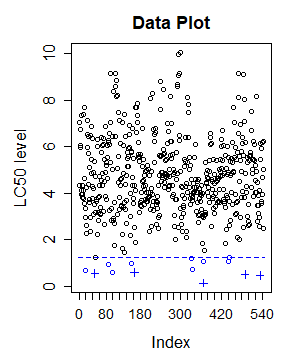}
	\caption{ {\em Toxicity} data set.  {\em Left:}  Diagnostic $k_0$ plot for the left tail.  {\em Right:} Time plot. The extreme outliers and moderate outliers are marked with + and $\circ$ respectively.}
	\label{fig:toxic-1}
\end{figure}

However, a more interesting phenomenon occurs in the lower tails where the classical boxplot identifies 4 bottom outliers in contrast to 14 bottom outliers of the tail-adjusted boxplot. Using the DAST algorithm for the bottom tail  with $V=2$, $k_0^*=30$, $k=k^*=200$, the 14 bottom outliers are split into two regimes with 5 extreme and 9 moderate outliers. Indeed, the diagnostic $k_0$ plot for the bottom tail in Figure \ref{fig:toxic-1} shows that there are two change points at $k_0=5$ and $k_0=14$ respectively.  The right panel of Figure \ref{fig:toxic-1} displays the time plot which indicates these 14 outliers for the bottom tail.

\subsection{French fire claims data set}

\begin{figure}[H]
	\centering	\includegraphics[width=0.32\textwidth]{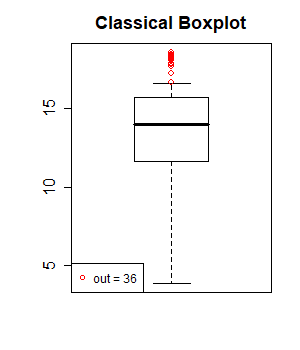}	\includegraphics[width=0.32\textwidth]{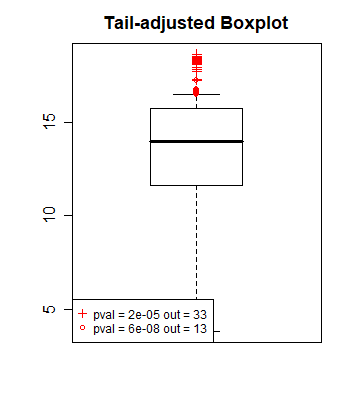}
	
	\caption{{\em Fire claim} data set.  {\em Left:} Classical boxplot.  {\em Right:} Tail-adjusted boxplot.}
	\label{fig:freclaim}
\end{figure}

\begin{figure}[H]
	\centering
	\includegraphics[width=0.32\textwidth]{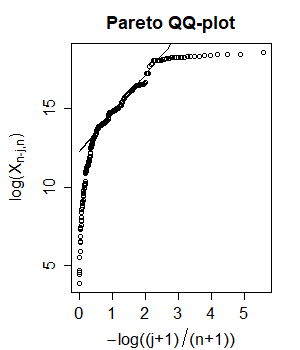}
	\includegraphics[width=0.32\textwidth]{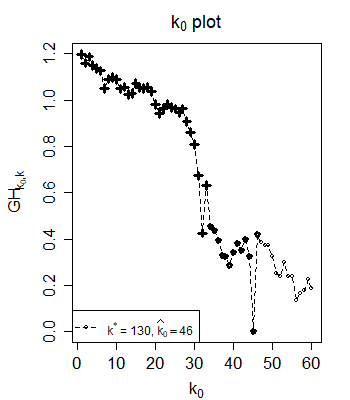}
	\includegraphics[width=0.32\textwidth]{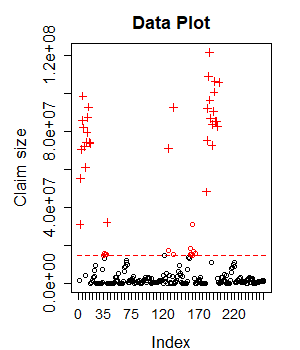}
	\caption{{\em Fire claim} data set. {\em Left:}  Pareto QQ-plot. {\em Middle:}  Diagnostic $k_0$ plot for the right hand tail. {\em Right:} Data plot. The extreme outliers and moderate outliers are marked with + and $\circ$.}
	\label{fig:freclaim-1}
\end{figure}

The {\em Fire Claim} data set involves $n= 261$ claim settlements issued by a private insurer in France during the time period 1996-2006 available from  \url{http://cas.uqam.ca/pub/R/web/
	CASdatasets-manual.pdf}. This data set was already analyzed for outliers in \cite{bhatt2019}. We here concentrate on the right tail only. The Pareto QQ-plot in Figure \ref{fig:freclaim-1}  has  an apparent linear trend, up to a top group of data which exhibit less spread than the data below.  Hence a different regime is present in the top data. This often appears in non-life insurance claim data due to tightened claim inspection and management with extreme claims. The boxplots are given on the log-scale in Figure \ref{fig:freclaim}, and the classical boxplot shows 36 top outliers.

For identifying the top outliers, the parameters of the algorithm  are chosen as  $V=2$, $k_0^*=40$, $k=k^*=130$. Then the algorithm returns a set of $\widehat{k}_0=33$ most extreme outliers and another group of $13$ outliers which can also be noticed from the diagnostic $k_0$ plot with break points around positions 33 and  46. Also the Pareto QQ-plot shows some intermediate data which deviate from the linear pattern below the level $\log x= 16.5$ on the vertical scale. The top 33 outliers were already detected in \cite{bhatt2019}.

\subsection{Condroz data set}

\begin{figure}[H]
	\centering
	\includegraphics[width=0.32\textwidth]{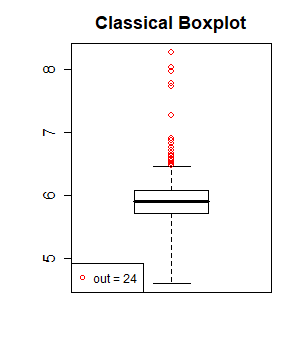}
	\includegraphics[width=0.32\textwidth]{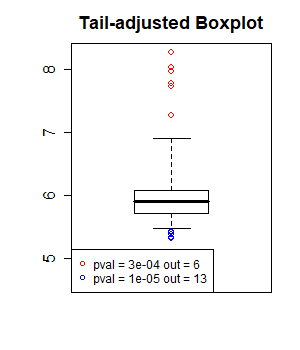}
	\caption{ {\em Condroz} data set.  {\em Left:} Classical boxplot.  {\em Right:} Tail-adjusted boxplot.}
	\label{fig:condroz}
\end{figure}

The {\em Condroz} data set  with calcium content measurements together with the pH level of soil samples in the Condroz region of Belgium was discussed in detail in \cite{goeg2005}, and has been analyzed for outliers in  \cite{beir1996}, \cite{VANDEWALLE2007}, \cite{hubert2008} and \cite{bhatt2019}. As in these references we consider the conditional distribution of the calcium content for pH levels lying between 7-7.5 leading to $n= 420$ data values. All authors put this example in the 
Fr\'echet domain ($\xi >0$), which is confirmed by  the Pareto QQ-plot in 
the left panel in Figure \ref{fig:condroz-1}  overall approximately exhibiting a linear pattern (see  \cite{beir1996}) except for the top 6 values  that jump out. These outliers appeared to be measurements from communities at the boundary of the Condroz region and hence can be considered to be sampled from another distribution (cfr. \cite{hubert2008}).
In \cite{hubert2008} the adjusted boxplot based on robust skewness measurement shows 12 outliers.

\begin{figure}[H]
	\centering
	\includegraphics[width=0.32\textwidth]{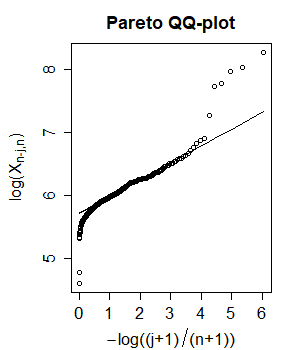}
	\includegraphics[width=0.32\textwidth]{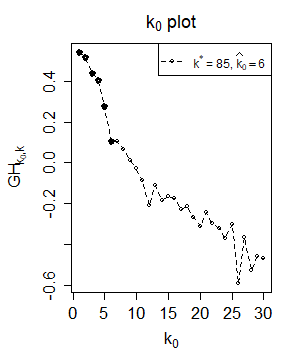}
	\includegraphics[width=0.32\textwidth]{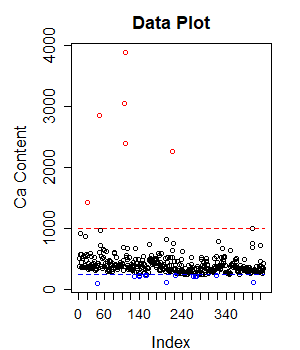}
	\caption{{\em Condroz} data set.  {\em Left:}  Pareto QQ-Plot. {\em Middle:}  Diagnostic $k_0$ plot for the right hand tail. {\em Right:} Data plot.}
	\label{fig:condroz-1}
\end{figure}

The classical and tail-adjusted boxplots are given in  Figure \ref{fig:condroz}  on the log scale. For identifying both top and bottom outliers, the parameters of the algorithm  are chosen as $V=1$, $k_0^*=30$, $k=k^*=85$. As expected for a heavy tailed distribution, the classical boxplot indicates  a high number of 24 outliers, while the tail-adjusted boxplot shows the  6 outliers corresponding with a visual inspection of the Pareto QQ-plot.   This is also in consensus with the findings of \cite{bhatt2019} where the problem of outlier identification in the heavy tailed regime ($\xi>0$) was already discussed. Additionally,  13 left tail outliers are identified in the tail-adjusted boxplot in contrast to the classical boxplot which identifies none. The middle panel of Figure \ref{fig:condroz-1} contains the diagnostic $k_0$ plot for the right tail  which shows a change point around the point $k_0 = 6$ for different values of $k$. The data plot is given in the right panel in Figure \ref{fig:condroz-1} with indication of  the identified outliers in the upper and lower tails.

\subsection{Air data set}
The  {\em Air Quality} data set obtained from the New York State Department of Conservation (ozone data) and the National Weather Service (meteorological data) is available at \url{https://stat.ethz.ch/R-manual/R-devel/library/datasets/html/airquality.html}. It contains  wind speeds (in miles per hour) for  New York, May to September 1973, see  \cite{chambers1992}. 
For this case the generalized QQ-plot shows a clear downward trend leading to $\xi <0$ for the right tail. Here also the boxplots are given on the log scale.  For identifying both top and bottom outliers, the parameters of the DAST algorithm  are chosen as $V=1$, $k=76$, $k^*=76$, $k_0^*=25$. Both the classical boxplot and tail-adjusted boxplot report 3 top outliers.  In \cite{hubert2008}  also 3 top outliers were found using a measure of skewness. The tail-adjusted boxplot however identifies 24 bottom outliers, none of which were detected by the classical boxplot.

\begin{figure}[H]
	\centering
	\includegraphics[width=0.32\textwidth]{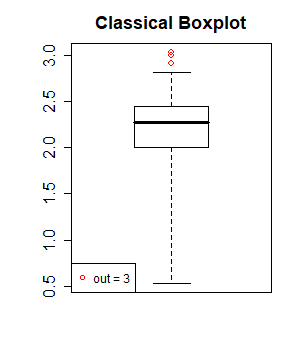}
	\includegraphics[width=0.32\textwidth]{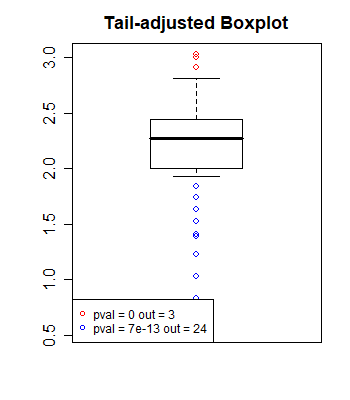}
	\includegraphics[width=0.32\textwidth]{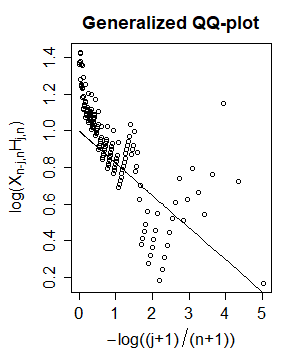}
	\caption{{\em Air quality} data set.  {\em Left:} Classical boxplot.  {\em Middle:} Tail-adjusted boxplot. {\em Right:} Generalized QQ-plot.}
	\label{fig:air}
\end{figure}


\section{Conclusion}

In this paper we provided a testing procedure for outlier detection based on extreme value methodology.
The test statistic is based on the deviations of trimmed Hill statistics when trimming consecutive extreme data points. While the Hill estimator is only a consistent estimator in case of a positive extreme value index, we show that this statistic is still useful for outlier detection in all max-domains of attraction. As a practical consequence a tail-adjusted boxplot is proposed, allowing to indicate possible outliers depending on the tail heaviness of the underlying distribution.

\section*{Appendix}
\subsection*{A.1. Further Simulation Results}
We report $E(\widehat{k}_0)\pm \sqrt{{\rm Var}(\widehat{k}_0)}$ as obtained by the DAST. The parameters are set at $V=1$, $a=1.2$ and $q=0.05$. Top row corresponds to $\xi$ known and bottom row corresponds to $\xi$ estimated according to section \ref{sec:aut-trim} for varying $k$ and $k^*$ and $k_0^*=7(k^*)^{1/3}$.

\begin{table}[H]
	\centering
	\scriptsize
	\begin{tabular}{|c|ccc|ccc|ccc|}
		\hline
		& \multicolumn{3}{c|}{$k$=200} & \multicolumn{3}{c|}{$k$=400} & \multicolumn{3}{c|}{$k$=600}	\\
		$L$ & $k^*=$200 & $k^*=$400 & $k^*=$600 &	$k^*=$200 & $k^*=$400 & $k^*=$600&$k^*=$200 & $k^*=$400 & $k^*=$600\\\hline
		0.005 & 7.1 $\pm$ 2.5& 7.1 $\pm$ 2.5& 7.1 $\pm$ 2.5& 7.5 $\pm$ 2.2& 7.5 $\pm$ 2.2& 7.5 $\pm$ 2.2& 8.1 $\pm$ 1.9& 8.1 $\pm$ 1.9& 8.1 $\pm$ 1.9\\ 
		& 7 $\pm$ 3& 7.1 $\pm$ 2.5& 7.1 $\pm$ 2.5& 7.5 $\pm$ 2.6& 7.5 $\pm$ 2.2& 7.5 $\pm$ 2.2& 8 $\pm$ 2& 8.1 $\pm$ 1.9& 8.1 $\pm$ 1.9\\\hline 
		0.05 & 1.8 $\pm$ 3& 1.8 $\pm$ 3& 1.8 $\pm$ 3& 2.2 $\pm$ 3.1& 2.2 $\pm$ 3.1& 2.2 $\pm$ 3.1& 2.8 $\pm$ 3.4& 2.8 $\pm$ 3.4& 2.8 $\pm$ 3.4\\ 
		& 2 $\pm$ 3.6& 1.9 $\pm$ 3.1& 1.8 $\pm$ 3& 2.2 $\pm$ 3.2& 2.2 $\pm$ 3.1& 2.2 $\pm$ 3.1& 2.7 $\pm$ 3.4& 2.8 $\pm$ 3.4& 2.8 $\pm$ 3.4\\\hline 
		1 & 0.2 $\pm$ 1.5& 0.2 $\pm$ 1.5& 0.2 $\pm$ 1.5& 0.2 $\pm$ 1.2& 0.2 $\pm$ 1.2& 0.2 $\pm$ 1.2& 0.2 $\pm$ 1.5& 0.2 $\pm$ 1.5& 0.2 $\pm$ 1.5\\ 
		& 0.3 $\pm$ 2.2& 0.2 $\pm$ 1.5& 0.2 $\pm$ 1.5& 0.2 $\pm$ 1.4& 0.2 $\pm$ 1.2& 0.2 $\pm$ 1.2& 0.2 $\pm$ 1.5& 0.2 $\pm$ 1.5& 0.2 $\pm$ 1.5\\\hline 
		3 & 3.3 $\pm$ 3.8& 3.3 $\pm$ 3.8& 3.3 $\pm$ 3.8& 2.2 $\pm$ 3.3& 2.2 $\pm$ 3.3& 2.2 $\pm$ 3.3& 1.1 $\pm$ 2.6& 1.1 $\pm$ 2.6& 1.1 $\pm$ 2.6\\ 
		& 3.5 $\pm$ 4.1& 3.4 $\pm$ 3.8& 3.3 $\pm$ 3.8& 2.3 $\pm$ 3.5& 2.2 $\pm$ 3.3& 2.2 $\pm$ 3.3& 1.2 $\pm$ 2.7& 1.1 $\pm$ 2.6& 1.1 $\pm$ 2.6\\\hline 
		10 & 8.8 $\pm$ 1.7& 8.8 $\pm$ 1.7& 8.8 $\pm$ 1.7& 8.4 $\pm$ 2& 8.4 $\pm$ 2& 8.4 $\pm$ 2& 7.7 $\pm$ 2.5& 7.7 $\pm$ 2.5& 7.7 $\pm$ 2.5\\ 
		& 8.9 $\pm$ 2.3& 8.8 $\pm$ 1.7& 8.8 $\pm$ 1.7& 8.5 $\pm$ 2.2& 8.4 $\pm$ 2& 8.4 $\pm$ 2& 7.7 $\pm$ 2.7& 7.7 $\pm$ 2.5& 7.7 $\pm$ 2.5\\\hline 
		30 & 9.8 $\pm$ 1.1& 9.8 $\pm$ 1.1& 9.8 $\pm$ 1.1& 9.7 $\pm$ 1.1& 9.7 $\pm$ 1.1& 9.7 $\pm$ 1.1& 9.5 $\pm$ 1.5& 9.5 $\pm$ 1.5& 9.5 $\pm$ 1.5\\ 
		& 9.9 $\pm$ 2.1& 9.8 $\pm$ 1.1& 9.8 $\pm$ 1.1& 9.7 $\pm$ 1.6& 9.7 $\pm$ 1.1& 9.7 $\pm$ 1.1& 9.6 $\pm$ 1.7& 9.5 $\pm$ 1.5& 9.5 $\pm$ 1.5\\\hline
		
	\end{tabular}
\end{table}

\begin{table}[H]
	\centering
	\scriptsize
	\begin{tabular}{|c|ccc|ccc|ccc|}
		\hline
		& \multicolumn{3}{c|}{$k$=100} & \multicolumn{3}{c|}{$k$=200} & \multicolumn{3}{c|}{$k$=300}	\\
		$L$ & $k^*=$100 & $k^*=$200 & $k^*=$300 &	$k^*=$100 & $k^*=$200 & $k^*=$300&$k^*=$100 & $k^*=$200 & $k^*=$300\\\hline
		0.005 & 7.2 $\pm$ 4.2& 7.2 $\pm$ 4.2& 7.2 $\pm$ 4.2& 7.2 $\pm$ 2.7& 7.2 $\pm$ 2.7& 7.2 $\pm$ 2.7& 7.4 $\pm$ 2.6& 7.4 $\pm$ 2.6& 7.4 $\pm$ 2.6\\ 
		& 6.6 $\pm$ 6.3& 7.2 $\pm$ 4.6& 7.2 $\pm$ 4.2& 7.3 $\pm$ 7.4& 7.2 $\pm$ 3.3& 7.2 $\pm$ 2.8& 8.2 $\pm$ 8.3& 7.3 $\pm$ 2.8& 7.4 $\pm$ 2.6\\\hline 
		0.05 & 2 $\pm$ 4.3& 2 $\pm$ 4.3& 2 $\pm$ 4.3& 1.8 $\pm$ 3& 1.8 $\pm$ 3& 1.8 $\pm$ 3& 1.9 $\pm$ 3.1& 1.9 $\pm$ 3.1& 1.9 $\pm$ 3.1\\ 
		& 2.7 $\pm$ 6.5& 2.1 $\pm$ 4.7& 2 $\pm$ 4.3& 3.6 $\pm$ 7.8& 1.9 $\pm$ 3.3& 1.8 $\pm$ 3& 4.9 $\pm$ 9.4& 2 $\pm$ 3.3& 1.9 $\pm$ 3.1\\\hline 
		1 & 0.6 $\pm$ 4.1& 0.6 $\pm$ 4.1& 0.6 $\pm$ 4.1& 0.2 $\pm$ 1.8& 0.2 $\pm$ 1.8& 0.2 $\pm$ 1.8& 0.2 $\pm$ 1.6& 0.2 $\pm$ 1.6& 0.2 $\pm$ 1.6\\ 
		& 2.3 $\pm$ 6.7& 0.9 $\pm$ 4.9& 0.7 $\pm$ 4.2& 3.8 $\pm$ 8& 0.4 $\pm$ 2.7& 0.2 $\pm$ 1.8& 5.1 $\pm$ 9.4& 0.3 $\pm$ 2.2& 0.2 $\pm$ 1.8\\\hline 
		3 & 4.2 $\pm$ 4.8& 4.2 $\pm$ 4.8& 4.2 $\pm$ 4.8& 3.6 $\pm$ 4& 3.6 $\pm$ 4& 3.6 $\pm$ 4& 2.9 $\pm$ 3.6& 2.9 $\pm$ 3.6& 2.9 $\pm$ 3.6\\ 
		& 6.4 $\pm$ 6.5& 4.3 $\pm$ 5.1& 4.2 $\pm$ 4.8& 7.9 $\pm$ 7.7& 3.7 $\pm$ 4.3& 3.6 $\pm$ 4.1& 8.7 $\pm$ 8.5& 3.1 $\pm$ 4.3& 3 $\pm$ 3.8\\\hline 
		10 & 9.2 $\pm$ 3.2& 9.2 $\pm$ 3.2& 9.2 $\pm$ 3.2& 8.8 $\pm$ 1.8& 8.8 $\pm$ 1.8& 8.8 $\pm$ 1.8& 8.6 $\pm$ 1.9& 8.6 $\pm$ 1.9& 8.6 $\pm$ 1.9\\ 
		& 10.1 $\pm$ 4.7& 9.3 $\pm$ 3.6& 9.2 $\pm$ 3.2& 10.7 $\pm$ 5.4& 8.9 $\pm$ 2.3& 8.8 $\pm$ 1.9& 11.5 $\pm$ 6.7& 8.7 $\pm$ 2& 8.6 $\pm$ 1.9\\\hline 
		30 & 10 $\pm$ 2.5& 10 $\pm$ 2.5& 10 $\pm$ 2.5& 9.8 $\pm$ 1.1& 9.8 $\pm$ 1.1& 9.8 $\pm$ 1.1& 9.7 $\pm$ 1& 9.7 $\pm$ 1& 9.7 $\pm$ 1\\ 
		& 10.7 $\pm$ 4.4& 10.2 $\pm$ 3.3& 10 $\pm$ 2.7& 11.2 $\pm$ 5.1& 9.8 $\pm$ 1.2& 9.8 $\pm$ 1.1& 11.9 $\pm$ 6.1& 9.8 $\pm$ 1.4& 9.7 $\pm$ 1\\\hline  \end{tabular}
	\caption{$k_0=10$ exponentiated outliers. {\em Top:} $\modt(1/0.5)$. {\em Bottom:} Burr(1,0.5,1/0.5).}
	\label{tab:xi-pos-ha-exp}
\end{table}

\newpage

\begin{table}[H]
	\centering
	\scriptsize
	\begin{tabular}{|c|ccc|ccc|ccc|}
		\hline
		& \multicolumn{3}{c|}{$k$=200} & \multicolumn{3}{c|}{$k$=400} & \multicolumn{3}{c|}{$k$=600}	\\
		$C$ & $k^*=$200 & $k^*=$400 & $k^*=$600 &	$k^*=$200 & $k^*=$400 & $k^*=$600&$k^*=$200 & $k^*=$400 & $k^*=$600\\\hline
		0.001 & 9.5 $\pm$ 1.3& 9.5 $\pm$ 1.3& 9.5 $\pm$ 1.3& 9.6 $\pm$ 1.2& 9.6 $\pm$ 1.2& 9.6 $\pm$ 1.2& 9.8 $\pm$ 1.1& 9.8 $\pm$ 1.1& 9.8 $\pm$ 1.1\\ 
		& 9.5 $\pm$ 1.6& 9.5 $\pm$ 1.3& 9.5 $\pm$ 1.3& 9.6 $\pm$ 1.4& 9.6 $\pm$ 1.2& 9.6 $\pm$ 1.2& 9.8 $\pm$ 1.1& 9.8 $\pm$ 1.1& 9.8 $\pm$ 1.1\\\hline 
		0.01 & 4.6 $\pm$ 3.6& 4.6 $\pm$ 3.6& 4.6 $\pm$ 3.6& 5.3 $\pm$ 3.5& 5.3 $\pm$ 3.5& 5.3 $\pm$ 3.5& 6.1 $\pm$ 3.3& 6.1 $\pm$ 3.3& 6.1 $\pm$ 3.3\\ 
		& 4.7 $\pm$ 3.9& 4.6 $\pm$ 3.6& 4.6 $\pm$ 3.6& 5.3 $\pm$ 3.7& 5.3 $\pm$ 3.5& 5.3 $\pm$ 3.5& 6 $\pm$ 3.3& 6.1 $\pm$ 3.3& 6.1 $\pm$ 3.3\\\hline 
		1 & 0.2 $\pm$ 1.5& 0.2 $\pm$ 1.5& 0.2 $\pm$ 1.5& 0.2 $\pm$ 1.2& 0.2 $\pm$ 1.2& 0.2 $\pm$ 1.2& 0.2 $\pm$ 1.5& 0.2 $\pm$ 1.5& 0.2 $\pm$ 1.5\\ 
		& 0.3 $\pm$ 2.2& 0.2 $\pm$ 1.5& 0.2 $\pm$ 1.5& 0.2 $\pm$ 1.4& 0.2 $\pm$ 1.2& 0.2 $\pm$ 1.2& 0.2 $\pm$ 1.5& 0.2 $\pm$ 1.5& 0.2 $\pm$ 1.5\\\hline 
		10 & 7.1 $\pm$ 3.9& 7.1 $\pm$ 3.9& 7.1 $\pm$ 3.9& 5.5 $\pm$ 4.5& 5.5 $\pm$ 4.5& 5.5 $\pm$ 4.5& 3.3 $\pm$ 4.5& 3.3 $\pm$ 4.5& 3.3 $\pm$ 4.5\\ 
		& 7.3 $\pm$ 4.2& 7.1 $\pm$ 3.9& 7.1 $\pm$ 3.9& 5.7 $\pm$ 4.6& 5.5 $\pm$ 4.5& 5.5 $\pm$ 4.5& 3.3 $\pm$ 4.5& 3.3 $\pm$ 4.5& 3.3 $\pm$ 4.5\\\hline 
		50 & 9.8 $\pm$ 1& 9.8 $\pm$ 1& 9.8 $\pm$ 1& 9.6 $\pm$ 1.2& 9.6 $\pm$ 1.2& 9.6 $\pm$ 1.2& 9.2 $\pm$ 2& 9.2 $\pm$ 2& 9.2 $\pm$ 2\\ 
		& 9.8 $\pm$ 1.4& 9.8 $\pm$ 1& 9.8 $\pm$ 1& 9.7 $\pm$ 1.4& 9.6 $\pm$ 1.2& 9.6 $\pm$ 1.2& 9.3 $\pm$ 2& 9.2 $\pm$ 2& 9.2 $\pm$ 2\\\hline 
		200 & 10 $\pm$ 0.5& 10 $\pm$ 0.5& 10 $\pm$ 0.5& 10 $\pm$ 0.5& 10 $\pm$ 0.5& 10 $\pm$ 0.5& 9.9 $\pm$ 0.7& 9.9 $\pm$ 0.7& 9.9 $\pm$ 0.7\\ 
		& 10.1 $\pm$ 1.6& 10 $\pm$ 0.5& 10 $\pm$ 0.5& 10 $\pm$ 1.1& 10 $\pm$ 0.5& 10 $\pm$ 0.5& 9.9 $\pm$ 0.7& 9.9 $\pm$ 0.7& 9.9 $\pm$ 0.7\\\hline 
	\end{tabular}
	
\end{table}

\begin{table}[H]
	\centering
	\scriptsize
	\begin{tabular}{|c|ccc|ccc|ccc|}
		\hline
		& \multicolumn{3}{c|}{$k$=100} & \multicolumn{3}{c|}{$k$=200} & \multicolumn{3}{c|}{$k$=300}	\\
		$C$ & $k^*=$100 & $k^*=$200 & $k^*=$300 &	$k^*=$100 & $k^*=$200 & $k^*=$300&$k^*=$100 & $k^*=$200 & $k^*=$300\\\hline
		0.001 & 9.5 $\pm$ 3& 9.5 $\pm$ 3& 9.5 $\pm$ 3& 9.4 $\pm$ 1.3& 9.4 $\pm$ 1.3& 9.4 $\pm$ 1.3& 9.5 $\pm$ 1.1& 9.5 $\pm$ 1.1& 9.5 $\pm$ 1.1\\ 
		& 9.7 $\pm$ 5.2& 9.6 $\pm$ 3.4& 9.5 $\pm$ 3.1& 10.1 $\pm$ 5.9& 9.4 $\pm$ 1.8& 9.4 $\pm$ 1.4& 11.1 $\pm$ 7.3& 9.6 $\pm$ 1.9& 9.5 $\pm$ 1.1\\\hline 
		0.01 & 4.7 $\pm$ 4.6& 4.7 $\pm$ 4.6& 4.7 $\pm$ 4.6& 4.8 $\pm$ 3.7& 4.8 $\pm$ 3.7& 4.8 $\pm$ 3.7& 5.1 $\pm$ 3.5& 5.1 $\pm$ 3.5& 5.1 $\pm$ 3.5\\ 
		& 4.7 $\pm$ 6.4& 4.8 $\pm$ 5& 4.7 $\pm$ 4.8& 5.5 $\pm$ 7.6& 4.9 $\pm$ 4.2& 4.9 $\pm$ 3.9& 6.6 $\pm$ 8.9& 5.1 $\pm$ 3.9& 5.1 $\pm$ 3.6\\\hline 
		1 & 0.6 $\pm$ 4.1& 0.6 $\pm$ 4.1& 0.6 $\pm$ 4.1& 0.2 $\pm$ 1.8& 0.2 $\pm$ 1.8& 0.2 $\pm$ 1.8& 0.2 $\pm$ 1.6& 0.2 $\pm$ 1.6& 0.2 $\pm$ 1.6\\ 
		& 2.3 $\pm$ 6.7& 0.9 $\pm$ 4.9& 0.7 $\pm$ 4.2& 3.8 $\pm$ 8& 0.4 $\pm$ 2.7& 0.2 $\pm$ 1.8& 5.1 $\pm$ 9.4& 0.3 $\pm$ 2.2& 0.2 $\pm$ 1.8\\\hline 
		10 & 7.9 $\pm$ 4.1& 7.9 $\pm$ 4.1& 7.9 $\pm$ 4.1& 7.3 $\pm$ 3.9& 7.3 $\pm$ 3.9& 7.3 $\pm$ 3.9& 6.6 $\pm$ 4.2& 6.6 $\pm$ 4.2& 6.6 $\pm$ 4.2\\ 
		& 9.3 $\pm$ 5.4& 8.1 $\pm$ 4.6& 8 $\pm$ 4.3& 10 $\pm$ 6& 7.4 $\pm$ 4.2& 7.3 $\pm$ 3.9& 10.7 $\pm$ 7.3& 6.7 $\pm$ 4.2& 6.6 $\pm$ 4.2\\\hline 
		50 & 10 $\pm$ 2.5& 10 $\pm$ 2.5& 10 $\pm$ 2.5& 9.8 $\pm$ 0.9& 9.8 $\pm$ 0.9& 9.8 $\pm$ 0.9& 9.8 $\pm$ 1& 9.8 $\pm$ 1& 9.8 $\pm$ 1\\ 
		& 10.7 $\pm$ 4.2& 10.2 $\pm$ 3& 10.1 $\pm$ 2.7& 11.3 $\pm$ 5.3& 9.9 $\pm$ 1.7& 9.8 $\pm$ 1.2& 12 $\pm$ 6.6& 9.8 $\pm$ 1.5& 9.8 $\pm$ 1\\\hline 
		200 & 10.2 $\pm$ 2.7& 10.2 $\pm$ 2.7& 10.2 $\pm$ 2.7& 10 $\pm$ 1& 10 $\pm$ 1& 10 $\pm$ 1& 10 $\pm$ 0.6& 10 $\pm$ 0.6& 10 $\pm$ 0.6\\ 
		& 10.9 $\pm$ 4.3& 10.4 $\pm$ 3.2& 10.2 $\pm$ 2.7& 11.5 $\pm$ 5.2& 10.1 $\pm$ 1.6& 10 $\pm$ 1.2& 12.2 $\pm$ 6.5& 10 $\pm$ 1.2& 10 $\pm$ 0.6\\\hline 	\end{tabular}
	\caption{ $k_0=10$ scaled outliers. {\em Top:} $\modt(1/0.5)$. {\em Bottom:} Burr(1,0.5,1/0.5).}
	\label{tab:xi-pos-ha-scl}
\end{table}

\newpage

\begin{table}[H]
	\centering
	\scriptsize
	\begin{tabular}{|c|ccc|ccc|ccc|}
		\hline
		& \multicolumn{3}{c|}{$k$=100} & \multicolumn{3}{c|}{$k$=200} & \multicolumn{3}{c|}{$k$=300}	\\
		$L$ & $k^*=$100 & $k^*=$200 & $k^*=$300 &	$k^*=$100 & $k^*=$200 & $k^*=$300&$k^*=$100 & $k^*=$200 & $k^*=$300\\\hline
		0.005 & 7.2 $\pm$ 2.6& 7.2 $\pm$ 2.6& 7.2 $\pm$ 2.6& 7.6 $\pm$ 2.2& 7.6 $\pm$ 2.2& 7.6 $\pm$ 2.2& 7.8 $\pm$ 2.1& 7.8 $\pm$ 2.1& 7.8 $\pm$ 2.1\\ 
		& 7.5 $\pm$ 6& 6.5 $\pm$ 4& 6.8 $\pm$ 3.3& 9.9 $\pm$ 7.7& 6.3 $\pm$ 3.5& 7 $\pm$ 2.8& 10.8 $\pm$ 8.5& 6.4 $\pm$ 3.4& 7.2 $\pm$ 2.6\\\hline 
		0.05 & 1.9 $\pm$ 3.2& 1.9 $\pm$ 3.2& 1.9 $\pm$ 3.2& 2.1 $\pm$ 3.2& 2.1 $\pm$ 3.2& 2.1 $\pm$ 3.2& 2.3 $\pm$ 3.2& 2.3 $\pm$ 3.2& 2.3 $\pm$ 3.2\\ 
		& 2.7 $\pm$ 6.2& 1.8 $\pm$ 4.2& 1.9 $\pm$ 3.9& 4 $\pm$ 7.3& 1.5 $\pm$ 3.2& 1.9 $\pm$ 3.4& 4.3 $\pm$ 7.6& 1.5 $\pm$ 3.2& 2 $\pm$ 3.2\\\hline 
		1 & 0.4 $\pm$ 2.2& 0.4 $\pm$ 2.2& 0.4 $\pm$ 2.2& 0.3 $\pm$ 1.8& 0.3 $\pm$ 1.8& 0.3 $\pm$ 1.8& 0.3 $\pm$ 2& 0.3 $\pm$ 2& 0.3 $\pm$ 2\\ 
		& 2.2 $\pm$ 4.7& 1.2 $\pm$ 4& 0.6 $\pm$ 3.1& 2.6 $\pm$ 5.3& 0.6 $\pm$ 2.6& 0.4 $\pm$ 2.1& 2.6 $\pm$ 5.4& 0.6 $\pm$ 2.6& 0.4 $\pm$ 2.1\\\hline 
		3 & 5 $\pm$ 3.5& 5 $\pm$ 3.5& 5 $\pm$ 3.5& 4.1 $\pm$ 3.4& 4.1 $\pm$ 3.4& 4.1 $\pm$ 3.4& 3.2 $\pm$ 3.3& 3.2 $\pm$ 3.3& 3.2 $\pm$ 3.3\\ 
		& 6.8 $\pm$ 4.4& 5.4 $\pm$ 4.5& 4.4 $\pm$ 4.1& 6.6 $\pm$ 4.8& 4.5 $\pm$ 3.9& 3.3 $\pm$ 3.8& 6.2 $\pm$ 5& 3.8 $\pm$ 3.8& 2.4 $\pm$ 3.5\\\hline 
		10 & 9.1 $\pm$ 1.8& 9.1 $\pm$ 1.8& 9.1 $\pm$ 1.8& 8.9 $\pm$ 1.6& 8.9 $\pm$ 1.6& 8.9 $\pm$ 1.6& 8.6 $\pm$ 1.8& 8.6 $\pm$ 1.8& 8.6 $\pm$ 1.8\\ 
		& 9.7 $\pm$ 2.7& 9.4 $\pm$ 2.8& 9.1 $\pm$ 2.3& 9.7 $\pm$ 3& 9.1 $\pm$ 2& 8.7 $\pm$ 2& 9.7 $\pm$ 3.3& 8.8 $\pm$ 2.1& 8.4 $\pm$ 2.1\\\hline 
		30 & 9.8 $\pm$ 1.6& 9.8 $\pm$ 1.6& 9.8 $\pm$ 1.6& 9.7 $\pm$ 1& 9.7 $\pm$ 1& 9.7 $\pm$ 1& 9.7 $\pm$ 1& 9.7 $\pm$ 1& 9.7 $\pm$ 1\\ 
		& 10.2 $\pm$ 2.5& 10.1 $\pm$ 2.3& 9.9 $\pm$ 1.9& 10.3 $\pm$ 2.7& 9.8 $\pm$ 1.3& 9.7 $\pm$ 1.2& 10.4 $\pm$ 3& 9.8 $\pm$ 1.3& 9.7 $\pm$ 1.1\\\hline \end{tabular}
\end{table}

\begin{table}[H]
	\centering
	\scriptsize
	\begin{tabular}{|c|ccc|ccc|ccc|}
		\hline
		& \multicolumn{3}{c|}{$k$=200} & \multicolumn{3}{c|}{$k$=400} & \multicolumn{3}{c|}{$k$=600}	\\
		$L$ & $k^*=$200 & $k^*=$400 & $k^*=$600 &	$k^*=$200 & $k^*=$400 & $k^*=$600&$k^*=$200 & $k^*=$400 & $k^*=$600\\\hline
		0.005 & 7.5 $\pm$ 2.1& 7.5 $\pm$ 2.1& 7.5 $\pm$ 2.1& 8.1 $\pm$ 1.8& 8.1 $\pm$ 1.8& 8.1 $\pm$ 1.8& 8.7 $\pm$ 1.6& 8.7 $\pm$ 1.6& 8.7 $\pm$ 1.6\\ 
		& 6.2 $\pm$ 3.9& 6.8 $\pm$ 3& 7.4 $\pm$ 2.5& 6.4 $\pm$ 3.5& 7.2 $\pm$ 2.5& 7.9 $\pm$ 2& 6.8 $\pm$ 3& 7.9 $\pm$ 2.2& 8.5 $\pm$ 1.8\\\hline 
		0.05 & 2.1 $\pm$ 3.1& 2.1 $\pm$ 3.1& 2.1 $\pm$ 3.1& 2.6 $\pm$ 3.2& 2.6 $\pm$ 3.2& 2.6 $\pm$ 3.2& 3.4 $\pm$ 3.4& 3.4 $\pm$ 3.4& 3.4 $\pm$ 3.4\\ 
		& 1.5 $\pm$ 3.4& 1.7 $\pm$ 3.2& 2 $\pm$ 3.2& 1.5 $\pm$ 2.9& 2 $\pm$ 3.1& 2.5 $\pm$ 3.2& 1.7 $\pm$ 3& 2.5 $\pm$ 3.2& 3.3 $\pm$ 3.5\\\hline 
		1 & 0.3 $\pm$ 1.5& 0.3 $\pm$ 1.5& 0.3 $\pm$ 1.5& 0.2 $\pm$ 1.4& 0.2 $\pm$ 1.4& 0.2 $\pm$ 1.4& 0.4 $\pm$ 1.8& 0.4 $\pm$ 1.8& 0.4 $\pm$ 1.8\\ 
		& 1.1 $\pm$ 3.3& 0.6 $\pm$ 2.3& 0.3 $\pm$ 1.6& 0.7 $\pm$ 2.3& 0.3 $\pm$ 1.6& 0.3 $\pm$ 1.8& 0.5 $\pm$ 2.1& 0.4 $\pm$ 1.8& 0.5 $\pm$ 2.2\\\hline 
		3 & 4.2 $\pm$ 3.6& 4.2 $\pm$ 3.6& 4.2 $\pm$ 3.6& 2.3 $\pm$ 3.2& 2.3 $\pm$ 3.2& 2.3 $\pm$ 3.2& 1 $\pm$ 2.6& 1 $\pm$ 2.6& 1 $\pm$ 2.6\\ 
		& 6 $\pm$ 4.2& 5 $\pm$ 4& 3.4 $\pm$ 3.8& 4.6 $\pm$ 4.1& 3.1 $\pm$ 3.6& 1.6 $\pm$ 3.2& 2.7 $\pm$ 3.7& 1.5 $\pm$ 3.1& 0.7 $\pm$ 2.6\\\hline 
		10 & 8.9 $\pm$ 1.7& 8.9 $\pm$ 1.7& 8.9 $\pm$ 1.7& 8.3 $\pm$ 2& 8.3 $\pm$ 2& 8.3 $\pm$ 2& 7.2 $\pm$ 2.7& 7.2 $\pm$ 2.7& 7.2 $\pm$ 2.7\\ 
		& 9.3 $\pm$ 1.9& 9 $\pm$ 1.7& 8.7 $\pm$ 1.9& 8.9 $\pm$ 2& 8.6 $\pm$ 2& 7.9 $\pm$ 2.5& 8.2 $\pm$ 2.5& 7.6 $\pm$ 2.7& 6.5 $\pm$ 3.4\\\hline 
		30 & 9.7 $\pm$ 0.8& 9.7 $\pm$ 0.8& 9.7 $\pm$ 0.8& 9.6 $\pm$ 0.9& 9.6 $\pm$ 0.9& 9.6 $\pm$ 0.9& 9.4 $\pm$ 1.2& 9.4 $\pm$ 1.2& 9.4 $\pm$ 1.2\\ 
		& 9.9 $\pm$ 1.5& 9.8 $\pm$ 1.1& 9.7 $\pm$ 0.9& 9.8 $\pm$ 1.1& 9.7 $\pm$ 0.9& 9.5 $\pm$ 1& 9.6 $\pm$ 1.1& 9.5 $\pm$ 1.1& 9.3 $\pm$ 1.3\\\hline 
	\end{tabular}
	\caption{$k_0=10$ exponentiated outliers. {\em Top:} Beta(1,1/0.5). {\em Bottom:} Reverse Burr(1,0.5,1/0.5).}
	\label{tab:xi-neg-ha-exp}
\end{table}

\newpage

\begin{table}[H]
	\centering
	\scriptsize
	\begin{tabular}{|c|ccc|ccc|ccc|}
		\hline
		& \multicolumn{3}{c|}{$k$=100} & \multicolumn{3}{c|}{$k$=200} & \multicolumn{3}{c|}{$k$=300}	\\
		$C$ & $k^*=$100 & $k^*=$200 & $k^*=$300 &	$k^*=$100 & $k^*=$200 & $k^*=$300&$k^*=$100 & $k^*=$200 & $k^*=$300\\\hline
		0.001 & 9.6 $\pm$ 1.8& 9.6 $\pm$ 1.8& 9.6 $\pm$ 1.8& 9.6 $\pm$ 0.9& 9.6 $\pm$ 0.9& 9.6 $\pm$ 0.9& 9.7 $\pm$ 1& 9.7 $\pm$ 1& 9.7 $\pm$ 1\\ 
		& 9.6 $\pm$ 3.5& 9.5 $\pm$ 2.5& 9.5 $\pm$ 2.1& 9.8 $\pm$ 3.7& 9.3 $\pm$ 1.7& 9.5 $\pm$ 1.3& 10.1 $\pm$ 4.3& 9.4 $\pm$ 1.8& 9.6 $\pm$ 1\\\hline 
		0.01 & 5.4 $\pm$ 3.6& 5.4 $\pm$ 3.6& 5.4 $\pm$ 3.6& 5.8 $\pm$ 3& 5.8 $\pm$ 3& 5.8 $\pm$ 3& 6.2 $\pm$ 2.9& 6.2 $\pm$ 2.9& 6.2 $\pm$ 2.9\\ 
		& 6.4 $\pm$ 7& 4.5 $\pm$ 4.6& 4.9 $\pm$ 4& 8.8 $\pm$ 8.4& 4.3 $\pm$ 4& 5 $\pm$ 3.5& 10 $\pm$ 9.1& 4.4 $\pm$ 3.8& 5.3 $\pm$ 3.4\\\hline 
		1 & 0.4 $\pm$ 2.2& 0.4 $\pm$ 2.2& 0.4 $\pm$ 2.2& 0.3 $\pm$ 1.8& 0.3 $\pm$ 1.8& 0.3 $\pm$ 1.8& 0.3 $\pm$ 2& 0.3 $\pm$ 2& 0.3 $\pm$ 2\\ 
		& 2.2 $\pm$ 4.7& 1.2 $\pm$ 4& 0.6 $\pm$ 3.1& 2.6 $\pm$ 5.3& 0.6 $\pm$ 2.6& 0.4 $\pm$ 2.1& 2.6 $\pm$ 5.4& 0.6 $\pm$ 2.6& 0.4 $\pm$ 2.1\\\hline 
		10 & 9 $\pm$ 1.6& 9 $\pm$ 1.6& 9 $\pm$ 1.6& 8.8 $\pm$ 1.6& 8.8 $\pm$ 1.6& 8.8 $\pm$ 1.6& 8.5 $\pm$ 1.9& 8.5 $\pm$ 1.9& 8.5 $\pm$ 1.9\\ 
		& 9.6 $\pm$ 2.6& 9.2 $\pm$ 2.4& 8.9 $\pm$ 2& 9.6 $\pm$ 3& 8.9 $\pm$ 1.9& 8.5 $\pm$ 2.1& 9.6 $\pm$ 3.5& 8.6 $\pm$ 2.1& 8.2 $\pm$ 2.4\\\hline 
		50 & 10 $\pm$ 1.4& 10 $\pm$ 1.4& 10 $\pm$ 1.4& 9.9 $\pm$ 0.8& 9.9 $\pm$ 0.8& 9.9 $\pm$ 0.8& 9.9 $\pm$ 0.8& 9.9 $\pm$ 0.8& 9.9 $\pm$ 0.8\\ 
		& 10.5 $\pm$ 3.2& 10.3 $\pm$ 2.9& 10.1 $\pm$ 2.4& 10.5 $\pm$ 3.1& 10 $\pm$ 1.4& 9.9 $\pm$ 1& 10.5 $\pm$ 3.2& 10 $\pm$ 1.3& 9.9 $\pm$ 0.8\\\hline 
		200 & 10 $\pm$ 1.2& 10 $\pm$ 1.2& 10 $\pm$ 1.2& 10 $\pm$ 0.7& 10 $\pm$ 0.7& 10 $\pm$ 0.7& 10 $\pm$ 0.7& 10 $\pm$ 0.7& 10 $\pm$ 0.7\\ 
		& 10.4 $\pm$ 2.6& 10.2 $\pm$ 2& 10.1 $\pm$ 1.7& 10.5 $\pm$ 2.6& 10.1 $\pm$ 1.1& 10 $\pm$ 1& 10.5 $\pm$ 2.8& 10 $\pm$ 0.9& 10 $\pm$ 0.7\\\hline  	\end{tabular}
\end{table}

\begin{table}[H]
	\centering
	\scriptsize
	\begin{tabular}{|c|ccc|ccc|ccc|}
		\hline
		& \multicolumn{3}{c|}{$k$=200} & \multicolumn{3}{c|}{$k$=400} & \multicolumn{3}{c|}{$k$=600}	\\
		$C$ & $k^*=$200 & $k^*=$400 & $k^*=$600 &	$k^*=$200 & $k^*=$400 & $k^*=$600&$k^*=$200 & $k^*=$400 & $k^*=$600\\\hline
		0.001 & 9.6 $\pm$ 1.1& 9.6 $\pm$ 1.1& 9.6 $\pm$ 1.1& 9.8 $\pm$ 1.1& 9.8 $\pm$ 1.1& 9.8 $\pm$ 1.1& 10 $\pm$ 1& 10 $\pm$ 1& 10 $\pm$ 1\\ 
		& 9.4 $\pm$ 2.4& 9.5 $\pm$ 1.6& 9.6 $\pm$ 1.3& 9.4 $\pm$ 1.9& 9.6 $\pm$ 1.2& 9.8 $\pm$ 1.1& 9.6 $\pm$ 1.4& 9.8 $\pm$ 1.1& 9.9 $\pm$ 1\\\hline 
		0.01 & 5.8 $\pm$ 3& 5.8 $\pm$ 3& 5.8 $\pm$ 3& 6.6 $\pm$ 2.8& 6.6 $\pm$ 2.8& 6.6 $\pm$ 2.8& 7.4 $\pm$ 2.4& 7.4 $\pm$ 2.4& 7.4 $\pm$ 2.4\\ 
		& 4.2 $\pm$ 4.3& 4.8 $\pm$ 3.6& 5.6 $\pm$ 3.2& 4.4 $\pm$ 4.1& 5.4 $\pm$ 3.3& 6.3 $\pm$ 3& 4.8 $\pm$ 3.7& 6.2 $\pm$ 3.1& 7.1 $\pm$ 2.7\\\hline 
		1 & 0.3 $\pm$ 1.5& 0.3 $\pm$ 1.5& 0.3 $\pm$ 1.5& 0.2 $\pm$ 1.4& 0.2 $\pm$ 1.4& 0.2 $\pm$ 1.4& 0.4 $\pm$ 1.8& 0.4 $\pm$ 1.8& 0.4 $\pm$ 1.8\\ 
		& 1.1 $\pm$ 3.3& 0.6 $\pm$ 2.3& 0.3 $\pm$ 1.6& 0.7 $\pm$ 2.3& 0.3 $\pm$ 1.6& 0.3 $\pm$ 1.8& 0.5 $\pm$ 2.1& 0.4 $\pm$ 1.8& 0.5 $\pm$ 2.2\\\hline 
		10 & 8.8 $\pm$ 1.8& 8.8 $\pm$ 1.8& 8.8 $\pm$ 1.8& 8.2 $\pm$ 2.2& 8.2 $\pm$ 2.2& 8.2 $\pm$ 2.2& 6.9 $\pm$ 3& 6.9 $\pm$ 3& 6.9 $\pm$ 3\\ 
		& 9.4 $\pm$ 2.5& 9 $\pm$ 2.2& 8.7 $\pm$ 2.1& 8.9 $\pm$ 2.5& 8.5 $\pm$ 2.2& 7.8 $\pm$ 2.8& 8 $\pm$ 2.9& 7.4 $\pm$ 3& 6 $\pm$ 3.7\\\hline 
		50 & 9.9 $\pm$ 0.9& 9.9 $\pm$ 0.9& 9.9 $\pm$ 0.9& 9.8 $\pm$ 0.9& 9.8 $\pm$ 0.9& 9.8 $\pm$ 0.9& 9.7 $\pm$ 1.2& 9.7 $\pm$ 1.2& 9.7 $\pm$ 1.2\\ 
		& 10.1 $\pm$ 2& 10 $\pm$ 1.6& 9.9 $\pm$ 1.2& 10 $\pm$ 1.4& 9.8 $\pm$ 0.9& 9.8 $\pm$ 1& 9.8 $\pm$ 1& 9.8 $\pm$ 1& 9.7 $\pm$ 1.3\\\hline 
		200 & 10 $\pm$ 0.6& 10 $\pm$ 0.6& 10 $\pm$ 0.6& 10 $\pm$ 0.6& 10 $\pm$ 0.6& 10 $\pm$ 0.6& 10 $\pm$ 0.8& 10 $\pm$ 0.8& 10 $\pm$ 0.8\\ 
		& 10.1 $\pm$ 1.5& 10 $\pm$ 0.9& 10 $\pm$ 0.6& 10.1 $\pm$ 1.1& 10 $\pm$ 0.7& 10 $\pm$ 0.6& 10 $\pm$ 0.7& 10 $\pm$ 0.7& 10 $\pm$ 0.8\\\hline 
	\end{tabular}
	\caption{$k_0=10$ scaled outliers. {\em Top:} Beta(1,1/0.5). {\em Bottom:} Reverse 
		Burr(1,0.5,1/0.5).}
	\label{tab:xi-neg-ha-scl}
\end{table}

\newpage
\begin{table}[H]
	\centering
	\scriptsize
	\begin{tabular}{|c|ccc|ccc|ccc|}
		\hline
		& \multicolumn{3}{c|}{$k$=100} & \multicolumn{3}{c|}{$k$=150} & \multicolumn{3}{c|}{$k$=200}	\\
		$L$ & $k^*=$100 & $k^*=$150 & $k^*=$200 &	$k^*=$100 & $k^*=$150 & $k^*=$200&$k^*=$100 & $k^*=$150 & $k^*=$200\\\hline
		0.005 & 7.4 $\pm$ 3.4& 7.4 $\pm$ 3.4& 7.4 $\pm$ 3.4& 7.6 $\pm$ 2.5& 7.6 $\pm$ 2.5& 7.6 $\pm$ 2.5& 7.7 $\pm$ 2.3& 7.7 $\pm$ 2.3& 7.7 $\pm$ 2.3\\ 
		& 6.4 $\pm$ 5.8& 6.9 $\pm$ 4.6& 7.3 $\pm$ 4& 6.6 $\pm$ 6.2& 6.8 $\pm$ 3.9& 7.3 $\pm$ 2.9& 6.8 $\pm$ 6.6& 6.9 $\pm$ 3.9& 7.5 $\pm$ 2.8\\\hline 
		0.05 & 2.2 $\pm$ 3.6& 2.2 $\pm$ 3.6& 2.2 $\pm$ 3.6& 2.3 $\pm$ 3.2& 2.3 $\pm$ 3.2& 2.3 $\pm$ 3.2& 2.5 $\pm$ 3.2& 2.5 $\pm$ 3.2& 2.5 $\pm$ 3.2\\ 
		& 2.7 $\pm$ 6.3& 2.5 $\pm$ 5& 2.4 $\pm$ 4.3& 2.6 $\pm$ 6.1& 2.4 $\pm$ 4.3& 2.4 $\pm$ 3.7& 3.1 $\pm$ 6.8& 2.3 $\pm$ 3.9& 2.5 $\pm$ 3.6\\\hline 
		1 & 0.3 $\pm$ 2.6& 0.3 $\pm$ 2.6& 0.3 $\pm$ 2.6& 0.2 $\pm$ 1.6& 0.2 $\pm$ 1.6& 0.2 $\pm$ 1.6& 0.2 $\pm$ 1.3& 0.2 $\pm$ 1.3& 0.2 $\pm$ 1.3\\ 
		& 1.7 $\pm$ 5.5& 0.8 $\pm$ 4.1& 0.5 $\pm$ 3.1& 2.2 $\pm$ 6& 0.6 $\pm$ 3.2& 0.3 $\pm$ 2.5& 2.8 $\pm$ 6.6& 0.5 $\pm$ 3.1& 0.3 $\pm$ 2.2\\\hline 
		3 & 2.6 $\pm$ 4.5& 2.6 $\pm$ 4.5& 2.6 $\pm$ 4.5& 1.9 $\pm$ 3.3& 1.9 $\pm$ 3.3& 1.9 $\pm$ 3.3& 1.5 $\pm$ 3.1& 1.5 $\pm$ 3.1& 1.5 $\pm$ 3.1\\ 
		& 5.5 $\pm$ 6.1& 3.5 $\pm$ 5.4& 2.9 $\pm$ 5& 5.9 $\pm$ 6.3& 3 $\pm$ 4.7& 2.1 $\pm$ 3.8& 6.5 $\pm$ 7& 2.8 $\pm$ 4.7& 1.8 $\pm$ 3.5\\\hline 
		10 & 8.7 $\pm$ 2.7& 8.7 $\pm$ 2.7& 8.7 $\pm$ 2.7& 8.4 $\pm$ 2.2& 8.4 $\pm$ 2.2& 8.4 $\pm$ 2.2& 8.2 $\pm$ 2.2& 8.2 $\pm$ 2.2& 8.2 $\pm$ 2.2\\ 
		& 9.8 $\pm$ 4& 9.2 $\pm$ 3.7& 8.9 $\pm$ 3.2& 9.8 $\pm$ 4& 8.8 $\pm$ 3& 8.5 $\pm$ 2.6& 10.1 $\pm$ 4.6& 8.7 $\pm$ 2.9& 8.3 $\pm$ 2.5\\\hline 
		30 & 9.8 $\pm$ 2.1& 9.8 $\pm$ 2.1& 9.8 $\pm$ 2.1& 9.7 $\pm$ 1& 9.7 $\pm$ 1& 9.7 $\pm$ 1& 9.6 $\pm$ 1.1& 9.6 $\pm$ 1.1& 9.6 $\pm$ 1.1\\ 
		& 10.4 $\pm$ 3.5& 10 $\pm$ 2.9& 9.9 $\pm$ 2.5& 10.5 $\pm$ 3.8& 9.9 $\pm$ 2.3& 9.7 $\pm$ 1.9& 10.8 $\pm$ 4.4& 9.8 $\pm$ 1.9& 9.7 $\pm$ 1.2\\\hline \end{tabular}
\end{table}

\begin{table}[H]
	\centering
	\scriptsize
	\begin{tabular}{|c|ccc|ccc|ccc|}
		\hline
		& \multicolumn{3}{c|}{$k$=100} & \multicolumn{3}{c|}{$k$=200} & \multicolumn{3}{c|}{$k$=300}	\\
		$L$ & $k^*=$100 & $k^*=$200 & $k^*=$300 &	$k^*=$100 & $k^*=$200 & $k^*=$300&$k^*=$100 & $k^*=$200 & $k^*=$300\\\hline
		0.005 & 7.8 $\pm$ 2.5& 7.8 $\pm$ 2.5& 7.8 $\pm$ 2.5& 8.4 $\pm$ 1.8& 8.4 $\pm$ 1.8& 8.4 $\pm$ 1.8& 8.8 $\pm$ 1.7& 8.8 $\pm$ 1.7& 8.8 $\pm$ 1.7\\ 
		& 6.8 $\pm$ 6.5& 6.4 $\pm$ 4.7& 6.9 $\pm$ 4.1& 8.8 $\pm$ 7.6& 6.2 $\pm$ 3.9& 6.8 $\pm$ 3.1& 9.7 $\pm$ 8.3& 6.3 $\pm$ 3.9& 7 $\pm$ 2.8\\\hline 
		0.05 & 2.6 $\pm$ 3.3& 2.6 $\pm$ 3.3& 2.6 $\pm$ 3.3& 3.3 $\pm$ 3.3& 3.3 $\pm$ 3.3& 3.3 $\pm$ 3.3& 3.9 $\pm$ 3.3& 3.9 $\pm$ 3.3& 3.9 $\pm$ 3.3\\ 
		& 2.5 $\pm$ 6.2& 1.9 $\pm$ 4.6& 1.8 $\pm$ 4& 3.3 $\pm$ 7.1& 1.7 $\pm$ 3.5& 1.9 $\pm$ 3.4& 3.7 $\pm$ 7.8& 1.7 $\pm$ 3.4& 2 $\pm$ 3.1\\\hline 
		1 & 0.3 $\pm$ 2.2& 0.3 $\pm$ 2.2& 0.3 $\pm$ 2.2& 0.3 $\pm$ 1.7& 0.3 $\pm$ 1.7& 0.3 $\pm$ 1.7& 0.4 $\pm$ 1.9& 0.4 $\pm$ 1.9& 0.4 $\pm$ 1.9\\ 
		& 3.2 $\pm$ 6.4& 1.3 $\pm$ 4.7& 0.7 $\pm$ 3.7& 4.4 $\pm$ 7.5& 0.9 $\pm$ 3.2& 0.4 $\pm$ 2.4& 4.8 $\pm$ 8.2& 0.7 $\pm$ 2.8& 0.3 $\pm$ 1.9\\\hline 
		3 & 1 $\pm$ 3& 1 $\pm$ 3& 1 $\pm$ 3& 0.4 $\pm$ 2.3& 0.4 $\pm$ 2.3& 0.4 $\pm$ 2.3& 0.3 $\pm$ 2.3& 0.3 $\pm$ 2.3& 0.3 $\pm$ 2.3\\ 
		& 7.8 $\pm$ 5.2& 5.7 $\pm$ 5.1& 4.7 $\pm$ 4.8& 8.4 $\pm$ 6.1& 5.2 $\pm$ 4.4& 3.7 $\pm$ 4.2& 8.7 $\pm$ 6.9& 4.4 $\pm$ 4.3& 2.7 $\pm$ 3.7\\\hline 
		10 & 7.9 $\pm$ 3.2& 7.9 $\pm$ 3.2& 7.9 $\pm$ 3.2& 6.5 $\pm$ 3.5& 6.5 $\pm$ 3.5& 6.5 $\pm$ 3.5& 5 $\pm$ 3.9& 5 $\pm$ 3.9& 5 $\pm$ 3.9\\ 
		& 10.2 $\pm$ 3.8& 9.5 $\pm$ 3.2& 9.2 $\pm$ 3& 10.7 $\pm$ 4.6& 9.2 $\pm$ 2.6& 8.9 $\pm$ 2.3& 10.8 $\pm$ 5& 8.9 $\pm$ 2.3& 8.5 $\pm$ 2.1\\\hline 
		30 & 9.6 $\pm$ 1.6& 9.6 $\pm$ 1.6& 9.6 $\pm$ 1.6& 9.4 $\pm$ 1.3& 9.4 $\pm$ 1.3& 9.4 $\pm$ 1.3& 9.1 $\pm$ 1.6& 9.1 $\pm$ 1.6& 9.1 $\pm$ 1.6\\ 
		& 10.6 $\pm$ 3.4& 10.2 $\pm$ 2.9& 10 $\pm$ 2.2& 11.1 $\pm$ 4.4& 10 $\pm$ 1.9& 9.8 $\pm$ 1.3& 11.3 $\pm$ 5& 9.9 $\pm$ 1.7& 9.7 $\pm$ 1.3\\\hline \end{tabular}
\end{table}

\begin{table}[H]
	\centering
	\scriptsize
	\begin{tabular}{|c|ccc|ccc|ccc|}
		\hline
		& \multicolumn{3}{c|}{$k$=100} & \multicolumn{3}{c|}{$k$=150} & \multicolumn{3}{c|}{$k$=200}	\\
		$L$ & $k^*=$100 & $k^*=$150 & $k^*=$200 &	$k^*=$100 & $k^*=$150 & $k^*=$200&$k^*=$100 & $k^*=$150 & $k^*=$200\\\hline
		0.005 & 7.7 $\pm$ 2.5& 7.7 $\pm$ 2.5& 7.7 $\pm$ 2.5& 8 $\pm$ 2& 8 $\pm$ 2& 8 $\pm$ 2& 8.2 $\pm$ 1.8& 8.2 $\pm$ 1.8& 8.2 $\pm$ 1.8\\ 
		& 6.1 $\pm$ 5.7& 6.4 $\pm$ 4.8& 6.8 $\pm$ 4.1& 6.7 $\pm$ 6.2& 6.2 $\pm$ 4.4& 6.7 $\pm$ 3.4& 7 $\pm$ 6.4& 6.2 $\pm$ 4.2& 6.9 $\pm$ 3.1\\\hline 
		0.05 & 2.4 $\pm$ 3.5& 2.4 $\pm$ 3.5& 2.4 $\pm$ 3.5& 2.6 $\pm$ 3.1& 2.6 $\pm$ 3.1& 2.6 $\pm$ 3.1& 2.8 $\pm$ 3.2& 2.8 $\pm$ 3.2& 2.8 $\pm$ 3.2\\ 
		& 2.2 $\pm$ 5.8& 2 $\pm$ 4.9& 2.1 $\pm$ 4.6& 2.4 $\pm$ 5.8& 1.8 $\pm$ 4.3& 2 $\pm$ 3.8& 2.8 $\pm$ 6.7& 1.8 $\pm$ 4.1& 2 $\pm$ 3.5\\\hline 
		1 & 0.4 $\pm$ 2.8& 0.4 $\pm$ 2.8& 0.4 $\pm$ 2.8& 0.3 $\pm$ 1.6& 0.3 $\pm$ 1.6& 0.3 $\pm$ 1.6& 0.3 $\pm$ 1.8& 0.3 $\pm$ 1.8& 0.3 $\pm$ 1.8\\ 
		& 2.3 $\pm$ 5.9& 1.1 $\pm$ 4.7& 0.7 $\pm$ 4.1& 2.8 $\pm$ 6.3& 0.9 $\pm$ 3.9& 0.6 $\pm$ 3.4& 3.3 $\pm$ 6.9& 0.9 $\pm$ 3.8& 0.5 $\pm$ 2.6\\\hline 
		3 & 1.4 $\pm$ 3.8& 1.4 $\pm$ 3.8& 1.4 $\pm$ 3.8& 0.9 $\pm$ 2.9& 0.9 $\pm$ 2.9& 0.9 $\pm$ 2.9& 0.7 $\pm$ 2.7& 0.7 $\pm$ 2.7& 0.7 $\pm$ 2.7\\ 
		& 6.6 $\pm$ 6& 4.7 $\pm$ 5.6& 3.4 $\pm$ 5.2& 6.9 $\pm$ 6.1& 4.3 $\pm$ 5.1& 2.8 $\pm$ 4.5& 7.2 $\pm$ 6.6& 4 $\pm$ 5& 2.4 $\pm$ 4.1\\\hline 
		10 & 8.1 $\pm$ 2.9& 8.1 $\pm$ 2.9& 8.1 $\pm$ 2.9& 7.5 $\pm$ 2.9& 7.5 $\pm$ 2.9& 7.5 $\pm$ 2.9& 7.1 $\pm$ 3.3& 7.1 $\pm$ 3.3& 7.1 $\pm$ 3.3\\ 
		& 9.9 $\pm$ 3.8& 9.3 $\pm$ 3.4& 8.9 $\pm$ 3.3& 10 $\pm$ 3.7& 9 $\pm$ 2.6& 8.6 $\pm$ 2.6& 10.2 $\pm$ 4.4& 8.9 $\pm$ 2.6& 8.4 $\pm$ 2.7\\\hline 
		30 & 9.7 $\pm$ 2& 9.7 $\pm$ 2& 9.7 $\pm$ 2& 9.5 $\pm$ 1& 9.5 $\pm$ 1& 9.5 $\pm$ 1& 9.4 $\pm$ 1.1& 9.4 $\pm$ 1.1& 9.4 $\pm$ 1.1\\ 
		& 10.6 $\pm$ 3.8& 10.2 $\pm$ 3.3& 10 $\pm$ 2.8& 10.6 $\pm$ 3.6& 10 $\pm$ 2.3& 9.8 $\pm$ 2& 10.7 $\pm$ 4& 10 $\pm$ 2.3& 9.8 $\pm$ 1.7\\\hline  \end{tabular}
	\caption{$k_0=10$ exponentiated outliers. {\em Top:} Lognormal(0,1). {\em Middle:} $\modn(0,1)$. {\em Bottom:} Weibull(1,1).}
	\label{tab:xi-zero-ha-exp}
\end{table}


\begin{table}[H]
	\centering
	\scriptsize
	\begin{tabular}{|c|ccc|ccc|ccc|}
		\hline
		& \multicolumn{3}{c|}{$k$=100} & \multicolumn{3}{c|}{$k$=150} & \multicolumn{3}{c|}{$k$=200}	\\
		$C$ & $k^*=$100 & $k^*=$150 & $k^*=$200 &	$k^*=$100 & $k^*=$150 & $k^*=$200&$k^*=$100 & $k^*=$150 & $k^*=$200\\\hline
		0.001 & 9.7 $\pm$ 2.5& 9.7 $\pm$ 2.5& 9.7 $\pm$ 2.5& 9.6 $\pm$ 1.1& 9.6 $\pm$ 1.1& 9.6 $\pm$ 1.1& 9.6 $\pm$ 0.9& 9.6 $\pm$ 0.9& 9.6 $\pm$ 0.9\\ 
		& 9.7 $\pm$ 4.7& 9.7 $\pm$ 3.8& 9.7 $\pm$ 3& 9.7 $\pm$ 4.7& 9.6 $\pm$ 2.8& 9.6 $\pm$ 1.9& 10 $\pm$ 5.3& 9.5 $\pm$ 2.5& 9.6 $\pm$ 1.4\\\hline 
		0.01 & 5.3 $\pm$ 4.2& 5.3 $\pm$ 4.2& 5.3 $\pm$ 4.2& 5.5 $\pm$ 3.5& 5.5 $\pm$ 3.5& 5.5 $\pm$ 3.5& 5.7 $\pm$ 3.3& 5.7 $\pm$ 3.3& 5.7 $\pm$ 3.3\\ 
		& 4.6 $\pm$ 6.4& 5 $\pm$ 5.3& 5.3 $\pm$ 4.7& 4.7 $\pm$ 6.4& 4.9 $\pm$ 4.8& 5.4 $\pm$ 4.1& 4.9 $\pm$ 6.9& 5 $\pm$ 4.5& 5.6 $\pm$ 3.7\\\hline 
		1 & 0.3 $\pm$ 2.6& 0.3 $\pm$ 2.6& 0.3 $\pm$ 2.6& 0.2 $\pm$ 1.6& 0.2 $\pm$ 1.6& 0.2 $\pm$ 1.6& 0.2 $\pm$ 1.3& 0.2 $\pm$ 1.3& 0.2 $\pm$ 1.3\\ 
		& 1.7 $\pm$ 5.5& 0.8 $\pm$ 4.1& 0.5 $\pm$ 3.1& 2.2 $\pm$ 6& 0.6 $\pm$ 3.2& 0.3 $\pm$ 2.5& 2.8 $\pm$ 6.6& 0.5 $\pm$ 3.1& 0.3 $\pm$ 2.2\\\hline 
		10 & 7.1 $\pm$ 4.6& 7.1 $\pm$ 4.6& 7.1 $\pm$ 4.6& 6.4 $\pm$ 4.2& 6.4 $\pm$ 4.2& 6.4 $\pm$ 4.2& 5.8 $\pm$ 4.4& 5.8 $\pm$ 4.4& 5.8 $\pm$ 4.4\\ 
		& 9 $\pm$ 4.6& 7.7 $\pm$ 4.7& 7.2 $\pm$ 4.7& 9.1 $\pm$ 4.7& 7.2 $\pm$ 4.5& 6.6 $\pm$ 4.4& 9.4 $\pm$ 5.4& 6.9 $\pm$ 4.5& 6 $\pm$ 4.4\\\hline 
		50 & 9.9 $\pm$ 2& 9.9 $\pm$ 2& 9.9 $\pm$ 2& 9.7 $\pm$ 0.9& 9.7 $\pm$ 0.9& 9.7 $\pm$ 0.9& 9.7 $\pm$ 1& 9.7 $\pm$ 1& 9.7 $\pm$ 1\\ 
		& 10.4 $\pm$ 3.5& 10.2 $\pm$ 3& 10 $\pm$ 2.3& 10.5 $\pm$ 3.4& 9.9 $\pm$ 2.1& 9.8 $\pm$ 1.5& 10.9 $\pm$ 4.4& 9.9 $\pm$ 2.1& 9.7 $\pm$ 1.4\\\hline 
		200 & 10 $\pm$ 1& 10 $\pm$ 1& 10 $\pm$ 1& 10 $\pm$ 0.5& 10 $\pm$ 0.5& 10 $\pm$ 0.5& 10 $\pm$ 0.5& 10 $\pm$ 0.5& 10 $\pm$ 0.5\\ 
		& 10.5 $\pm$ 3.1& 10.2 $\pm$ 2.1& 10.1 $\pm$ 1.5& 10.6 $\pm$ 3.2& 10.1 $\pm$ 1.5& 10 $\pm$ 1.1& 10.9 $\pm$ 4& 10.1 $\pm$ 1.3& 10 $\pm$ 1\\\hline 	\end{tabular}
\end{table}

\begin{table}[H]
	\centering
	\scriptsize
	\begin{tabular}{|c|ccc|ccc|ccc|}
		\hline
		& \multicolumn{3}{c|}{$k$=100} & \multicolumn{3}{c|}{$k$=200} & \multicolumn{3}{c|}{$k$=300}	\\
		$C$ & $k^*=$100 & $k^*=$200 & $k^*=$300 &	$k^*=$100 & $k^*=$200 & $k^*=$300&$k^*=$100 & $k^*=$200 & $k^*=$300\\\hline
		0.001 & 9.7 $\pm$ 1& 9.7 $\pm$ 1& 9.7 $\pm$ 1& 9.8 $\pm$ 0.6& 9.8 $\pm$ 0.6& 9.8 $\pm$ 0.6& 9.9 $\pm$ 0.6& 9.9 $\pm$ 0.6& 9.9 $\pm$ 0.6\\ 
		& 9.5 $\pm$ 4.2& 9.5 $\pm$ 3& 9.5 $\pm$ 2.3& 10.5 $\pm$ 5.2& 9.3 $\pm$ 2.3& 9.5 $\pm$ 1.6& 11 $\pm$ 6& 9.4 $\pm$ 2.1& 9.5 $\pm$ 1.1\\\hline 
		0.01 & 6.3 $\pm$ 3.2& 6.3 $\pm$ 3.2& 6.3 $\pm$ 3.2& 7.1 $\pm$ 2.5& 7.1 $\pm$ 2.5& 7.1 $\pm$ 2.5& 7.6 $\pm$ 2.2& 7.6 $\pm$ 2.2& 7.6 $\pm$ 2.2\\ 
		& 4.9 $\pm$ 6.8& 4.5 $\pm$ 5.2& 4.9 $\pm$ 4.7& 6.3 $\pm$ 7.9& 4.2 $\pm$ 4.5& 5 $\pm$ 3.9& 7.3 $\pm$ 8.5& 4.4 $\pm$ 4.3& 5.2 $\pm$ 3.5\\\hline 
		1 & 0.3 $\pm$ 2.2& 0.3 $\pm$ 2.2& 0.3 $\pm$ 2.2& 0.3 $\pm$ 1.7& 0.3 $\pm$ 1.7& 0.3 $\pm$ 1.7& 0.4 $\pm$ 1.9& 0.4 $\pm$ 1.9& 0.4 $\pm$ 1.9\\ 
		& 3.2 $\pm$ 6.4& 1.3 $\pm$ 4.7& 0.7 $\pm$ 3.7& 4.4 $\pm$ 7.5& 0.9 $\pm$ 3.2& 0.4 $\pm$ 2.4& 4.8 $\pm$ 8.2& 0.7 $\pm$ 2.8& 0.3 $\pm$ 1.9\\\hline 
		10 & 6.8 $\pm$ 3.9& 6.8 $\pm$ 3.9& 6.8 $\pm$ 3.9& 4.6 $\pm$ 4.4& 4.6 $\pm$ 4.4& 4.6 $\pm$ 4.4& 2.8 $\pm$ 4.2& 2.8 $\pm$ 4.2& 2.8 $\pm$ 4.2\\ 
		& 10 $\pm$ 3.2& 9.3 $\pm$ 2.9& 9 $\pm$ 2.9& 10.3 $\pm$ 3.9& 9 $\pm$ 2.3& 8.5 $\pm$ 2.6& 10.6 $\pm$ 4.7& 8.7 $\pm$ 2.4& 8 $\pm$ 2.8\\\hline 
		50 & 9.8 $\pm$ 1.6& 9.8 $\pm$ 1.6& 9.8 $\pm$ 1.6& 9.6 $\pm$ 1& 9.6 $\pm$ 1& 9.6 $\pm$ 1& 9.4 $\pm$ 1.4& 9.4 $\pm$ 1.4& 9.4 $\pm$ 1.4\\ 
		& 10.7 $\pm$ 3.7& 10.3 $\pm$ 2.9& 10.1 $\pm$ 2.5& 11.2 $\pm$ 4.4& 10.1 $\pm$ 2.1& 9.9 $\pm$ 1.5& 11.5 $\pm$ 5.2& 10 $\pm$ 1.9& 9.8 $\pm$ 1.2\\\hline 
		200 & 10 $\pm$ 1.2& 10 $\pm$ 1.2& 10 $\pm$ 1.2& 9.9 $\pm$ 0.7& 9.9 $\pm$ 0.7& 9.9 $\pm$ 0.7& 9.9 $\pm$ 0.8& 9.9 $\pm$ 0.8& 9.9 $\pm$ 0.8\\ 
		& 10.9 $\pm$ 3.7& 10.3 $\pm$ 2.6& 10.2 $\pm$ 2.2& 11.2 $\pm$ 4.1& 10.2 $\pm$ 1.8& 10.1 $\pm$ 1.3& 11.5 $\pm$ 4.8& 10.1 $\pm$ 1.5& 10 $\pm$ 1\\\hline 	\end{tabular}
\end{table}

\begin{table}[H]
	\centering
	\scriptsize
	\begin{tabular}{|c|ccc|ccc|ccc|}
		\hline
		& \multicolumn{3}{c|}{$k$=100} & \multicolumn{3}{c|}{$k$=150} & \multicolumn{3}{c|}{$k$=200}	\\
		$C$ & $k^*=$100 & $k^*=$150 & $k^*=$200 &	$k^*=$100 & $k^*=$150 & $k^*=$200&$k^*=$100 & $k^*=$150 & $k^*=$200\\\hline
		0.001 & 9.7 $\pm$ 1.7& 9.7 $\pm$ 1.7& 9.7 $\pm$ 1.7& 9.8 $\pm$ 1.1& 9.8 $\pm$ 1.1& 9.8 $\pm$ 1.1& 9.8 $\pm$ 1& 9.8 $\pm$ 1& 9.8 $\pm$ 1\\ 
		& 9.4 $\pm$ 4.1& 9.5 $\pm$ 3.5& 9.6 $\pm$ 2.8& 9.5 $\pm$ 4.2& 9.4 $\pm$ 2.6& 9.6 $\pm$ 2& 9.8 $\pm$ 4.5& 9.4 $\pm$ 2.6& 9.5 $\pm$ 1.6\\\hline 
		0.01 & 5.9 $\pm$ 3.5& 5.9 $\pm$ 3.5& 5.9 $\pm$ 3.5& 6.3 $\pm$ 2.8& 6.3 $\pm$ 2.8& 6.3 $\pm$ 2.8& 6.7 $\pm$ 2.7& 6.7 $\pm$ 2.7& 6.7 $\pm$ 2.7\\ 
		& 4.3 $\pm$ 6.3& 4.4 $\pm$ 5.3& 4.8 $\pm$ 4.8& 4.5 $\pm$ 6.4& 4.3 $\pm$ 4.7& 4.8 $\pm$ 4.1& 4.8 $\pm$ 6.7& 4.3 $\pm$ 4.6& 4.9 $\pm$ 3.9\\\hline 
		1 & 0.4 $\pm$ 2.8& 0.4 $\pm$ 2.8& 0.4 $\pm$ 2.8& 0.3 $\pm$ 1.6& 0.3 $\pm$ 1.6& 0.3 $\pm$ 1.6& 0.3 $\pm$ 1.8& 0.3 $\pm$ 1.8& 0.3 $\pm$ 1.8\\ 
		& 2.3 $\pm$ 5.9& 1.1 $\pm$ 4.7& 0.7 $\pm$ 4.1& 2.8 $\pm$ 6.3& 0.9 $\pm$ 3.9& 0.6 $\pm$ 3.4& 3.3 $\pm$ 6.9& 0.9 $\pm$ 3.8& 0.5 $\pm$ 2.6\\\hline 
		10 & 6.6 $\pm$ 4.2& 6.6 $\pm$ 4.2& 6.6 $\pm$ 4.2& 5.6 $\pm$ 4.4& 5.6 $\pm$ 4.4& 5.6 $\pm$ 4.4& 4.6 $\pm$ 4.5& 4.6 $\pm$ 4.5& 4.6 $\pm$ 4.5\\ 
		& 9.6 $\pm$ 3.5& 8.8 $\pm$ 3.4& 8.2 $\pm$ 3.6& 9.7 $\pm$ 3.8& 8.6 $\pm$ 3.2& 7.8 $\pm$ 3.6& 9.9 $\pm$ 4.3& 8.4 $\pm$ 3.3& 7.4 $\pm$ 3.8\\\hline 
		50 & 9.8 $\pm$ 1.4& 9.8 $\pm$ 1.4& 9.8 $\pm$ 1.4& 9.7 $\pm$ 1.2& 9.7 $\pm$ 1.2& 9.7 $\pm$ 1.2& 9.6 $\pm$ 1.4& 9.6 $\pm$ 1.4& 9.6 $\pm$ 1.4\\ 
		& 10.5 $\pm$ 3.4& 10.2 $\pm$ 2.7& 10 $\pm$ 2.3& 10.6 $\pm$ 3.3& 10.1 $\pm$ 2.2& 9.9 $\pm$ 1.8& 10.8 $\pm$ 4& 10 $\pm$ 2.1& 9.8 $\pm$ 1.5\\\hline 
		200 & 10 $\pm$ 1.2& 10 $\pm$ 1.2& 10 $\pm$ 1.2& 10 $\pm$ 0.6& 10 $\pm$ 0.6& 10 $\pm$ 0.6& 10 $\pm$ 0.6& 10 $\pm$ 0.6& 10 $\pm$ 0.6\\ 
		& 10.7 $\pm$ 3.6& 10.4 $\pm$ 2.7& 10.2 $\pm$ 2.3& 10.8 $\pm$ 3.5& 10.2 $\pm$ 2.1& 10.1 $\pm$ 1.7& 11 $\pm$ 3.9& 10.2 $\pm$ 1.8& 10.1 $\pm$ 1.3\\\hline	\end{tabular}
	\caption{$k_0=10$ scaled outliers. {\em Top:} Lognormal(0,1). {\em Middle:} $\modn(0,1)$. {\em Bottom:} Weibull(1,1).}
	\label{tab:xi-zero-ha-scl}
\end{table}

\subsection*{A.2. Proofs of the Theorems}

\vspace{2mm}

{\bf Proof of Theorem 2.1} \\

\noindent Considering $X_{i,n}\stackrel{d}{=} U(Y_{i,n})$, where  $Y_{1,n}\leq \cdots \leq Y_{n,n}$ denote the order statistics of an i.i.d. sample from the standard Pareto distribution, i.e. setting $\ell=1$ and $\xi=1$ in \eqref{Patype}, then 
\begin{equation}
\label{e:xi-trim-u}
H_{k_0,k}=\frac{k_0}{k-k_0}\log \frac{U(Y_{n-k_0,n})}{U(Y_{n-k,n})}+\frac{1}{k-k_0}\sum_{i=k_0}^{k-1} \log \frac{U(Y_{n-i,n})}{U(Y_{n-k,n})}.
\end{equation}
Based on \eqref{e:bounds}	we have
\begin{equation}
\label{e:log-u}
\frac{\log U(tx)-\log U(t)}{q_0(t)}=\frac{x^{\xi_{-}}-1}{\xi_{-}}+Q_0(t)\Phi_{\xi_{-}, \rho}(x)+Q_0(t)R(t,x)
\end{equation}
where $\xi_-=\min (\xi,0)$ and the remainder $|R(t,x)| \leq \epsilon x^{\xi_{-}+\rho +\delta}$ for all $t \geq t_0$ and $x>1$.
\\
Since $Y_{n-k,n} \stackrel{\mathbb{P}}{\rightarrow }\infty$,  we use \eqref{e:log-u} with $t=Y_{n-k,n}$ and $x=Y_{n-i,n}/Y_{n-k,n}$, so that with probability tending to 1
\begin{align*}
\frac{(k-k_0)H_{k_0,k}}{k q_0(Y_{n-k,n})}&=\underbrace{\frac{k_0}{k}\Psi_{\xi_{-}}\Big(\frac{Y_{n-k_0,n}}{Y_{n-k,n}}\Big)+ \frac{1}{k}\sum_{i=k_0}^{k-1}\Psi_{\xi_{-}}\Big(\frac{Y_{n-i,n}}{Y_{n-k,n}}\Big)}_{A_{k_0,k}(n)}\\
&+Q_0(Y_{n-k,n})\underbrace{\Bigg(\frac{k_0}{k}\Phi_{\xi_{-}, \rho}\Big(\frac{Y_{n-k_0,n}}{Y_{n-k,n}}\Big)+\frac{1}{k}\sum_{i=k_0}^{k-1}\Phi_{\xi_{-}, \rho}\Big(\frac{Y_{n-i,n}}{Y_{n-k,n}}\Big)\Bigg)}_{B_{k_0,k}(n)}
\end{align*}
\begin{align}
\label{e:xi-trim-y}
\hspace{40mm}+Q_0(Y_{n-k,n})\underbrace{\Bigg(\frac{k_0}{k}R\Big(Y_{n-k,n},\frac{Y_{n-k_0,n}}{Y_{n-k,n}}\Big)+\frac{1}{k}\sum_{i=k_0}^{k-1}R\Big(Y_{n-k,n},\frac{Y_{n-i,n}}{Y_{n-k,n}}\Big)\Bigg)}_{C_{k_0,k}(n)}.
\end{align}
Since $Y_{n-k,n}/(n/k)\stackrel{\mathbb{P}}{\longrightarrow}1$ and $Q_0$ is regularly varying with index $\rho$, $Q_0(Y_{n-k,n}) / Q_0(n/k)\stackrel{\mathbb{P}}{\longrightarrow}1$. By \eqref{e:q-asy} and Corollary  2.3.5, Theorem 2.3.6 and assumption (3.5.14) in \cite{de2006}
$$\lim_{k\rightarrow \infty}\sqrt{k}Q_0\Big(\frac{n}{k}\Big)=\lim_{k\rightarrow \infty}\sqrt{k}Q\Big(\frac{n}{k}\Big)\Big(\frac{1}{\rho} \mathbbm{1}_{\{\rho<0\}}+ \mathbbm{1}_{\{\rho=0\}}\Big)=\lambda\Big(\frac{1}{\rho}\mathbbm{1}_{\{\rho<0\}}+ \mathbbm{1}_{\{\rho=0\}}\Big).$$
Therefore,
\begin{equation}
\label{e:q0-lim}
\sqrt{k}Q_0(Y_{n-k,n})= \lambda\Big(\frac{1}{\rho}\mathbbm{1}_{\{\rho<0\}}+ \mathbbm{1}_{\{\rho=0\}}\Big)+o_{\mathbb{P}}(1).
\end{equation}
We next consider the asymptotic behaviour of $A_{k_0,k}(n)$, $B_{k_0,k}(n)$ and $C_{k_0,k}(n)$.\hspace{35mm} $\Box$

\begin{proposition}
	\label{prop:a-def}
	$$A_{k_0,k}(n)\stackrel{d}{=}\begin{cases}
	\frac{1}{k}\sum_{i=k_0+1}^{k}Z_i,& \xi\geq 0\\
	\frac{k_0}{k}(\exp(\xi \sum_{j=k_0+1}^{k}Z_j/j)-1)/\xi+ \frac{1}{k}\sum_{i=k_0}^{k-1}(\exp(\xi \sum_{j=i+1}^{k}Z_j/j)-1)/\xi, & \xi<0
	\end{cases}$$
\end{proposition}

\begin{proof}
	$$A_{k_0,k}(n)=	\frac{k_0}{k}\Psi_{\xi_{-}}\Big(\frac{Y_{n-k_0,n}}{Y_{n-k,n}}\Big)+ \frac{1}{k}\sum_{i=k_0}^{k-1}\Psi_{\xi_{-}}\Big(\frac{Y_{n-i,n}}{Y_{n-k,n}}\Big).$$
	
	\noindent
	When $\xi \geq 0$ we have $\Psi_{\xi_{-}}(x)=\log x$, and
	\begin{align}
	\label{e:a-def-1}
	A_{k_0,k}(n)&=\frac{k_0}{k} \log \frac{Y_{n-k_0,n}}{Y_{n-k,n}}+\frac{1}{k}\sum_{i=k_0}^{k-1}  \log \frac{Y_{n-i,n}}{Y_{n-k,n}}\\
	\nonumber
	&\stackrel{d}{=}\frac{k_0}{k} (E_{n-k_0,n}-E_{n-k,n})+\frac{1}{k}\sum_{i=k_0}^{k-1}(E_{n-i,n}-E_{n-k,n})
	\end{align}
	where $E_{1,n}\leq \ldots \leq E_{n,n}$ denote the order statistics of an i.i.d. sample of size $n$ from the standard exponential distribution.
	\\
	Using the  R\'enyi representation  of exponential order statistics (see section 4.4 in \citep{beir2004})  $E_{n-i,n}=\sum_{j=i+1}^{n} Z_j/j$ where $Z_1, Z_2,\ldots$ are i.i.d. standard exponential rv's,  now yields 
	\begin{align*}
	A_{k_0,k}(n)&=\frac{k_0}{k} \Big(\sum_{j=k_0+1}^{n} \frac{Z_j}{j}-\sum_{j=k+1}^{n} \frac{Z_j}{j}\Big)+\frac{1}{k}\sum_{i=k_0}^{k-1}\Big(\sum_{j=i+1}^{n} \frac{Z_j}{j}-\sum_{j=k+1}^{n} \frac{Z_j}{j}\Big)\\
	&=\frac{k_0}{k}\sum_{j=k_0+1}^{k}\frac{Z_j}{j}+\frac{1}{k}\sum_{i=k_0}^{k-1}\sum_{j=i+1}^{k}\frac{Z_j}{j}.
	\end{align*}
	Interchanging the order of summation in $i$ and $j$ in the second summand, we obtain
	\begin{align*}
	A_{k_0,k}(n)&=\frac{k_0}{k}\sum_{j=k_0+1}^{k}\frac{Z_j}{j}+\frac{1}{k}\sum_{j=k_0+1}^{k} \sum_{i=k_0}^{j-1} \frac{Z_j}{j}=\frac{1}{k}\sum_{j=k_0+1}^{k} Z_j
	\end{align*}
	which completes the proof for case $\xi\geq 0$.
	
	\noindent
	When $\xi<0$ we have $\Psi_{\xi_{-}}(x)= \Psi_{\xi}(x)$ leading to
	\begin{align*}
	A_{k_0,k}(n)&=\frac{k_0}{k \xi} \Big(\Big(\frac{Y_{n-k_0,n}}{Y_{n-k,n}}\Big)^{\xi}-1\Big)+\frac{1}{k\xi}\sum_{i=k_0}^{k-1}\Big(   \Big(\frac{Y_{n-i,n}}{Y_{n-k,n}}\Big)^{\xi}-1\Big)
	\end{align*}	
	\begin{align}
	\label{e:a-def-2}
	\nonumber
	&=\frac{k_0}{k\xi}  \Big(\exp \Big(\xi \log \frac{Y_{n-k_0,n}}{Y_{n-k,n}}\Big) -1\Big)+\frac{1}{k\xi}  \sum_{i=k_0}^{k-1}\Big(\exp \Big(\xi \log \frac{Y_{n-i,n}}{Y_{n-k,n}}\Big)-1\Big)\\
	&\stackrel{d}{=}\frac{k_0}{k\xi}  \Big(\exp (\xi [E_{n-k_0,n}- E_{n-k,n}]) -1\Big)+\frac{1}{k \xi}\sum_{i=k_0}^{k-1}\Big(\exp (\xi[ E_{n-i,n}- E_{n-k,n}]) -1\Big).
	\end{align}
	Again using the R\'enyi representation gives
	\begin{align*}
	A_{k_0,k}(n)&=\frac{k_0}{k}  \Big(\exp (\xi \sum_{j=k_0+1}^{k}Z_j/j) -1\Big)/\xi+\frac{1}{k}\sum_{i=k_0}^{k-1}\Big(\exp (\xi \sum_{j=i+1}^{k}Z_j/j) -1\Big)/\xi,
	\end{align*}
	which completes the proof for case $\xi<0$.\end{proof}

\begin{proposition}
	\label{lem:b-conv}
	As $k,n \to \infty$, $k/n \to 0$ and $k_0=o(k)$
	$$B_{k_0,k}(n)=c_{\xi_-,\rho}+o_{\mathbb{P}}(1)$$
	\begin{align}
	\label{e:bias-1-def}
	c_{\xi_-,\rho}=\begin{cases}
	\frac{1}{1-\xi_--\rho}(1-{(\frac{k_0}{k})}^{1-\xi_--\rho}), &\rho<0\\
	\frac{1}{\xi_{-}(1-\xi_{-1})^2}(1-{(\frac{k_0}{k})}^{1-\xi_{-}})-\frac{1}{1-\xi_{-}}{(\frac{k_0}{k})}^{1-\xi_{-}}\log\frac{k_0}{k} , &\rho=0, \:\xi_{-}<0\\
	(1-\frac{k_0}{k})-\frac{k_0}{k}\log\frac{k_0}{k}, &\rho=0, \:\xi_{-}=0.
	\end{cases}
	\end{align}
	
	Hence, with  \eqref{e:q0-lim},
	$$Q_0(Y_{n-k,n})B_{k_0,k}(n)=\lambda\Big(\frac{1}{\rho}\mathbbm{1}_{\{\rho<0\}}+ \mathbbm{1}_{\{\rho=0\}}\Big)\frac{ c_{\xi_{-},\rho}}{\sqrt{k}}+o_{\mathbb{P}}(k^{-1/2})$$
\end{proposition}

\begin{proof}
	$$B_{k_0,k}(n)=\frac{k_0}{k}\Phi_{\xi_{-}, \rho}\Big(\frac{Y_{n-k_0,n}}{Y_{n-k,n}}\Big)+\frac{1}{k}\sum_{i=k_0}^{k-1}\Phi_{\xi_{-}, \rho}\Big(\frac{Y_{n-i,n}}{Y_{n-k,n}}\Big).$$
	
	\noindent{\em In case $\rho<0$,} then
	$$B_{k_0,k}(n)=\frac{1}{\xi_{-}+\rho}\Bigg(\frac{k_0}{k}\Big(\frac{Y_{n-k_0,n}}{Y_{n-k,n}}\Big)^{\xi_{-}+\rho}+\frac{1}{k}\sum_{i=k_0}^{k-1}\Big(\frac{Y_{n-i,n}}{Y_{n-k,n}}\Big)^{\xi_{-}+\rho}-1\Bigg).$$
	By Lemma \ref{lem:r-def1} below, $B_{k_0,k}=1/(1-\xi_{-}-\rho)(1-(k_0/k)^{1-\xi_{-}-\rho)}+o_{\mathbb{P}}(1)$.
	
	\vspace{0.3cm}
	\noindent{\em In case $\xi_{-}< 0,\:\rho=0$, }
	$$B_{k_0,k}(n)=\frac{k_0}{k\xi_{-}}\Big(\frac{Y_{n-k_0,n}}{Y_{n-k,n}}\Big)^{\xi_{-}}\log\frac{Y_{n-k_0,n}}{Y_{n-k,n}}+\frac{1}{k\xi_{-}}\sum_{i=k_0}^{k-1}\Big(\frac{Y_{n-i,n}}{Y_{n-k,n}}\Big)^{\xi_{-}}\log\frac{Y_{n-i,n}}{Y_{n-k,n}}$$
	By Lemma \ref{lem:r-def2} below, $B_{k_0,k}=1/(\xi_{-}(1-\xi_{-})^2)(1-(k_0/k)^{\xi_{-}})-1/(1-\xi_{-})((k_0/k)^{\xi_{-}}\log( k_0/k))+o_{\mathbb{P}}(1).$
	
	\vspace{0.3cm}
	\noindent{\em In case $\xi_{-}= 0, \rho=0$,}
	$$B_{k_0,k}(n)=\frac{k_0}{2k}\Big(\log\frac{Y_{n-k_0,n}}{Y_{n-k,n}}\Big)^2+\frac{1}{2k}\sum_{i=k_0}^{k-1}\Big(\log\frac{Y_{n-i,n}}{Y_{n-k,n}}\Big)^2.$$
	By Lemma \ref{lem:r-def3}	below, $B_{k_0,k}=1-(k_0/k)-(k_0/k)\log (k_0/k)+o_{\mathbb{P}}(1).$
	
\end{proof}

\noindent The following Proposition concerning $C_{k_0,k}(n)$ in  \eqref{e:xi-trim-y}  ends the proof of Theorem 2.1.
\noindent
\begin{proposition}
	\label{lem:c-conv}
	As $k,n \to \infty$, $k/n \to 0$ and $k_0=o(k)$
	$$C_{k_0,k}(n)=o_{\mathbb{P}}(1).$$
	Thus, in view of  \eqref{e:q0-lim}, 
	$$Q_0(Y_{n-k,n})C_{k_0,k}(n)=o_{\mathbb{P}}(k^{-1/2}).$$
\end{proposition} 

\begin{proof}
	$$C_{k_0,k}(n)=\frac{k_0}{k}R\Big(Y_{n-k,n},\frac{Y_{n-k_0,n}}{Y_{n-k,n}}\Big)+\frac{1}{k}\sum_{i=k_0}^{k-1}R\Big(Y_{n-k,n},\frac{Y_{n-i,n}}{Y_{n-k,n}}\Big).$$
	Since $Y_{n-k,n} \stackrel{\mathbb{P}}{\longrightarrow}1$,  with probability tending to 1 we have
	$$\Big|R\Big(Y_{n-k,n},\frac{Y_{n-i,n}}{Y_{n-k,n}}\Big)\Big|\leq \epsilon\Big(\frac{Y_{n-i,n}}{Y_{n-k,n}}\Big)^{\xi_{-}+\rho+\delta},\hspace{2mm}i=k_0, \ldots,k-1.$$
	Therefore,
	\begin{align}\label{e:c-def}|C_{k_0,k}(n)|&\leq\epsilon\underbrace{ \Big| \frac{k_0}{k}\Big(\frac{Y_{n-k_0,n}}{Y_{n-k,n}}\Big)^{\xi_{-}+\rho+\delta}+\frac{1}{k}\sum_{i=k_0}^{k-1}\Big(\frac{Y_{n-i,n}}{Y_{n-k,n}}\Big)^{\xi_{-}+\rho+\delta}\Big|}_{C^*_{k_0,k}(n)}.
	\end{align}
	
	\noindent
	If $\xi_{-}+\rho<0$, then for $\delta$ sufficiently small, $\xi_{-}+\rho+\delta<0$.
	
	\noindent By Lemma \ref{lem:r-def1} below,
	$$C^*_{k_0,k}=1+\frac{1}{(1-\xi_{-}-\rho-\delta)}\Big(1-\Big(\frac{k_0}{k}\Big)^{1-\xi_{-}-\rho-\delta}\Big)+o_{\mathbb{P}}(1)$$
	
	\noindent
	If $\xi_{-}+\rho=0$, then for every $1>\delta>0$, 
	$$|C^*_{k_0,k}|\leq \frac{1}{k}\sum_{i=0}^{k-1}\Big(\frac{Y_{n-i,n}}{Y_{n-k,n}}\Big)^{\delta}\stackrel{d}{=}\frac{1}{k}\sum_{i=0}^{k-1}U_{i,k}^{-\delta}=\frac{1}{1-\delta}+o_{\mathbb{P}}(1)$$
	where $U_{1,k}\leq \ldots \leq U_{k,k}$ denote the order statistics of uniform (0,1) i.i.d. sample of length $k$ and the convergence  follows from the law of large numbers.
	\\
	Thus,  $C^*_{k_0,k}=O_{\mathbb{P}}(1)$ which in view of \eqref{e:c-def} implies $C_{k_0,k}=o_{\mathbb{P}}(1)$. 
\end{proof}

\noindent

\noindent
{\bf Proof of Theorem 2.2} 
Note that
\begin{equation}
\label{e:T-trim}
1-T_{k_0,k}=\frac{(k_0+1)\log \frac{X_{n-k_0,n}}{X_{n-k_0-1,n}}}{(k-k_0) H_{k_0,k}}=
\frac{V_{k_0+1}}{(k-k_0) H_{k_0,k}},
\end{equation}
with $V_{k_0+1}$ defined in \eqref{def:V}.\\

\noindent With the same notations as in the proof of Theorem 2.1, we have
\begin{eqnarray}
V_{k_0+1}&=& (k_0+1) \log \frac{U(Y_{n-k_0,n})}{U(Y_{n-k_0-1,n})}  \nonumber \\
&=& (k_0+1) \log \frac{U\Big({Y_{n-k_0,n} \over Y_{n-k_0-1,n}}Y_{n-k_0-1,n}\Big)}{U\Big(Y_{n-k_0-1,n}\Big)}. \label{e:logUdiff}
\end{eqnarray}
Using  \eqref{e:log-u} with $t=Y_{n-k_0-1,n}$ and $x={Y_{n-k_0,n} \over Y_{n-k_0-1,n}}$ we obtain
\begin{eqnarray}
V_{k_0+1}&=& (k_0+1) q_0(Y_{n-k_0-1,n})\left\{
\Psi_{\xi_-}\left({Y_{n-k_0,n} \over Y_{n-k_0-1,n}}\right)+Q_0 (Y_{n-k_0-1,n})\Phi_{\xi_-,\rho}\left({Y_{n-k_0,n} \over Y_{n-k_0-1,n}}\right)\right. \nonumber\\
&& \hspace{5cm}\left.
+ Q_0(Y_{n-k_0-1,n})R\left(Y_{n-k_0-1,n},{Y_{n-k_0,n} \over Y_{n-k_0-1,n}}\right)\right\},
\label{e:Vk0}
\end{eqnarray}
where $|R\left(Y_{n-k_0-1,n},{Y_{n-k_0,n} \over Y_{n-k_0-1,n}}\right)|\leq \epsilon \left( {Y_{n-k_0,n} \over Y_{n-k_0-1,n}}\right)^{\xi_-+\rho+\delta} $.\\

\noindent
Since $Y_{n-k_0,n}/(n/k_0) = O_{\mathbb{P}}(1)$ when $k_0=o(k)$, we have $$\Phi_{\xi_-,\rho}\left({Y_{n-k_0,n} \over Y_{n-k_0-1,n}}\right)=O_{\mathbb{P}}(1) \hspace{5mm} 
R\left(Y_{n-k_0-1,n},{Y_{n-k_0,n} \over Y_{n-k_0-1,n}}\right)=o_{\mathbb{P}}(1).$$\\

\noindent
A representation of the denominator of the right hand side of \eqref{e:T-trim} follows from Theorem 2.1. We now combine \eqref{e:logUdiff} with that result.
\\

\noindent
Since $Q_0$ is regularly varying with index $\rho \leq 0$, i.e. $Q_0 (x)=x^\rho \ell_0 (x)$ for some slowly varying function $\ell_0$, we invoke Potter bounds so that for any $\delta_1,\delta_2 >0$ and $t$ large enough
$$
{\ell_0 (tx) \over \ell_0 (t)} \leq (1+\delta_1)x^{\delta_2}, \text{ with } x>1,
$$ 
(see for instance Proposition B.1.9 in de Haan and Ferreira (2006)) so that for any $\delta >0$
\begin{eqnarray*}
	(k_0+1)\frac{Q_0(Y_{n-k_0-1,n})}{Q_0(Y_{n-k,n})}
	&=&
	(k_0+1)\frac{ Q_0\left({Y_{n-k_0-1,n}\over Y_{n-k,n} } Y_{n-k,n} \right)}{Q_0(Y_{n-k,n})} \\
	&=& (k_0+1)\left( {Y_{n-k_0-1,n}\over Y_{n-k,n} }\right)^\rho
	\; {\ell_0 \left( {Y_{n-k_0-1,n}\over Y_{n-k,n} }Y_{n-k,n}\right) \over \ell_0 (Y_{n-k,n})} \\
	&=& O_{\mathbb{P}}\left(k_0({k \over k_0})^{\rho+\delta} \right).
\end{eqnarray*}
Since $Y_{n-k,n}/(n/k) \to_{\mathbb{P}} 1$, $Q_0(Y_{n-k,n})= O_{\mathbb{P}}(1/\sqrt{k})$ using the assumption $\sqrt{k}Q_0(n/k) \to \lambda$, from which
\begin{equation}
(k_0+1)Q_0(Y_{n-k_0-1,n})= Q_0 (Y_{n-k,n}) O_{\mathbb{P}}\left(k_0({k \over k_0})^{\rho+\delta} \right) = O_{\mathbb{P}}\left(k^{-1/2}\,k_0\,({k \over k_0})^{\rho+\delta} \right).
\label{e:Q0k0}
\end{equation}

\vspace{0.5cm}\noindent
Furthermore, using  (2.3.18) in de Haan and Ferreira (2006) stating that
$$
{q_0 (tx) \over q_0(t)} = x^{\xi_-}\left(1+  Q_0(t) x^{\xi_-}
\Psi_\rho (x)+   Q_0(t)\tilde{R}(t,x)\right) ,
$$
with $|\tilde{R}(t,x)|\leq \epsilon x^{\xi_- +\rho+\delta}$,
we obtain with $t=Y_{n-k,n}$ and $x={Y_{n-k_0-1,n}\over Y_{n-k,n}}$ that
\begin{equation}
\left( {Y_{n-k,n} \over Y_{n-k_0-1,n}}\right)^{\xi_-}
{q_0 (Y_{n-k_0-1,n})\over q_0(Y_{n-k,n}) }
=
1+ Q_0(Y_{n-k,n})\Psi_\rho \left(  {Y_{n-k_0,n} \over Y_{n-k,n}}\right) + Q_0(Y_{n-k,n})\tilde{R}\left(Y_{n-k,n},{Y_{n-k_0-1,n}\over Y_{n-k,n}}\right).
\label{e:fracq0}
\end{equation}
Also 
\begin{equation}
\Psi_\rho \left(  {Y_{n-k_0,n} \over Y_{n-k,n}}\right)
= 
\begin{cases}
O_{\mathbb{P}}(1) \hspace{1.4cm} \mbox{ if } \rho <0, \\
O_{\mathbb{P}}(\log (k/k_0)) \mbox{ if } \rho =0.
\end{cases}
\label{e:Psi}
\end{equation}
Also,
\begin{equation}
|\tilde{R}\left(Y_{n-k,n},{Y_{n-k_0-1,n}\over Y_{n-k,n}}\right)| \leq \epsilon \left( {Y_{n-k_0-1,n} \over Y_{n-k,n}}\right)^{\rho+\delta} =o_{{\mathbb{P}}} \left( (k/k_0)^{\rho +\delta}\right).
\label{e:tR}
\end{equation}
Hence, combining \eqref{e:fracq0}, \eqref{e:Psi} and \eqref{e:tR} we obtain 
\
\begin{equation}
\left( {k_0+1 \over k} \right)^{\xi_-}{q_0 (Y_{n-k_0-1,n})\over q_0(Y_{n-k,n}) } = 1+ o_{{\mathbb{P}}} \left( k^{-1/2}(k/k_0)^{\rho +\delta}\right).
\label{e:q0R}
\end{equation}

\vspace{0.5cm}\noindent
Next, with the method of proof developed in Proposition 5.1,
\begin{equation}
(k_0+1) \Psi_{\xi_-}\left({Y_{n-k_0,n} \over Y_{n-k_0-1,n}}\right) =
\begin{cases}
Z_{k_0}, & \text{ if } \xi_- =0, \\
{k_0 +1 \over \xi_-}\left( e^{\xi_- {Z_{k_0}\over k_0+1}} -1\right), & \text{ if } \xi_- <0.
\end{cases}
\label{e:Psik0}
\end{equation}

\noindent	By Theorem \ref{thm:xi-dist},  	$(k-k_0)H_{k_0,k}/q_0(Y_{n-k,n})\stackrel{d}{=}k \mathbb{Z}_{k_0,k}+O_\mathbb{P}(k^{-1/2})$ where
\begin{equation}
\label{e:Z-}
\mathbb{Z}_{k_0,k}=\begin{cases}
\sum_{j=k_0+1}^kZ_j,& \xi_- =0\\
k_0(\exp(\xi_- \sum_{j=k_0+1}^{k}Z_j/j)-1)/\xi_-+ \sum_{i=k_0}^{k-1}(\exp(\xi_- \sum_{j=i+1}^{k}Z_j/j)-1)/\xi_-& \xi_-<0.
\end{cases}
\end{equation}

\noindent
The result now follows combining \eqref{e:T-trim},
\eqref{e:Vk0}, \eqref{e:Q0k0}, \eqref{e:q0R}, \eqref{e:Psik0} and \eqref{e:Z-}. \hspace{29mm} $\Box$

\begin{lemma}
	\label{lem:r-def1}
	With $Y_1, \ldots, Y_n$ an i.i.d. sample from the standard Pareto distribution, we have for any $\beta<0$,
	as $k,n \to \infty$, $k/n \to 0$ and $k_0=o(k)$
	\begin{equation}
	\label{e:R-conv-2}
	\Big|\frac{k_0}{k}\Big(\frac{Y_{n-k_0,n}}{Y_{n-k,n}}\Big)^{\beta}+\frac{1}{k}\sum_{i=k_0}^{k-1}\Big(\frac{Y_{n-i,n}}{Y_{n-k,n}}\Big)^{\beta}-1-\frac{\beta}{1-\beta}\Big(1-\Big(\frac{k_0}{k}\Big)^{1-\beta}\Big)\Big|\stackrel{\mathbb{P}}{\longrightarrow}0.
	\end{equation}
	
\end{lemma}

\begin{proof}
	\begin{eqnarray*}
		\frac{k_0}{k}\Big(\frac{Y_{n-k_0,n}}{Y_{n-k,n}}\Big)^{\beta}+\frac{1}{k}\sum_{i=k_0}^{k-1}\Big(\frac{Y_{n-i,n}}{Y_{n-k,n}}\Big)^{\beta}&=&\frac{1}{k}\sum_{i=0}^{k-1}\Big(\frac{Y_{n-i,n}}{Y_{n-k,n}}\Big)^{\beta}-
		\frac{1}{k}\sum_{i=0}^{k_0-1}\Bigg(
		\Big(\frac{Y_{n-i,n}}{Y_{n-k,n}}\Big)^{\beta}-\Big(\frac{Y_{n-k_0,n}}{Y_{n-k,n}}\Big)^{\beta}\Bigg)\\
		&\stackrel{d}{=}&\underbrace{\frac{1}{k}\sum_{i=1}^{k}\Big(\frac{\Gamma_{i}}{\Gamma_{k+1}}\Big)^{-\beta}}_{\Delta_{k}}-\underbrace{\frac{1}{k}\sum_{i=1}^{k_0}\Bigg(\Big(\frac{\Gamma_{i}}{\Gamma_{k+1}}\Big)^{-\beta}-\Big(\frac{\Gamma_{k_0+1}}{\Gamma_{k+1}}\Big)^{-\beta}\Bigg)}_{\Delta_{k_0,k}},
	\end{eqnarray*}
	where $\Gamma_j= \sum_{i=1}^j E_j$ ($j=1,2,\ldots$).
	\\
	By Lemma C.4 in \cite{bhatt2019},  $\Delta_k\stackrel{\mathbb{P}}{\longrightarrow}1/(1-\beta)$.
	\\
	If $k_0=k_0(n) \leq M$ for all $n$, then $\Delta_{k_0,k}=o_{\mathbb{P}}(1)$ and  hence the proof of  \eqref{e:R-conv-2} follows.
	\\
	If $ k_0=k_0(n) \to \infty$, then 
	\begin{equation}\label{e:del-conv}
	\Delta_{k_0,k}-\frac{\beta}{1-\beta}\Big(\frac{k_0}{k}\Big)^{1-\beta}=o_{\mathbb{P}}(1)
	\end{equation}
	since
	\begin{eqnarray}
	\label{e:b-2}
	\Big|\Delta_{k_0,k}-\frac{\beta}{1-\beta}\Big(\frac{k_0}{k}\Big)^{1-\beta}\Big|&=&   \Big|\frac{1}{k}\sum_{i=1}^{k_0}\Big(\Big(\frac{\Gamma_{i}}{\Gamma_{k+1}}\Big)^{-\beta}-\Big(\frac{\Gamma_{k_0+1}}{\Gamma_{k+1}}\Big)^{-\beta}\Big)-{\beta \over 1-\beta}{\Big(\frac{k_0}{k}\Big)}^{1-\beta}\Big|\\\nonumber
	&=&{\Big(\frac{k_0}{k}\Big)}^{1-\beta}\Big|\Big(\frac{\Gamma_{k_0+1}/k_0}{\Gamma_{k+1}/k}\Big)^{-\beta}\Big(\underbrace{\frac{1}{k_0}\sum_{i=1}^{k_0}\Big(\frac{\Gamma_{i}}{\Gamma_{k_0+1}}\Big)^{-\beta}-1}_{\Delta_{k_0}}\Big)-\frac{\beta}{1-\beta}\Big|.\\\nonumber
	\end{eqnarray}
	where $((k\Gamma_{k_0+1})/(k_0\Gamma_{k+1}))^{-\beta}\stackrel{\mathbb{P}}{\longrightarrow}1$ and $\Delta_{k_0}\stackrel{\mathbb{P}}{\longrightarrow}\beta/(1-\beta)$ and   follows from Lemmas C.3 and C.4 in \cite{bhatt2019} respectively.  This completes the proof.
\end{proof}

\begin{lemma}
	\label{lem:r-def2}
	With $\beta<0$ and  $$B_{k_0,k}(n)=\frac{k_0}{k}\Big(\frac{Y_{n-k_0,n}}{Y_{n-k,n}}\Big)^{\beta}\log\frac{Y_{n-k_0,n}}{Y_{n-k,n}}+\frac{1}{k}\sum_{i=k_0}^{k-1}\Big(\frac{Y_{n-i,n}}{Y_{n-k,n}}\Big)^{\beta}\log\frac{Y_{n-i,n}}{Y_{n-k,n}},$$ we have as $k,n \to \infty$, $k/n \to 0$ and $k_0=o(k)$
	\begin{equation}
	\label{e:b-gam-log}
	\Big|B_{k_0,k}(n)-\frac{1}{(1-\beta)^2}+\frac{\beta}{1-\beta}\Big(\frac{k_0}{k}\Big)^{1-\beta}\Big(\frac{1}{\beta(1-\beta)}+\log \frac{k}{k_0}\Big)\Big|\stackrel{\mathbb{P}}{\longrightarrow}0.
	\end{equation}
\end{lemma}

\begin{proof}
	Here
	\begin{eqnarray*}
		B_{k_0,k}(n)&\stackrel{d}{=}&\frac{k_0}{k}\Big(\frac{\Gamma_{k_0+1}}{\Gamma_{k+1}}\Big)^{-\beta}\log \frac{\Gamma_{k+1}}{\Gamma_{i}}+\frac{1}{k}\sum_{i=k_0+1}^k\Big(\frac{\Gamma_{i}}{\Gamma_{k+1}}\Big)^{-\beta}\log \frac{\Gamma_{k+1}}{\Gamma_{i}}\\
		&=&\underbrace{\frac{1}{k}\sum_{i=1}^k\Big(\frac{\Gamma_{i}}{\Gamma_{k+1}}\Big)^{-\beta}\log \frac{\Gamma_{k+1}}{\Gamma_{i}}}_{\Delta_{k}}-\underbrace{\frac{1}{k}\sum_{i=1}^{k_0}\Big(\Big(\frac{\Gamma_{i}}{\Gamma_{k+1}}\Big)^{-\beta}\log \frac{\Gamma_{k+1}}{\Gamma_{i}}-\Big(\frac{\Gamma_{k_0+1}}{\Gamma_{k+1}}\Big)^{-\beta}\log \frac{\Gamma_{k+1}}{\Gamma_{k_0+1}}\Big)}_{\Delta_{k_0,k}}.
	\end{eqnarray*}
	By Lemma C.4 in \cite{bhatt2019}, $\Delta_k\stackrel{\mathbb{P}}{\longrightarrow}1/(1-\beta)^2$ as $k \to \infty$.
	\\	
	
	\noindent
	We next show that 
	\begin{equation}
	\label{e:del}
	\Big|\Delta_{k_0,k}-\frac{1}{(1-\beta)^2}\Big(\frac{k_0}{k}\Big)^{1-\beta}-\frac{\beta}{1-\beta}\Big(\frac{k_0}{k}\Big)^{1-\beta}\log \frac{k}{k_0}\Big|\stackrel{\mathbb{P}}{\longrightarrow}0,
	\end{equation}
	which will complete the proof of the Lemma. \\
	If $k_0=k_0(n) \leq M$ for all $n$, then using $\Gamma_i/\Gamma_{k}=O_\mathbb{P}(1/k)$, one can show $\Delta_{k_0,k}=o_{\mathbb{P}}(1)$.
	
	\noindent	In case $k_0=k_0(n) \to \infty$, we write
	\begin{align*}
	\Delta_{k_0,k}&=\Big(\frac{\Gamma_{k_0+1}}{\Gamma_{k+1}}\Big)^{-\beta}\frac{1}{k}\sum_{i=1}^{k_0}\Big(\Big(\frac{\Gamma_{i}}{\Gamma_{k_0+1}}\Big)^{-\beta}\log \frac{\Gamma_{k+1}}{\Gamma_{i}}-\log \frac{\Gamma_{k+1}}{\Gamma_{k_0+1}}\Big)\\
	&=\underbrace{\Big(\frac{\Gamma_{k_0+1}}{\Gamma_{k+1}}\Big)^{-\beta}\frac{1}{k}\sum_{i=1}^{k_0}\Big(\frac{\Gamma_{i}}{\Gamma_{k_0+1}}\Big)^{-\beta}\log \frac{\Gamma_{k_0+1}}{\Gamma_{i}}}_{\Delta_{1,k_0,k}}
	+
	\underbrace{\Big(\frac{\Gamma_{k_0+1}}{\Gamma_{k+1}}\Big)^{-\beta}\log \frac{\Gamma_{k+1}}{\Gamma_{k_0+1}}\frac{1}{k}\sum_{i=1}^{k_0}\Big(\Big(\frac{\Gamma_{i}}{\Gamma_{k_0+1}}\Big)^{-\beta}-1\Big)}_{\Delta_{2,k_0,k}},
	\end{align*}
	where 
	$$\Big|\Delta_{1,k_0,k}-\frac{1}{(1-\beta)^2}\Big(\frac{k_0}{k}\Big)^{1-\beta}\Big|=\Big(\frac{k_0}{k}\Big)^{1-\beta}\Big|\Big(\frac{\Gamma_{k_0+1}/k_0}{\Gamma_{k+1}/k}\Big)^{-\beta}\frac{1}{k_0}\sum_{i=1}^{k_0}\Big(\frac{\Gamma_{k_0+1}}{\Gamma_{i}}\Big)^{-\beta}\log \frac{\Gamma_{i}}{\Gamma_{k_0+1}}-\frac{1}{(1-\beta)^2}\Big|.$$
	By Lemma C.3 in \cite{bhatt2019}, $((k\Gamma_{k_0+1})/(k_0\Gamma_{k+1}))^{-\beta}\stackrel{\mathbb{P}}{\longrightarrow}1$. Using ideas similar to proof of Lemma C.4 in \cite{bhatt2019},  $(1/k_0)\sum_{i=1}^{k_0} (\Gamma_{k_0+1}/\Gamma_i)^{-\beta}\log (\Gamma_i/\Gamma_{k_0+1}) \stackrel{\mathbb{P}}{\longrightarrow} 1/(1-\beta)^2$.  Therefore,
	\begin{equation}
	\label{e:del-1}
	\Delta_{1,k_0,k}-\frac{1}{(1-\beta)^2}\Big(\frac{k_0}{k}\Big)^{1-\beta}=o_{\mathbb{P}}(1).
	\end{equation}
	Furthermore,
	$$\Big|\Delta_{2,k_0,k}-\frac{\beta}{1-\beta}\Big(\frac{k_0}{k}\Big)^{1-\beta}\log \frac{k}{k_0}\Big|=\delta_{k_0,k}\Big|\Big(\frac{\Gamma_{k_0+1}/k_0}{\Gamma_{k+1}/k}\Big)^{-\beta}\frac{\log(\Gamma_{k+1}/\Gamma_{k_0+1})}{\log (k/k_0)}\underbrace{\frac{1}{k_0}\sum_{i=1}^{k_0}\Big(\Big(\frac{\Gamma_{i}}{\Gamma_{k_0+1}}\Big)^{-\beta}-1\Big)}_{\Delta_{2,k_0}} -\frac{\beta}{1-\beta}\Big|$$
	where $\delta_{k_0,k}=(k_0/k)^{1-\beta}\log (k/k_0) \leq 1$.
	
	\noindent If $k_0=k_0(n) \leq M$ for all $n$, then $\Delta_{2,k_0,k}=o_{\mathbb{P}}(1)$ since then $\Gamma_{k_0+1}/\Gamma_{k+1}= O(1/k)$.
	
	\noindent In case $k_0=k_0(n) \to \infty$, by Lemma C.4 in \cite{bhatt2019}, $\Delta_{2,k_0}=\beta/(1-\beta)+o_\mathbb{P}(1)$. Using ideas similar to proof of Lemma C.3 in \cite{bhatt2019},  $$\Big(\frac{\Gamma_{k_0+1}/k_0}{\Gamma_{k+1}/k}\Big)^{-\beta}\log\frac{\Gamma_{k+1}}{\Gamma_{k_0+1}}\Big(\log \frac{k}{k_0}\Big)^{-1}\stackrel{\mathbb{P}}{\longrightarrow} 1.$$
	Therefore,
	\begin{equation}
	\label{e:del-2}
	\Delta_{2,k_0,k}-\frac{\beta}{(1-\beta)}\Big(\frac{k_0}{k}\Big)^{1-\beta}\log \frac{k}{k_0}=o_{\mathbb{P}}(1).
	\end{equation}
	In view of  \eqref{e:del-1} and \eqref{e:del-2}, the proof of  \eqref{e:del} follows.
\end{proof}

\begin{lemma}
	\label{lem:r-def3}
	For $\beta<0$ and  $$B_{k_0,k}(n)=\frac{k_0}{2k}\Big(\log\frac{Y_{(n-k_0,n)}}{Y_{n-k,n}}\Big)^{2}+\frac{1}{2k}\sum_{i=k_0}^{k-1}\Big(\log\frac{Y_{(n-i,n)}}{Y_{n-k,n}}\Big)^{2},$$
	we have as $k,n \to \infty$, $k/n \to 0$ and $k_0=o(k)$
	\begin{equation}
	\label{e:b-gam-log}
	\Big|B_{k_0,k}(n)-1+\frac{k_0}{k}\Big(1+\log\frac{k}{k_0}\Big)\Big|\stackrel{\mathbb{P}}{\longrightarrow}0.
	\end{equation}
\end{lemma}

\begin{proof}
	Here
	\begin{eqnarray*}
		B_{k_0,k}(n)&\stackrel{d}{=}&\frac{k_0}{2k}\Big(\log \frac{\Gamma_{k+1}}{\Gamma_{k_0+1}}\Big)^2+\frac{1}{2k}\sum_{i=k_0+1}^k\Big(\log \frac{\Gamma_{k+1}}{\Gamma_{i}}\Big)^2\\
		&=&\underbrace{\frac{1}{2k}\sum_{i=1}^k\Big(\log \frac{\Gamma_{k+1}}{\Gamma_{i}}\Big)^2}_{\Delta_{k}}-\underbrace{\frac{1}{2k}\sum_{i=1}^{k_0}\Big(\Big(\log \frac{\Gamma_{k+1}}{\Gamma_{i}}\Big)^2-\Big(\log \frac{\Gamma_{k+1}}{\Gamma_{k_0+1}}\Big)^2\Big)}_{\Delta_{k_0,k}}.\\
	\end{eqnarray*}
	By using ideas similar to proof of Lemma C.4 in \cite{bhatt2019}, $\Delta_k \stackrel{\mathbb{P}}{\longrightarrow} 1$.
	
	\noindent If $k_0=k_0(n) \leq M$ for all $n$, then using $\Gamma_i/\Gamma_k=O_\mathbb{P}(1/k)$, one can show $\Delta_{k_0,k}=o_{\mathbb{P}}(1)$.
	
	\noindent If $k_0=k_0(n)\rightarrow \infty$,
	\begin{equation}
	\label{e:del-new}
	\Big|\Delta_{k_0,k}-\frac{k_0}{k}\Big(1+\log\frac{k}{k_0}\Big)\Big|\stackrel{\mathbb{P}}{\longrightarrow}0.
	\end{equation}
	To this end we write
	\begin{align*}
	\Delta_{k_0,k}&=\frac{1}{2k}\sum_{i=1}^{k_0}\Big(\Big(\log \frac{\Gamma_{k_0+1}}{\Gamma_{i}}+\log \frac{\Gamma_{k+1}}{\Gamma_{k_0+1}}\Big)^2-\Big(\log \frac{\Gamma_{k+1}}{\Gamma_{k_0+1}}\Big)^2\Big)\\
	&=\underbrace{\frac{1}{2k}\sum_{i=1}^{k_0}\Big(\log \frac{\Gamma_{k_0+1}}{\Gamma_{i}}\Big)^2}_{\Delta_{1,k_0,k}}
	+
	\underbrace{\log \frac{\Gamma_{k+1}}{\Gamma_{k_0+1}}\frac{1}{k}\sum_{i=1}^{k_0}\log \frac{\Gamma_{k_0+1}}{\Gamma_{i}}}_{\Delta_{2,k_0,k}}.
	\end{align*}
	Using ideas of the proof of Lemma C.4 in \cite{bhatt2019}, we have
	\begin{equation}
	\label{e:del-new-1}
	\Big|\Delta_{1,k_0,k}-\frac{k_0}{k}\Big|=\frac{k_0}{k}\Big|\underbrace{\frac{1}{2k_0}\sum_{i=1}^{k_0}\Big(\log \frac{\Gamma_{k_0+1}}{\Gamma_{i}}\Big)^2}_{\Delta_{k_0}}-1\Big|=o_{\mathbb{P}}(1).
	\end{equation}
	Furthermore
	\begin{equation}
	\label{e:del-new-2}
	\Big|\Delta_{2,k_0,k}-\frac{k_0}{k}\log \frac{k}{k_0}\Big|=\frac{k_0}{k}\log \frac{k}{k_0}\Big|\frac{\log (\Gamma_{k+1}/\Gamma_{k_0+1})}{\log (k/k_0)}\frac{1}{k_0}\sum_{i=1}^{k_0}\log \frac{\Gamma_{k_0+1}}{\Gamma_{i}}-1\Big|=o_{\mathbb{P}}(1)
	\end{equation}
	where $\log (\Gamma_{k+1}/\Gamma_{k_0+1})/\log (k/k_0) \stackrel{\mathbb{P}}{\longrightarrow}1$ follows from Lemma C.3 in \cite{bhatt2019}. 
	
	\noindent Using ideas similar to proof of Lemma C.4 in \cite{bhatt2019}. $(1/k_0)\sum_{i=1}^{k_0}\log(\Gamma_{k_0+1}/\Gamma_{i})=o_{\mathbb{P}}(1)$.
	
	\noindent In view of  \eqref{e:del-new-1} and \eqref{e:del-new-2}, the proof of Relation \eqref{e:del-new} follows.
\end{proof}

\end{document}